\documentclass{article}

\usepackage{arxiv}

\usepackage[utf8]{inputenc} 
\usepackage[T1]{fontenc}    
\usepackage{hyperref}       
\usepackage{url}            
\usepackage{booktabs}       
\usepackage{amsfonts}       
\usepackage{nicefrac}       
\usepackage{microtype}      
\usepackage{lipsum}		
\usepackage{graphicx}
\usepackage{natbib}
\usepackage{doi}

\usepackage{amsmath}
\usepackage{amsthm}
\usepackage{subcaption}

\newtheorem{theorem}{Theorem}  
\newtheorem{lemma}{Lemma}      
\newtheorem{proposition}{Proposition}
\newtheorem{corollary}{Corollary}
\newtheorem{definition}{Definition}

\newtheorem{remark}{Remark}     
\newtheorem{example}{Example}    

\title{Optimal Contest Design with Entry Restriction}

\date{} 					

\author{ Hanbing Liu\\
	Gaoling School of Artificial Intelligence\\
	Renmin University of China\\
	Beijing, China\\
	\texttt{liuhanbing@ruc.edu.cn} \\
	\And
	Ningyuan Li \\
	CFCS, School of Computer Science\\
	Peking University\\
	Beijing, China \\
	\texttt{liningyuan@pku.edu.cn}
    \AND
    Weian Li \\
	School of Software \\
	Shandong University \\
	Shandong, China \\
	\texttt{weian.li@sdu.edu.cn}
    \And
    Qi Qi\thanks{Corresponding author.} \\
	Gaoling School of Artificial Intelligence \\
	Renmin University of China\\
	Beijing, China \\
	\texttt{qi.qi@ruc.edu.cn} \\
    \And
    Changyuan Yu \\
	Baidu Inc. \\
	Beijing, China \\
	\texttt{yuchangyuan@baidu.com}
}

\begin{document}
\maketitle


\begin{abstract}
    This paper explores the design of contests involving \( n \) contestants, focusing on how the designer decides on the number of contestants allowed and the prize structure with a fixed budget. We characterize the unique symmetric Bayesian Nash equilibrium of contestants and find the optimal contests design for the maximum individual effort objective and the total effort objective.
\end{abstract}

\section{Introduction}
\label{sec:intro}



In real-world competitive scenarios, shortlists are commonly used. Specifically, organizers do not admit all registered candidates into the final contest but instead select a subset of candidates to form a shortlist, eliminating others. For example, in the Olympic Games, only top-ranked athletes qualify for the final competition due to scheduling constraints. Similarly, in the International Collegiate Programming Contest, the offline finals have limited venue capacity, so only teams meeting specific requirements advance. A similar scenario is observed in corporate recruitment, where only applicants who surpass a certain threshold in their resumes are invited for interviews.
This raises an important question: does a shortlist benefit the organizer, and how does it affect the competitive results?

To answer this question, we use the setting of a ``contest'' as a basic competitive (or game) model and examine the impact of a shortlist on contest design. Contests are a fundamental topic in mechanism design within economics, where a contest designer determines the prize structure and winning rules to achieve specific objectives. Contestants, in turn, invest irreversible and costly efforts to compete under these rules. Since contests effectively model competitive scenarios, they have widespread real-world applications, including sports, programming competitions, resource allocation and so on. A rank-order contest is a common contest format in which prizes are allocated strictly based on contestants' performance\footnote{The term “performance” is an abstract concept that varies by context. For example, in some cases, it may refer to measurable metrics such as race completion time in the Olympic 100-meter sprint or bid prices in an auction.}. Specifically, the contestant with the best performance wins the first prize, the second-best wins the second prize, and so on.


From a theoretical perspective, research on rank-order contests primarily focuses on characterizing contestant equilibrium in performance and optimizing contest design to achieve various objectives, such as the total effort \cite{MS01}, the maximum individual effort \cite{CHS19}, or other criteria \cite{AS09,EGG21}.
Broadly speaking, the existence of a shortlist can incentivize admitted contestants to exert more effort, as they anticipate facing stronger competition. However, eliminating too many contestants may reduce motivation, as fewer opponents make it easier to secure a desirable prize. Thus, determining the optimal shortlist size is a crucial problem.
To address this, several key questions must be considered sequentially. First, we need to understand how the shortlist size affects contestants' performance under any given prize structure. Once this relationship is established, the next step is to jointly determine the optimal shortlist size and prize structure to achieve various objectives. This is particularly challenging for the total effort objective, as there is a tradeoff between the number of admitted contestants and the incentive effect. Intuitively, reducing the number of contestants increases individual effort, but total effort may not follow a simple monotonic pattern. Finally, it is essential to analyze how the presence of a shortlist impacts the designer’s objective by comparing the optimal contests with and without a shortlist.

\subsection{Main Results}

\noindent \textbf{Model.} 
To address the aforementioned questions, we first construct a model of a contest with a shortlist. Consider a rank-order contest with $n \geq 2$ initially registered contestants. Each contestant’s ability, $x_i$, is private information drawn from a publicly known distribution $F(x)$. We assume that the contest designer is aware of all contestants' abilities upon registration. This assumption is reasonable, as designers can assess abilities based on submitted registration materials, such as competition records, CVs, or asset certificates.
The contest proceeds in two stages. First, the designer selects the number of shortlisted contestants, $2 \leq m\leq n$, and determines the prize structure  $\vec{V}=(V_1,V_2, \cdots, V_m) \in \mathbb{R}_{1\times m}$, specifying the rewards for each rank, within a budget $B$. The designer then ranks all $n$ contestants by ability in non-increasing order, admitting the top $m$ contestants to the next round while eliminating the remaining $n-m$ contestants. In the second stage, the shortlisted contestants observe the prize structure and recognize their eligibility. They then exert costly effort to compete for the prizes, with their effort levels influenced by their abilities. Each admitted contestant aims to maximize the utility (which is equal to the obtained prize minus the cost incurred by effort) by choosing the level of effort exerted.
The contest designer's decision variables consist of two components: the number of shortlisted contestants $m$ and the prize structure $\vec{V}$, to optimize two kinds of objectives, the maximum individual effort and the total effort.

Based on this contest-theoretical model, our main results are divided into three parts:
\subsubsection{Posterior distribution and equilibrium effort}
Our first result characterizes contestant equilibrium under any shortlist and prize structure. We use the symmetric Bayesian Nash equilibrium (sBNE) to define the contest’s solution, and can fully characterize the equilibrium effort function. 

\begin{theorem}[Unique sBNE of Contestants (Sketch)]
    For any ability distribution $F$, any size of shortlist $m$ and any prize structure $\vec{V}$, the unique symmetric Bayesian Nash equilibrium exists and can be expressed in a closed-form.
\end{theorem}

The sBNE is typically derived using the first-order condition of a contestant's expected utility function. A key component of this function is the probability of ranking in each position, which, in traditional rank-order contests, is determined by the prior distribution of abilities. However, in our model, the elimination stage alters contestants' beliefs, i.e., admitted contestants receive a signal of eligibility and update their beliefs about opponents from the prior to a posterior distribution.
Therefore, the main challenge in deriving this equilibrium is accurately computing an admitted contestant’s posterior belief about opponents (see Propositions~\ref{prop:posteriorBeliefs} and Corollary~\ref{prop:marginalBelief} for details). On the other hand, since contestants with different abilities form distinct posterior beliefs, calculating the sBNE in our model is significantly more complex due to the non-identical posterior distributions, compared to the traditional settings.

\subsubsection{Optimal Contest Design}
With a detailed characterization of contestant equilibrium, we next address the problem of the optimal contest design under two common objectives: the maximum individual effort and the total effort. In the traditional rank-order contest design problem, determining the optimal prize structure is already complex due to its high-dimension. In our work, we not only determine the optimal size of shortlist but also design the optimal prize structure within a fixed budget, making the problem more complex. However, we provide a fundamentally different characterization of the optimal contest, as detailed below.


Before presenting the detailed description of the optimal contests, we introduce a technical guideline to simplify our analysis. Specifically, we focus on a special type of contest, referred to as a ``simple contest'', where the designer allocates the budget equally across the prizes. Interestingly, we identify two sufficient conditions (see as Proposition \ref{prop:DesignGuideline}) under which the optimal contest is exactly a simple contest. These conditions allow us to reduce the prize structure's dimensionality from $m-1$ to one. First, when the objective is to maximize the effort of a certain rank, the optimal contest is always a simple contest, regardless of the ability distribution or cost function. Second, a simple contest is optimal for any linear combination of contestants' efforts only when the cost function is linear\footnote{A linear cost function is commonly used in both theoretical and practical contexts to characterize direct output in contests. For example, in an all-pay auction—often considered a type of contest—the payment follows a linear cost structure. In a political election setting, the cost can be interpreted as the investment in competition. Additionally, several studies \cite{MS01, AS09, DV09} have examined contests with linear cost functions.}.



\noindent \textbf{Maximum Individual Effort.}
Based on our equilibrium effort function, we determine that the optimal contest for maximizing the maximum individual effort admits only two contestants and allocates the entire budget to a single prize, i.e., a two-contestant, winner-take-all contest.

\begin{theorem}[Optimal Contest for the Maximum Individual Effort Objective]
    The optimal contest that maximizes the maximum individual effort is a two-contestant winner-take-all contest, where $m=2$ and $V_1=B$.
\end{theorem}

The intuition behind this result is twofold. First, for a given number of contestants, a larger first prize incentivizes greater effort. Second, for a fixed prize structure, admitting an additional contestant surprisingly reduces the expected maximum individual effort. This occurs because with more competitors, the expected prize for each contestant decreases, leading to lower effort to minimize costs.
Combining these insights, we conclude that a two-contestant, winner-take-all contest is optimal.


\noindent \textbf{Total effort.}
We examine the optimal contest design with a linear cost function to maximize total effort. 

\begin{theorem}[Optimal Contest for the Total Effort Objective]
    For any ability distribution $F$, the optimal contest for maximizing the total effort can be described as:
    \begin{enumerate}
        \item The number of prize $l^*$ is equal to the shortlist size $m^*$ minus one, i.e., $l^*=m^*(n)-1$.
        \item The budget is equally divided into these prizes, i.e, $V_1=V_2=\cdots= V_{l^*}=B/l^*$.
        \item The optimal shortlist is proportional to $n$, i.e., $m^*(n)=kn$, where $k$ is the solution to 
        \[
            \int_k^1 F^{-1}(1-q)(\frac{1}{q}-(2k-k^2)\frac{1}{q^2})dq=0.
        \]
    \end{enumerate}
\end{theorem}

As stated in the guidelines, when the cost function is linear, the optimal contest for any linear combination of contestants' efforts is a simple contest, particularly for the total effort objective. This reduces the design process to determining only the number of prizes and the shortlist size.
Furthermore, we find that the optimal number of prizes equals the shortlist size minus one. The reasoning is twofold. First, a single zero-prize is necessary to incentivize higher effort (see Corollary \ref{coro:Consolation}), as without it, some contestants may exert no effort while still receiving a non-negative prize. Second, having more than one zero-prize is unnecessary (see Corollary \ref{coro:EmptyPrize}). If the prize structure is fixed, admitting an additional contestant increases competition but lowers the probability of any contestant winning a prize. Since effort is costly, contestants will adjust by exerting less effort to balance the reduced expected rewards.
Ultimately, the optimal contest is a complete simple contest. This further simplifies the designer's decision-making from two variables to one, requiring only the selection of the optimal shortlist size.


However, the optimal shortlist size heavily depends on the ability distribution. Due to the tradeoff between shortlist size and exerted effort, total effort does not vary monotonically with this size. To derive this result, we first express the total effort objective using a Beta distribution (see Lemma \ref{lem:betaRepTotalEffort}). This representation allows us to obtain an asymptotic form of the objective, enabling the analysis of the optimal shortlist size. In other words, we can design an efficient algorithm to find the optimal shortlist size for any ability distribution. 

In general, we also prove a uniformly upper bound on the optimal size, that is, for any ability distribution, the optimal size does not exceed $0.3162n$, as $n$ goes to infinity and this upper bound is tight achieved by a specific distribution.
\begin{theorem}[Tight Upper Bound for the Optimal Shortlist Size]
    For an arbitrary ability distribution $F$, when $n\rightarrow +\infty$, the optimal shortlist size that maximizes ex-ante total effort, denoted by $m^{*}$, has the following linear upper bound with respect to $n$:
    $$ \lim_{n \rightarrow \infty} \frac{m^*(n)}{n} \leq \bar{k},$$
    where $\bar{k} \approx 31.62\%$ is the solution to the equation $\ln k=(2-k)(k-1)$.

    Moreover, there exists a distribution such that $m^*(n)/n=\bar{k}$, meaning the bound is tight. 
\end{theorem}
    


\subsubsection{Comparison with the optimal contest without shortlist.}
In addition to the equilibrium analysis and the optimal contest design for different objectives, we also compare the performance of the optimal contest with a shortlist to that without one 

In previous results, when the cost function is linear, the optimal contest without a shortlist is a winner-take-all contest for both the maximum effort objective \cite{CHS19} and the total effort objective \cite{MS01}. Denote by $\mathcal{C}^{(1,n)}$ the $n$-contestant winner-take-all contest. For the maximum individual effort objective, we have shown that the optimal contest with a shortlist is a two-player winner-take-all contest, denoted as $\mathcal{C}^{(1,2)}$. For the total effort objective, although we know the optimal contest is a complete simple contest, the optimal size cannot be uniformly determined and highly depends on the ability distribution. For any ability distribution $F$, let $\mathcal{C}^{(m^*(F,n)-1, m^*(F,n))}$ denote the optimal complete simple contest with a shortlist. 

To compare these contests, we analyze the outputs across three types of contests for any ability distribution: (i) the optimal contest without a shortlist, $\mathcal{C}^{(1,n)}$, achieves $\Theta(1)$ of the maximum individual and total effort; (ii) a two-contestant winner-take-all contest, $\mathcal{C}^{(1,2)}$,  yields $\Theta(\log n)$ of the maximum individual and total effort;  and (iii) the optimal contest with a shortlist, $\mathcal{C}^{(m^*(F,n)-1, m^*(F,n))}$, reaches $\Theta(n)$ in total effort. These effort bounds not only establish the theorem but also highlight the impact of shortlist on the total effort.


For the maximum individual effort objective, we define the ratio of maximum individual effort as the gap between $\mathcal{C}^{(1,2)}$ and $\mathcal{C}^{(1,n)}$, and show that this ratio is $\Theta(\log n)$ for any distribution.

\begin{theorem}
    For any ability distribution $F$, under the maximum individual effort objective, $\mathcal{C}^{(1,2)}$ results in $\Theta(\log n)$ times the maximum individual effort of $\mathcal{C}^{(1,n)}$. Specifically,  
    $$
        \frac{\text{ME}(\mathcal{C}^{(1,2)},F)}{\text{ME}(\mathcal{C}^{(1,n)},F)}
        = \Theta(\log n),
    $$
    where $\text{ME}(\cdot)$ represents the maximum individual effort.
\end{theorem}


For the total effort objective, since the optimal contests may vary under different ability distributions, we first fix a specific distribution and then compare the optimal contests under it. Our findings can be summarized as follows.

\begin{theorem}
    Fixed any ability distribution $F$, under the total effort objective, the optimal contest $\mathcal{C}^{(m^*(F,n)-1, m^*(F,n))},F)$ can achieve $\Theta(n)$ times the total effort compared to $\mathcal{C}^{(1,n)}$. Specifically,
    $$
        \frac{\text{TE}(\mathcal{C}^{(m^*(F,n)-1, m^*(F,n))},F)}{\text{TE}(\mathcal{C}^{(1,n)},F)}
         = \Theta(n),
    $$
    where $\text{TE}(\cdot)$ is the total effort and  $m^*(F,n)$ is the optimal shortlist size for the ability distribution $F$. 
\end{theorem}

\subsection{Related Literature}

\subsubsection{Further related works}
Our paper belongs to the field of single contest design, specifically focusing on rank-order contests (also known as all-pay contests). Most papers in this field concentrate on characterizing contestant equilibrium and designing optimal prize structures for various objectives.
In the early period, \cite{GH88} identify the optimal contests for both identical and non-identical contestants based on their abilities, aiming to maximize total effort. \cite{MS01} extend this work by considering cases where contestants' abilities are drawn from a publicly known distribution, and they design the optimal contests for linear, concave, and convex cost functions.
\cite{KG03} study procurement contests in both symmetric and asymmetric settings and propose the corresponding optimal contest designs. \cite{AS09} analyze crowdsourcing contests with a large number of participants, characterizing the asymptotically optimal prize structure to maximize the effort of the top contestants.
\cite{GK16} introduce contests with simple contestants who can strategically choose whether to participate, designing the optimal prize structure to maximize the total ability of participants. \cite{LLWZ18} and \cite{LL23} explore the optimal rank-order contest designs allowing negative prizes.
\cite{CHS19} propose an optimal crowdsourcing contest catering for the maximum effort objective, and \cite{EGG21} study the designer's threshold objective, deriving the optimal contest structures. \cite{G23} investigate the impact of prize structure and ability distributions on contestants' equilibrium efforts. 
\cite{SSYJ24} consider incorporating entry restriction into contests, but do not find the contestant equilibrium and the optimal contest design. 
Additionally, several studies \cite{BKV96, BK98, BKV12} analyze the equilibria of all-pay auctions.
Our work considers a rank-order contest with a shortlist, introducing an elimination stage before effort is exerted, and investigates the optimal contest design for maximizing both the maximum individual effort and total effort, which confirms the positive effect of shortlist for contest designers.

In addition, due to the presence of a shortlist, our paper is related to the concepts of elimination and signaling in contests. The seminal paper by \cite{MS06} divides contestants into two sub-contests, where the winners of each sub-contest compete again in the final contest. \cite{FL12} investigate the optimal design of multi-stage Tullock contests \cite{T08}. \cite{LMZ18} introduce information disclosure policies in all-pay contests and compare different disclosure strategies. \cite{LSA19} summarize the intersection between contest design and information disclosure. \cite{MPS21} study all-pay sequential elimination contests in a complete information setting. \cite{FW22} incorporate information disclosure and bias in a two-stage Tullock contest. \cite{R24} examines a multi-stage all-pay contest, where each stage eliminates the contestant with the lowest effort, and analyzes contestants' strategies at each stage.
Additionally, some studies \cite{CKZ17, C24, KZZ24} apply the framework of Bayesian persuasion to contests and investigate optimal strategies for information disclosure.
In contrast to these papers, our contest model introduces an elimination stage via a shortlist. After elimination, contestants receive a signal indicating whether they are admitted, which influences their beliefs about the ability distributions of their opponents. Moreover, while several papers aim to improve contest performance through an elimination stage, the methods discussed in those works tend to be quite sophisticated. In contrast, our model is much simpler—requiring only the addition of a shortlist initially—yet it achieves significantly better results in terms of both maximum individual effort and total effort compared to traditional optimal contests.

\section{Model and Preliminaries}
\label{sec:pre}

In this section, we formally introduce the contest model with a shortlist. 
Consider a single-contest setting with $n$ registered contestants. Each contestant $i \in [n]=\{1,2,\cdots, n\}$ has a private bounded ability level, denoted by $0<x_i<+\infty$, which is independently and identically drawn from a publicly known distribution  $F(\cdot)$ with a continuous probability density function $f(\cdot)>0$. For convenience, we also represent ability using the quantile $q_i := 1-F(x_i)$ which follows a uniform distribution on $[0, 1]$, i.e., $U[0,1]$. The inverse function, $v(q_i):=F^{-1}(1-q_i)=x_i$,  is strictly decreasing, meaning that a lower quantile corresponds to a higher ability level.

As mentioned in Section \ref{sec:intro}, we consider a contest with a shortlist. Initially, the contest designer determines the size of the shortlist, i.e., the number of admitted contestants, $2\leq m \leq n$, and establishes the prize structure $\vec{V}=(V_1, V_2, \cdots, V_m)$ within the budget $B$, where  $V_1 \geq V_2\geq  \cdots \geq V_m$ and $\sum_{i=1}^p V_i \leq B$. Denote by $\mathcal{C}=(m,\vec{V})$ the designer's decision variable, i.e., the contest configuration. Then, the designer then discerns the abilities of the registered contestants\footnote{In practice, ``ability'' can be reflected in various ways, such as applicants' resumes, athletes' historical scores, or other relevant data, prior to the contest. We assume the designer is able to assess and discern the contestants' abilities.},
selecting the top $m$ contestants based on ability into the shortlist, while eliminating the others. Note that the designer determines the shortlist size and prize structure before knowing the exact abilities of the contestants. In other words, these decisions are based solely on the number of registered contestants $n$ and the prior ability distribution $F$. 

Next, each admitted contestant $i$, i.e., whose ability $x_i$ ranks in the top $m$, is informed and then strategically chooses an effort level $e_i$ to compete for prizes. The cost of exerting effort is given by $g(e_i)/x_i$, where $g(\cdot): \mathbb{R}_{\geq 0} \rightarrow \mathbb{R}_{\geq 0}$ is a strictly increasing, continuous and differentiable function with $g(0)=0$. Intuitively, for the same effort level $e$, a contestant with higher ability incurs a lower cost than one with lower ability. After effort levels are chosen, the contest designer allocates prizes based on a rank-order rule: the contestant with the highest effort receives the first prize $V_1$, the second-highest effort earns $V_2$, and so on. 

In summary, our model follows these sequential steps:
\begin{enumerate}
    \item The contest designer determines the number of admitted contestants, $m$, and sets the prize structure $\vec{V}$ within the budget $B$.
    \item The top $m$ contestants based on ability are selected, while the others are eliminated.
    \item The remaining $m$ contestants exert effort to compete for prizes.
    \item Prizes are awarded based on effort, in non-descending order.
\end{enumerate}

Given our contest model, we now define a contestant's utility. Since eliminated contestants exert no effort and receive no prize, their utility is zero.
Let $\mathbf{e}= (e_1, e_2,\cdots, e_m)$ denote the effort profile of all admitted contestants. For an admitted contestant $i$, the utility is given by the prize earned minus the incurred cost:
$$
    u_i(\mathbf{e}) = V_{Rank(i,\mathbf{e})}-\frac{g(e_i)}{x_i},
$$
where $Rank(i,\mathbf{e})$ represents contestant $i$'s rank based on the effort profile $e$. Each contestant aims to maximize utility by choosing an effort level.

For an admitted contestant $i$, the exerted effort $e_i$ depends only on her ability $x_i$ and her strategy function $b_i(x): \mathbb{R} \rightarrow \mathbb{R}$, i.e., $e_i=b_i(x_i)$. We assume that $b_i(x)$ is monotone non-decreasing, meaning that a contestant with higher ability will not exert lower effort.
In this paper, suppose all admitted contestants follow the same strategy function, i.e., the strategy is symmetric: $b_i(x)=b_j(x)$ for any $x>0$ and any two admitted contestants $i$ and $j$. Since our model is under the incomplete information setting, we adopt the symmetric Bayesian Nash equilibrium as the solution concept. 
\begin{definition}\label{def:sBNE}
    An effort function constitutes a symmetric Bayesian Nash equilibrium (sBNE) if and only if, for any admitted contestant, choosing this function maximizes expected utility when all others do the same. Specifically, an effort function $b^*(x)$ is an sBNE if and only if, for any admitted contestant $i$,
$$
    b^*(x_i) \in \arg\max_{e_i} \sum_{j=1}^m V_j \cdot \text{Pr}_{ij}(e_i) -\frac{g(e_i)}{x_i},
$$
    where $\text{Pr}_{ij}(e_i)$ represents the probability that contestant $i$ is ranked at $j$ in the contest, given that all other admitted contestants also adopt $b^*(x)$. 
\end{definition}

Lastly, the contest designer aims to optimize an objective function by establishing the contest configuration $\mathcal{C}=(m, \vec{w})$, including the shortlist size and prize structure. In this work, we focus on two objective functions: maximum individual effort and total effort. Specifically, given a contest configuration $\mathcal{C}$, the maximum individual effort objective is defined as
$$
    \text{ME}(\mathcal{C},F)= \mathbb{E}_{x_1, x_2,\cdots, x_n \sim F}\big[e_{(1)}\big],
$$
where $e_{(1)}$ is the highest effort among all contestants under the sBNE. The total effort objective is defined as
$$
    \text{TE}(\mathcal{C},F)= \mathbb{E}_{x_1, x_2,\cdots, x_n \sim F}\big[\sum_{i=1}^n e_{i}\big],
$$
where $e_i$ is contestant $i$'s exerted effort under the sBNE, assuming that eliminated contestants exert zero effort. 


\section{Contestant Equilibrium}
\label{sec:playerSBNE}
In this section, we fully characterize the unique symmetric Bayesian Nash equilibrium of admitted contestants for any shortlist size $m$ and any prize structure $\vec{V}$. A key step in achieving this characterization is to explicitly represent each contestant's utility, allowing us to derive and solve the first-order condition.
Unlike traditional contest settings without a shortlist, where the probability of contestant $i$ is ranked at $j$ can be directly computed using the prior ability distribution, our setting requires calculating these probabilities based on posterior beliefs about the abilities of other admitted contestants. This necessitates a detailed understanding of the posterior ability distribution. Therefore, we first describe the closed form of posterior beliefs in Subsection \ref{subsec:Posterior Beliefs}. Then, in Subsection \ref{subsec:Equilibrium Efforts}, we provide a complete characterization of contestant equilibrium efforts.

\subsection{Posterior Beliefs}
\label{subsec:Posterior Beliefs}

Given a contest with a shortlist size of $m$ and a prize structure $\vec{V}$, we first update an admitted contestant's posterior belief about the ability levels of other admitted contestants. This posterior belief, represented as the joint posterior distribution, is derived using Bayes' rule, as stated in the following proposition.


\begin{proposition}[Posterior Beliefs]\label{prop:posteriorBeliefs}
    For any admitted contestant (w.l.o.g., labeled as Contestant 1), the joint posterior probability density function of the ability levels of the other admitted contestants is given by:
    \[
        \beta_1(\mathbf{x}) =   
        \begin{cases} 
        \frac{\binom{n-1}{m-1}F^{n-m}(x^{(1)})\prod_{i=2}^{m}f(x_i)}{J(F,n,m,x_1)} & \text{if } x^{(1)} \leq x_1, \\
        \frac{\binom{n-1}{m-1}F^{n-m}(x_1)\prod_{i=2}^{m}f(x_i)}{J(F,n,m,x_1)} & \text{if } x^{(1)} > x_1,
        \end{cases}
    \]
    where $x_1$ is the contestant 1's ability level and $x^{(1)}:=\min_{j\in [m]\setminus \{1\}}x_i$ is the lowest ability level of other admitted contestants. The normalization denominator is defined as $J(F(\cdot),n,m,x) := \binom{n-1}{m-1}F^{n-m}(x)(1-F(x))^{m-1}+\binom{n-1}{m-1}(m-1)\int_0^{F(x)}t^{n-m}(1-t)^{m-2} \, dt$.
\end{proposition}

We further marginalize this belief to obtain the posterior ability distribution of a specific admitted contestant.


\begin{corollary}[Marginal Posterior Beliefs]\label{prop:marginalBelief}
For any admitted contestant (labeled as Contestant 1), the posterior probability density function of another admitted contestant's ability (labeled as Contestant 2) is:

\[
\beta_1(z) =   
\begin{cases} 
\frac{\binom{n-1}{m-1}\left [ F^{n-m}(z)(1-F(z))^{m-2}+(m-2)\int_0^{F(z)}t^{n-m}(1-t)^{m-3}\, dt \right ]f(z)}{J(F,n,m,x_1)} & \text{if } z \leq x_1, \\
\frac{\binom{n-1}{m-1}\left [  F^{n-m}(x_1)(1-F(x_1))^{m-2} +(m-2)\int_0^{F(x_1)} t^{n-m}(1-t)^{m-3}\, dt \right ] f(z)}{J(F,n,m,x_1)} & \text{if } z > x_1.
\end{cases}
\]

The posterior cumulative distribution function is then expressed as:
\[
\Pr_{\beta_1}(X_2\leq z) =   
\begin{cases} 
\frac{\binom{n-1}{m-1} \left [ B_{F(z)}(n-m+1,m-1) + (m-2)\int_{0}^{z} B_{F(t)}(n-m+1, m-2) \, dF(t) \right ]}{J(F,n,m,x_1)} & \text{if } z \leq x_1, \\
\Pr_{\beta_1}(X_2\leq x_1) + \\ \quad\frac{\binom{n-1}{m-1}(F(z)-F(x_1))\left [ F^{n-m}(x_1)(1-F(x_1))^{m-2}+(m-2)B_{F(x_1)}(n-m+1, m-2)\right ]}{J(F,n,m,x_1)} & \text{if } z > x_1.
\end{cases}
\]
where $B_{x}(a,b):=\int_{0}^{x} t^{a-1}(1-t)^{b-1} \, dt$ denotes an incomplete beta function.

Note that both $\beta_1(\cdot)$ and $\Pr_{\beta_1}(\cdot)$ are continuous.
\end{corollary}

Using the marginal posterior belief, we derive two types of stochastic dominance to illustrate the effect of a shortlist.

\begin{proposition}[Stochastic Dominance of Posterior over Prior]\label{prop:stoDomPos}
    Based on the prior p.d.f. \(f(x)\) and posterior p.d.f. \(\beta_i(x)\), let \(\Pr_{f}(X_j \leq z)\) and \(\Pr_{\beta_i}(X_j \leq z)\) denote the probabilities that contestant \(i\) believes that the ability level of contestant \(j \neq i\) is at most $z$ before and after receiving the admitted signal, respectively. Then, for all $z$, we have:
    \[
    \Pr_{\beta_i}(X_j \leq z) \leq \Pr_{f}(X_j \leq z).
    \]
\end{proposition} 

\begin{proposition}[Posterior of Higher Ability Stochastically Dominates Lower Ability]\label{prop:StoDomAbi}
    For two admitted contestants, say $i$ and $j$, if contestant $i$ has a higher ability than contestant $j$ (i.e., $x_i > x_j$), then contestant $i$'s posterior belief about the ability of any other admitted contestant  $k$ first-order stochastically dominates that of contestant $j$. Formally, for all $z\geq 0$, we have:
    \[
    \Pr_{\beta_i}(X_k \leq z) \leq \Pr_{\beta_j}(X_k \leq z).
    \]
\end{proposition}

A direct implication of stochastic dominance is higher expectation\footnote{Formally, if random variables $X$ an $Y$ satisfy $\Pr(X\leq z) \leq \Pr(Y\leq z)$ for all $z$, i.e., $Y$ is first-order stochastically dominated by $X$, then $\mathbb{E}[X] \geq \mathbb{E}[Y]$ holds.}. Therefore, the above two propositions imply that after the shortlist is applied, each admitted contestant perceives their opponents as stronger. Moreover, the stronger a contestant is, the more competitive they perceive their opponents to be.


\begin{example}
    Figure~\ref{fig:beliefs} shows the belief change of contestant $x_1$ resulting from shortlist. We rewrite the posterior belief as $\beta_1(z)=q_{x_1}(z)f(z)$, then the factor $q_{x_1}(z)$ contains the information brings by the admission signal. When a contestant gets admitted, her belief for low ability is discounted, and the factor is also smaller for lower ability, then the belief for ability higher than hers is adjusted by a constant factor accordingly, as show in Figure~\ref{fig:beliefs-a},~\ref{fig:beliefs-b}. The cumulative probability functions are plotted in Figure~\ref{fig:beliefs-c}. An opponent is viewed as stronger after shortlist (Proposition~\ref{prop:stoDomPos}), and stronger contestant also sees her opponent as stronger (Proposition~\ref{prop:StoDomAbi}).
        \begin{figure}[h]
        \centering
    \begin{subfigure}[ht]{0.30\textwidth}
        \centering
        \includegraphics[width=\textwidth]{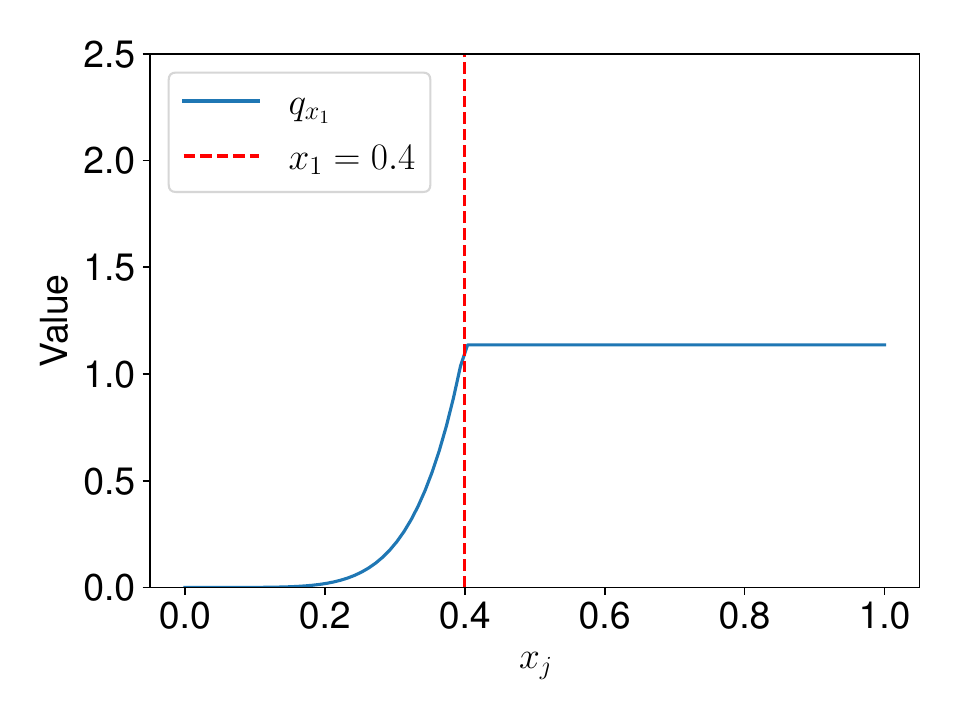}
        \subcaption{Factor $q_{x_i}$}
        \label{fig:beliefs-a}
        \end{subfigure}
    \begin{subfigure}[ht]{0.30\textwidth}
        \centering
        \includegraphics[width=\textwidth]{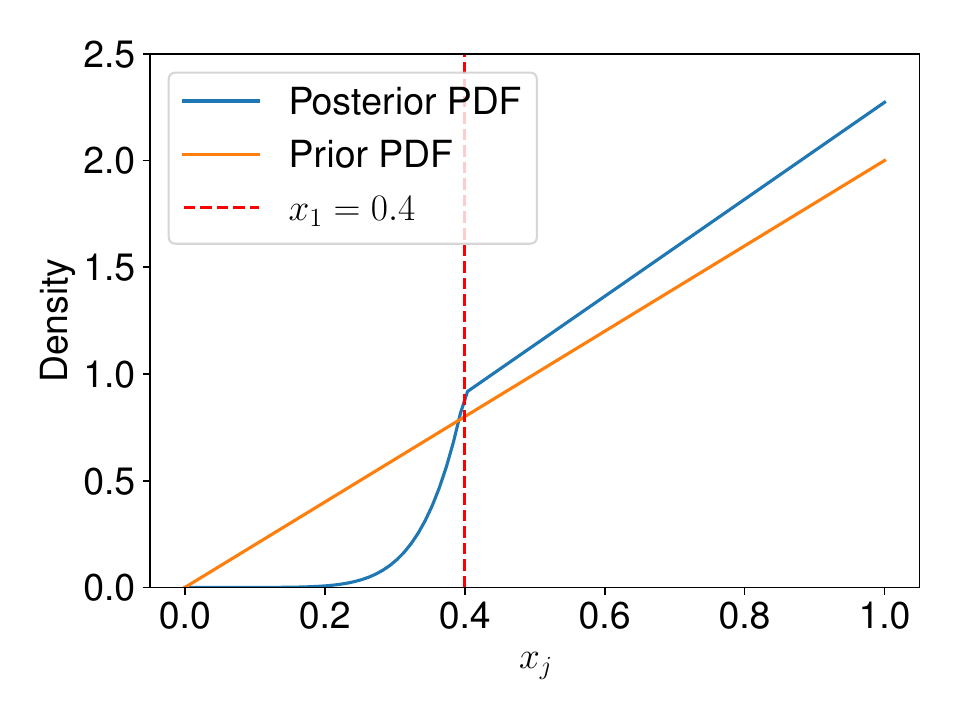}
        \subcaption{PDF}
        \label{fig:beliefs-b}
        \end{subfigure}
    \begin{subfigure}[ht]{0.30\textwidth}
        \centering
        \includegraphics[width=\textwidth]{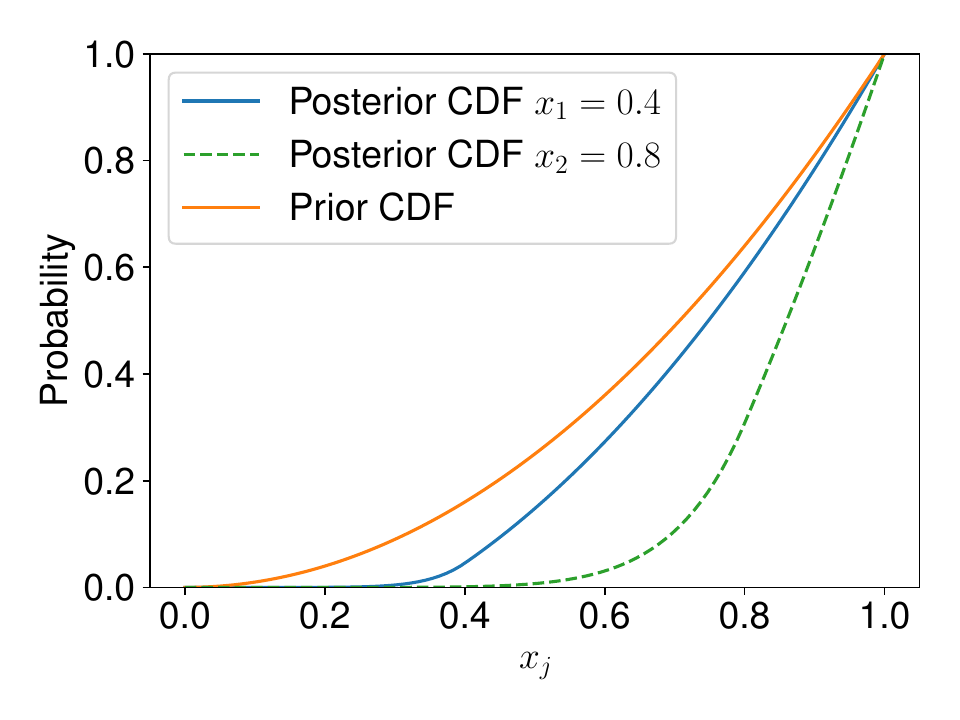}
        \subcaption{CDF}
        \label{fig:beliefs-c}
        \end{subfigure}
    \caption{Posterior beliefs of player $x_1$ ($n = 5,m=2,F(x)=x^2$).}
    \label{fig:beliefs}
    \end{figure}
\end{example}
Although each contestant believes that every opponent becomes stronger after the shortlist, a counterintuitive result emerges: if we consider all opponents together, the contest environment actually becomes less competitive. This is because the number of opponents is the primary factor influencing competitiveness, and its reduction makes the contestant believe that her probability of achieving a certain rank is higher. Formally, we present the following proposition:

\begin{proposition}[Threatens of Opponents Decrease after Shortlist]\label{prop:ThreatenDesc}
For any contestant (labeled as 1) and any $l < m$, she perceives the opponent with the $l^{\text{th}}$ highest ability as weaker after the shortlist, in comparison to herself:
\[
\Pr_{\beta_1}(X_{(l)} > x_1) \leq \Pr_{f}(X_{(l)}>x_1),
\]
where $X_{(l)}$ is the $l^{\text{th}}$ highest ability among her opponents, $[m]\backslash\{1\}$. 

Moreover, the fewer opponents that are admitted, the weaker they are perceived by the this contestant compared to herself. In other words, $\Pr_{\beta_1}(X_{(l)} > x_1)$ increase with $m$. 
\end{proposition}

It is important to note that, for a given contestant, although the marginal posterior belief is identical for all her opponents, they are not independent. The intuition is straightforward: knowing that a contestant weaker than herself has been admitted to the game will lower her expectations for another admitted contestant. Therefore, in general, $\beta_1(\mathbf{x}) \neq \prod_{i=2}^{m} \beta_1(x_i)$, and the joint probability distribution is required to fully characterize the posterior beliefs.


\subsection{Equilibrium Efforts}
\label{subsec:Equilibrium Efforts}
With the detailed characterization of posterior beliefs in place, we can now calculate the symmetric Bayesian Nash equilibrium. To facilitate the analysis, we first introduce some symbols. Let \( b: x_i \mapsto e_i \) be a strictly increasing strategy function, and let its inverse function be \( \gamma(\cdot) \). For simplicity, we use \( \gamma_i \) For simplicity, we use \( \gamma(e_i) \), which represents the corresponding ability \( x_i \) for an effort  \( e_i \) under the function $b(\cdot)$. Thus, the decision problem for contestant \(i\) to exert effort  \(e_i\) is equivalent to reporting an ability \(\gamma_i\).

Now, we can calculate the probabilities that the admitted contestant $i$ exerts $e_i$ (which may not necessarily be $b(x_i)$) and is ranked at position $l$, based on her posterior beliefs.


\begin{proposition}[Subjective Probability]\label{prop:winProb}
When admitted contestant $i$ exerts effort $e_i$ given her ability $x_i$, her perceived probability of obtaining rank $l$, denoted as $P_{(i,l)}(\gamma_i \mid x_i)$, is:

\begin{enumerate}
    \item When $\gamma_i \leq x_i$,
\[
P_{(i,l)} =   
\begin{cases} 
\frac{\binom{n-1}{l-1}F(\gamma_i)^{n-l}(1-F(\gamma_i))^{l-1}}{J(F,n,m,x_i)} & \text{if } l < m, \\
\frac{\binom{n-1}{m-1}\left [ F^{n-m}(x_i)(1-F(x_i))^{m-1}+(m-1)\int_{F(\gamma_i)}^{F(x_i)}t^{n-m}(1-t)^{m-2} \, dt \right ]}{J(F,n,m,x_i)} & \text{if } l = m.
\end{cases}
\]
    \item When $\gamma_i > x_i$,
\[
P_{(i,l)} =   
\begin{cases} 
\frac{\binom{n-1}{m-1}\binom{m-1}{l-1}(1-F(\gamma_i))^{l-1}\left [F^{n-m}(x_i)(F(\gamma_i)-F(x_i))^{m-l}+(m-l)\int_{0}^{F(x_i)}t^{n-m}(F(\gamma_i)-t)^{m-l-1}\, dt\right]}{J(F,n,m,x_i)} & \text{if } l < m, \\
\frac{\binom{n-1}{m-1}F^{n-m}(x_i)(1-F(\gamma_i))^{m-1}}{J(F,n,m,x_i)} & \text{if } l = m,
\end{cases}
\]
\end{enumerate} 
where $P_{(i,l)}$ is $P_{(i,l)}(\gamma_i \mid x_i)$ in short and $\gamma_i = b^{-1}(e_i)$ and $P_{(i,l)}(\gamma_i \mid x_i)$ is a continuous function of $\gamma_i$. 
\end{proposition}
\begin{remark}\label{rmk:subjectiveProb}
    Under the sBNE (i.e., when the observer truthfully reports her ability level, $\gamma_i=x_i$), her subjective probability of obtaining rank $l<m$ is given by $\frac{\binom{n-1}{l-1}F(x_i)^{n-l}(1-F(x_i))^{l-1}}{\sum_{j=1}^{m}\binom{n-1}{j-1}F^{n-j}(x_i)(1-F(x_i))^{j-1}}$ (since $J(F, n, m, x_i) = \sum_{j=1}^{m}\binom{n-1}{j-1}F^{n-j}(x_i)(1-F(x_i))^{j-1}$, as shown in Lemma~\ref{lem:normal}), it is exactly the prior probability of achieving rank $l$, given that the contestant is admitted (i.e., ranks among the top $m$). Notably, this subjective probability decreases as the shortlist size $m$ increases. 

    This further highlights the signaling effect of shortlisting. As the shortlist size decreases, admission becomes more difficult, leading admitted contestants to perceive themselves as relatively more competitive (as established in Proposition~\ref{prop:ThreatenDesc}). Consequently, they believe they have a higher probability of securing a better rank.
    
    
\end{remark}

Having detailed probabilities for each rank, we can now express the expected utility for each admitted contestant. Next, we move to derive the symmetric strategy function under equilibrium, thereby characterizing the sBNE of contestants. First, we integrate the first-order condition of the utility function to establish a necessary condition for the equilibrium strategy function $b(\cdot)$. Then, by verifying its non-negativity and monotonicity, we confirm the validity of the derived symmetric strategy. Consequently, we conclude that the sBNE exists and is unique, given as follows.


\begin{theorem}[Unique sBNE of Contestants]\label{thm:contestantSBNE}
    For any ability distribution $F$, any size of shortlist $m$ and any prize structure $\vec{V}$, the unique symmetric Bayesian Nash equilibrium exists and can be expresses as:
    \[
        b^*(x) = g^{-1}\left(\int_{0}^{x}\frac{\sum_{l=1}^{m-1}\binom{n-1}{l-1}(n-l)(V_l-V_{l+1})F^{n-l-1}(t)(1-F(t))^{l-1}f(t)}{J(F,n,m,t)} t\, dt \right), 
    \]
    where $g(\cdot)$ is the cost function.
\end{theorem}
\begin{remark}\label{rmk:PrizeGap}
    From the expression of equilibrium effort, it becomes clear that contestants are incentivized by the gap between consecutive prizes, $(V_{l}-V_{l+1})$,  rather than the absolute value of the prizes themselves. Furthermore, if the cost function $g(\cdot)$ is linear, the contributions of these gaps to effort exertion remain independent of one another.
\end{remark}

\section{Optimal Contest Design}
\label{sec:optimal design}

In this section, we discuss the optimal contest design under different objectives. In Subsection~\ref{subsec:GeneralGuideline}, we identify a guideline (i.e., two sufficient conditions) such that the optimal contest is a form of simple contest.  In Subsection~\ref{subsec:HighestEffort}, we focus on the maximum individual effort objective, showing that the two-contestant winner-take-all contest is optimal. In Subsection~\ref{subsec:TotalEffort}, we analyze the total effort objective in detail and find that: (a) the optimal contest is a complete simple contest, i.e, the number of prizes is exactly the shortlist size minus one; (b) the optimal shortlist size for a given distribution grows asymptotically linearly with $n$; (c) the optimal shortlist size is no more than $31.62\%\,n$ for any distribution asymptotically.


\subsection{General Guideline}\label{subsec:GeneralGuideline}

For a given set of contestants $[n]$ and an ability distribution $F$, the designer aims to determine the shortlist size $2\leq m \leq n$ and and a prize structure $\vec{V}=\{V_1,\ldots,V_m\}$ that satisfies the rank-order property $V_1\geq V_2 \geq\ldots\geq V_m \geq 0$ and the budget constraint $\sum_{l=1}^mV_l \leq B$, with the goal of maximizing the ex-ante effort-based objective under equilibrium. Formally, 

\[
    \begin{aligned}
        \mathop{\arg \max}_{m,\vec{V}} \quad & \mathbb{E}_{\vec{x} \sim F^n} [ \text{OBJ}(b(x_{(1)}), \ldots,b(x_{(m)}))] \\
        s.t. \quad & \sum_{l=1}^{m} V_l \leq B \\
            & V_l  \geq V_{l+1} \geq 0,
    \end{aligned}
\]
where $x_{(i)}$ denotes the $i^\text{th}$ highest realized ability level, and $\text{OBJ}(\cdot)$ represents the component-wise non-decreasing objective function of the efforts, arranged in decreasing order. 

To solve this problem, we first establish two key properties of the equilibrium effort from the designer's perspective, which will help reduce the space of decision variables.


\begin{corollary}[No Consolation Prize Should be Set]\label{coro:Consolation}
    For any contest that allocates a non-zero prize for the last place, i.e., \(V_m > 0\), setting \( V_m = 0\) increases the equilibrium effort of every contestant.
\end{corollary}

Recall from Remark~\ref{rmk:PrizeGap} that the gap between consecutive prizes motivates contestants. Since a consolation prize makes prizes for higher ranks less attractive, it negatively impacts the effort objectives.

\begin{corollary}[Empty Prizes Discourage contestants]\label{coro:EmptyPrize}
For a given list of contestant ability $\vec{x}$ nd a prize structure with $k$ non-zero prize, i.e., \( V_1 \geq V_2 \geq ... \geq V_k > V_{k+1} = 0 \), extending the shortlist from $k+1$ to $m$ contestants decreases the equilibrium efforts of all contestants in the former shortlist.
\end{corollary}

From an equilibrium perspective, a contestant perceives herself as less competitive as the contest becomes more crowded (Proposition~\ref{prop:ThreatenDesc}). Consequently, her subjective probability of winning any non-empty prizes decreases when more contestants are admitted (Remark~\ref{rmk:subjectiveProb}). This results in a lower expected payoff, which discourages contestants from exerting more effort.

Before further discussion, we extend the definition of a Simple Contest \cite{EGG21} to the shortlist setting, a special type of contest that is crucial for effort-based objectives.

\begin{definition}[Simple Contest]
    A contest with $n$ participants and a shortlist size $m$ is called a simple contest if all of its non-zero prizes are equal. Furthermore, if all the prizes in the contest are equal, meaning either $V_i = 0$ for all $i$ or $V_1 = \ldots = V_m \neq 0$, then the contest is referred to as a trivial simple contest.
\end{definition}

Note that, in a trivial simple contest, every contestant will exert zero effort in the contest. Additionally, winner-take-all contest, i.e., $V_1=B,$, is a special case of non-trivial simple contest.

In fact, in contest settings where the designer has the option to create a shortlist, the optimal contest design remains a simple contest for the linear combination of efforts objective, as long as the cost function is linear or only the effort of one certain rank is concerned. 


\begin{proposition}[General Design Guideline]\label{prop:DesignGuideline}
    If the designer's objective is a non-negative, non-zero linear combination of the contestants' effort under equilibrium, denoted by $\vec{e}=(b(x_{(1)}), \ldots,b(x_{(m)}))$ (with contestants re-indexed according to their rankings), either ex-ante (in expectation) or ex-post (for a specific ability level profile of registered contestants), i.e., $u_d=\mathbb{E}_{\vec{x}}[\vec{c}\cdot e(\vec{x})]$ or $u_d(\vec{x})=\vec{c}\cdot e(\vec{x})$, then the optimal contest that maximizes the designer's objective will be a non-trivial simple contest if any of the following conditions are satisfied:
    \begin{enumerate}
        \item The designer only cares about the effort of exactly one ranking, i.e., $c_i>0$ and $\forall \,l\neq i$, $c_l=0$.
        \item The cost function is linear, i.e., $g(e)=ke$ for some $k>0$.
    \end{enumerate}
\end{proposition}

Recall that the difference between consecutive prizes, $(V_l-V_{l+1})$ incentivizes contestants' efforts (Remark~\ref{rmk:PrizeGap}). The above conditions ensure that the designer's objective is linear with respect to the gaps $(V_l-V_{l+1})$. Therefore, we can identify the most profitable gap (the one with the highest weight after adjustments) and allocate the entire budget to enlarging that gap. This results in an optimal simple contest. Additionally, this approach provides an $O(n^2)$ algorithm for finding the optimal contest.

We now focus on two specific types of linear objectives, which are commonly used in practice and have been extensively studied in the literature: the maximum individual effort and total effort.

\subsection{Maximum Individual Effort}\label{subsec:HighestEffort}

Under the maximum individual effort objective, the designer aims to maximize the expected effort of the strongest contestant in equilibrium, i.e., $\mathbb{E}_{x \sim X_{(1)}}[b(x)]$. As a result, the general design guideline (Proposition~\ref{prop:DesignGuideline}) indicates that the optimal contest is a simple contest when the goal is to maximize individual effort. This reduces the design problem to finding the best combination of shortlist size $m$ and the number of prizes $l$. 
On one hand, when the prize structure is fixed, there is no need to admit more than $l+1$ contestants, as additional competition diminishes the strongest contestant's enthusiasm (Corollary~\ref{coro:EmptyPrize}), which implies that the optimal contest satisfies $m=l+1$. On the other hand, intuitively, as more contestants are admitted, the expected prize awarded to the strongest contestant decreases (Proposition~\ref{prop:ThreatenDesc}), leading to a reduction in the effort exerted. 
In summary, a two-contestant winner-take-all contest (i.e., $m=2$ and $l=1$) is the optimal design for maximizing individual effort. Since this analysis holds for any realization of ability levels, we can conclude the following theorem:


\begin{theorem}[Optimal Contest for the Ex-post Maximum Individual Effort]\label{thm:ExpostHighestEffort}
    The optimal contest that maximizes the ex-post maximum individual effort is a two-contestant winner-take-all contest, where $m=2$ and $V_1=B$.

    The resulting maximum ex-post individual effort is as follows:
    \[
    g^{-1}\left ( B \int_{0}^{x_{1}} \frac{(n-1)f(t)t}{F(t)+(n-1)(1-F(t))} \, dt \right ),
    \]
    where $x_1$ denotes the highest realized ability level among all contestants.
\end{theorem}

Thus, even for the the ex-ante maximum individual effort, the optimality of the two-contestant winner-take-all contest still holds.

\begin{corollary}[Optimal Contest for the Ex-ante Maximum Individual Effort]\label{coro:OptimalMaximumEffort}
    The optimal contest that maximizes the ex-ante maximum individual effort is a two-contestant winner-take-all contest, where $m=2$ and $V_1=B$. $m=2$ and $V_1=B$.

    The resulting maximum ex-ante individual effort is as follows:
    \[
    \mathbb{E}_{x\sim X_{(1)}} \left [ g^{-1}\left ( B \int_{0}^{x} \frac{(n-1)f(t)t}{F(t)+(n-1)(1-F(t))} \, dt \right ) \right ],
    \]
    where random variable $X_{(1)}$ denotes the highest realized ability level among all contestants.
\end{corollary}

\subsection{Total Effort}\label{subsec:TotalEffort}

Under the total effort objective, the designer aims to maximize the expected total effort of all admitted contestants in equilibrium, i.e., $\mathbb{E}_{X_{(1)}, \ldots, X_{(m)}}[\sum_{j=1}^{m}b(x_{(j)})]$.
In this subsection, we assume a linear cost function, $g(e)=ke$ with $k>0$ and show that the optimal contest as follows.
\begin{theorem}[Optimal Contest for the Total Effort Objective]\label{thm: opt for total}
    For any ability distribution $F$, the optimal contest for maximizing the total effort can be described as:
    \begin{enumerate}
        \item The number of prize $l^*$ is equal to the shortlist size $m^*$ minus one, i.e., $l^*=m^*(n)-1$.
        \item The budget is equally divided into these prizes, i.e, $V_1=V_2=\cdots= V_{l^*}=B/l^*$.
        \item The optimal shortlist is proportional to $n$, i.e., $m^*(n)=kn$, where $k$ is the solution to 
        \[
            \int_k^1 F^{-1}(1-q)(\frac{1}{q}-(2k-k^2)\frac{1}{q^2})dq=0.
        \]
    \end{enumerate}
\end{theorem}

To show the optimality of the above characterization, we first define the contests satisfy condition (1) and (2) as a complete simple contest. 
\begin{definition}[Complete Simple Contest]
    A simple contest is complete, if it is a non-trivial simple contest with shortlist size $m$ that has exactly $m-1$ prizes, i.e., $V_1 = \ldots =V_{m-1}>V_m=0$.
\end{definition}

Then, we show that the optimal contest maximizing the total effort is exactly a complete simple contest. By general design guideline, the optimal contest in this case is a simple contest. Moreover, due to linearity, the budget appears as a constant factor in the objective under simple contests and therefore does not affect optimality. To facilitate discussion, we set $B=1$ as a standing assumption.


Using the equilibrium effort expression from Theorem~\ref{thm:contestantSBNE}, the ex-ante total effort of a simple contest with shortlist size $m$ and $l$ equal prizes, denoted by $S(m,n,l)$, is given by:
\[
\mathbb{E}_{X_{(1)}, \ldots, X_{(m)}} \left [\sum_{i=1}^{m}\int_{0}^{x_{(i)}}\frac{\binom{n-1}{l}F^{n-l-1}(t)(1-F(t))^{l}}{\sum_{j=1}^{m}\binom{n-1}{j-1}F^{n-j}(t)(1-F(t))^{j-1}}\frac{f(t)}{1-F(t)} t\, dt \right ],
\]
where $X_{(i)}$ is the $i^{\text{th}}-$highest ability level among all contestants, and $x_{(i)}$ is its realization.




In Subsection~\ref{subsec:HighestEffort}, we used Corollary~\ref{coro:EmptyPrize} to show that, when prizes are fixed, admitting more contestants decreases the maximum individual effort. This monotonicity can extend to the total effort. When more contestants are allowed into the contest, the effort of the previously admitted contestants decreases, as their subjective probability of winning declines (Remark~\ref{rmk:subjectiveProb}). On the other hand, the effort contributed by the newcomers cannot compensate for this, as their ability levels are lower. Thus, we have:

\begin{proposition}\label{thm:ConpleteSimpleContest}
    The optimal contest with a shortlist that maximizes total effort is a complete simple contest.
\end{proposition}


The decision problem now reduces to determining the optimal shortlist size $m$ (with the number of prizes being $m-1$ accordingly). We explore additional properties of the optimal contest through asymptotic analysis. First, we rewrite the total effort objective using the beta distribution:
\begin{lemma}[Beta Representation for Total Effort]\label{lem:betaRepTotalEffort}
    The ex-ante total effort in a complete simple contest, denoted by $S(m,n)$ (abbreviated from $S(m,n,m-1)$ when no confusion arises), can be expressed using the beta distribution $\beta(x,a,b)$, as follows:
    \[
    \begin{aligned}
        S(m,n) = & 
        \int_0^1F^{-1}(q)\beta(q,n-m+1,m)\,dq \\
        & +\int_0^1F^{-1}(q)\frac{q}{1-q} \frac{m}{n-m}\frac{\beta(q,n-m,m)}{\int_0^q\beta(x,n-m,m)\,dx}\int_q^1\beta(x,n-m,m+1)\,dx\, dq,
    \end{aligned}
    \] where $\beta(x,a,b) =x^{a-1}(1-x)^{b-1}B(a,b)^{-1}$ is the probability density function of the beta distribution, parametrized by positive integers $a$ and $b$, and $B(a,b)=\frac{\Gamma(a)\Gamma(b)}{\Gamma(a+b)}$ is the beta function.
\end{lemma}

The concentration behavior of the beta distribution\footnote{In general, $\beta(q,\alpha,\beta)$ concentrates around its maximum point $\mu=\frac{\alpha-1}{\alpha+\beta-2}$ as $\alpha$ and $\beta$ go large.} allows us to asymptotically approximate the integration terms in the expression (Lemmas~\ref{lem:1} and~\ref{lem:2}). Thus, the representation simplifies to:
\begin{lemma}[Asymptotic Expression for Total Effort]\label{lem:AsyRep}
    The ex-ante total effort in a complete simple contest has the following asymptotic expression with respect to $n$:
    $$\lim_{n\rightarrow +\infty}\frac{S(m,n)}{n}=\int_0^{1-k}F^{-1}(q)\frac{q}{1-q}\frac{k}{1-k}(\frac{1}{q}-\frac{k}{q(1-q)})dq,$$
    where $k=m/n$, and the convergence rate is independent of $k$. 
\end{lemma}

Following this expression, we treat the admitting ratio $k$ as a decision variable and solve the first-order condition. This leads to the asymptotic optimal size, which is proportional to $n$.

\begin{proposition}[Optimal size is Asymptotically Linear]\label{thm:OptAsmLinear}
The ex-ante total effort in a complete simple contest converges to a function of $k=m/n$ as $n \rightarrow \infty$, and the convergence rate is independent of 
$k$. Therefore, the optimal shortlist size grows asymptotically linearly with $n$. Formally, there exists a $k^*\in(0,1)$ such that:
$$\lim_{n\rightarrow \infty}m^*(n)/n=k^*,$$
where $k^*$ is the solution of the following equation:
\[
\int_k^1 F^{-1}(1-q)(\frac{1}{q}-(2k-k^2)\frac{1}{q^2})dq=0.
\]
\end{proposition}

Proposition~\ref{thm:OptAsmLinear} provides a general way to solve the optimal size for any given distribution, i.e., to solve the equation of first-order condition, as can be seen from following examples. 

\begin{example}[Optimal size for Uniform Distributions]\label{exam:OptimalUniform}
    For a uniform distribution $U[0,b]$ that starting at $0$, i.e., $F(x)=x/b$, then $F^{-1}(1-q)=b(1-q)$. Following Theorem~\ref{thm:OptAsmLinear}, the optimal ratio $k^*$ satisfies:\(\int_k^1 b(1-q)(\frac{1}{q}-(2k-k^2)\frac{1}{q^2})dq=0 \).
    This equation simplifies to $\ln k=1+\frac{4-2k}{k^2-2k-1}$, which can be solved numerically, obtaining $k^* \approx 15.07\%.$
\end{example}
\begin{example}[Optimal size for Square Function Distribution]\label{exam:OptimalSquare}
    For square function distribution, i.e., $F(x) = x^2, x\in[0,1]$, then $F^{-1}(1-q) = (1-q)^{1/a}$. The first-order condition $\sqrt{1-k}(k-4)+(k^2-2k+2)\ln(1+\sqrt{(1-k)})+1/2 \cdot(k^2-2k-2)\ln k =0$ gives that $k^*\approx 20.67\%$. 
\end{example}
\begin{example}[Optimal size for Exponential Distributions]\label{exam:OptimalExp}
    For exponential distribution, i.e., $F(x)=1-e^{-\lambda x}$, then $F^{-1}(1-q)=-\frac{1}{\lambda}\ln q$. The first-order condition $\frac{1}{2\lambda}((\ln k)^2+(4-2k)\ln k+2(k-2)(k-1))= 0$ gives that $k^*\approx 9.70\%$, which is independent of the parameter $\lambda$.
\end{example}


In this way, we have fully characterized the optimal contest with a shortlist for the total effort objective. A natural question arises: Is there an upper bound for the optimal admitting ratio? In other words, can the contest designer eliminate a portion of contestants before knowing anything about the ability distribution, while still seeking to maximize total effort? Surprisingly, the answer is yes. Formally, we state the following theorem:

\begin{theorem}[Tight Upper Bound for the Optimal size]\label{thm:UniversalBound}
    For an arbitrary ability distribution $F$, when $n\rightarrow +\infty$, the optimal shortlist size that maximizes ex-ante total effort, denoted by $m^{*}$, has the following linear upper bound with respect to $n$:
    $$ \lim_{n \rightarrow \infty} \frac{m^*(n)}{n} \leq \bar{k},$$
    where $\bar{k} \approx 31.62\%$ is the solution to the equation $\ln k=(2-k)(k-1)$.

    Moreover, there exists a distribution such that $m^*(n)/n=\bar{k}$, meaning the bound is tight. 
\end{theorem}

Interestingly, with no assumptions on the ability distribution, the designer can eliminate up to approximately 68.38\% of the contestants without sacrificing the optimal contest that maximizes total effort. Less competitive contestants contribute little effort, but their presence significantly discourages stronger contestants, as they are perceived to be as competitive as the others. By eliminating weaker contestants, the stronger contestants have more room to fully compete, which increases the total effort and allows the designer to fully leverage their informational advantage.

\section{Shortlist vs. No Shortlist}
\label{sec: compare}
How much does a contest designer benefit from implementing a shortlist? In this section, we address this question by comparing the optimal contest designs with and without a shortlist under a linear cost function. Specifically, we focus on three types of contests:
\begin{enumerate}
    \item The $n$-contestant winner-take-all contest – This is the optimal contest without a shortlist for both maximum individual effort and total effort \cite{MS01,CHS19}.
    \item The two-contestant winner-take-all contest – This is the optimal contest with a shortlist for the maximum individual effort (by Theorem~\ref{thm:ExpostHighestEffort}).
    \item The complete simple contest with a shortlist size of $m^*$ and $m^*-1$ prizes - This is the optimal contest with a shortlist for the total effort (by Proposition~\ref{thm:ConpleteSimpleContest}).
\end{enumerate}

For any ability distribution $F$ and $n$ initial registered contestants, let $S^{(1)}(m,n,l)$ and $S(m,n,l)$ denote the maximum individual effort and total effort, respectively, in a simple contest with a shortlist size of $m$ and $l\leq m$ prizes. To quantify the gap in objectives, we establish bounds on $S^{(1)}(m,n,l)$ and $S(m,n,l)$ for the three contest types discussed above.

Before deriving bounds on the objectives, we first express the maximum individual effort and total effort objectives in terms of quantiles. Given any ability distribution $F$, recall that the quantile is defined as $q:=1-F(x)$ and its reverse function $v(q):=F^{-1}(1-q)=x$. Using this notation, we obtain the following result.

\begin{lemma}[Quantile Representation for Effort]\label{lem:QuantileRep}
    By using quantile $q:=1-F(x)$ and its reverse function $v(q):=F^{-1}(1-q)=x$, the ex-ante maximum individual effort of a simple contest is:
    \[
    S^{(1)}(m,n, l)= n\int_0^1|v'(q)|\int_0^qG^{(1)}_{(m,l)}(t)\,dt\,dq,
    \]
    where $l$ is the number of prizes and $G_{l,m}^{(1)}(t):=\frac{\binom{n-1}{l}(1-t)^{n-l-1}t^{l-1}}{\sum_{j=1}^{m}\binom{n-1}{j-1}(1-t)^{n-j}t^{j-1}}\int_{0}^{t}(1-p)^{n-1}\,dp.$
    
    And the ex-ante total effort of a simple contest expresses as:
    \[
    S(m,n, l)= n\int_0^1|v'(q)|\int_0^qG_{(m,l)}(t)\,dt\,dq,
    \]similarly, $G_{(m,l)}(t):=\frac{\binom{n-1}{l}(1-t)^{n-l-1}t^{l-1}}{\sum_{j=1}^{m}\binom{n-1}{j-1}(1-t)^{n-j}t^{j-1}}\int_{0}^{t}\sum_{j=1}^{m}\binom{n-1}{j-1}p^{j-1}(1-p)^{n-j}\,dp$, and we use $H(q):=\int_0^qG(t)\,dt$ to denote the distribution-free part.  
\end{lemma}

Based on this representation, we establish the ranges for both the maximum individual effort and the total effort achieved by the $n$-contestant winner-take-all contest and the $2$-contestant winner-take-all contest, respectively, as stated in the following two lemmas.

\begin{lemma}\label{lem:bound on n,1}
    For any ability distribution, the $n$-contestant winner-take-all contest achieves a maximum individual effort of $S^{(1)}(n, n,1) = \Theta(1)$ and a total effort of $S(n, n,1) = \Theta(1)$.
\end{lemma}
    
\begin{lemma}\label{lem:bound on 2,1}
    For any ability distribution, the $2$-contestant winner-take-all contest achieves a maximum individual effort of $S^{(1)}(2, n,1) = \Theta(\log n)$ and a total effort of $S(2, n,1) = \Theta(\log n)$.
\end{lemma}

On the other hand, using the Beta representation for total effort (Lemma \ref{lem:betaRepTotalEffort}), we can derive the range of total effort achieved by a complete simple contest. 

\begin{lemma}\label{lem:bound on m,m-1}
    Fixed any ability distribution $F$, the optimal complete simple contest with a shortlist size $m^*$ achieves a total effort of $S(m^*,n, m^*-1) = \Theta(n)$, where $m^*$ depends on the distribution $F$.
\end{lemma}

Based on these three lemmas, we compare the optimal contests with and without a shortlist in terms of two types of objectives.

\noindent \textbf{Maximum Individual Effort.}
For the maximum individual effort objective, we only focus on comparing the two-contestant winner-take-all contest with the $n$-contestant winner-take-all contest in terms of the maximum individual effort they can achieve, as these are optimal for all feasible ability distributions. Based on Lemmas \ref{lem:bound on n,1} and \ref{lem:bound on 2,1}, we derive the following theorem.

\begin{theorem}\label{thm: 2,1 vs n,1 max effort}
    For any ability distribution, under the maximum individual effort objective, the two-contestant winner-take-all contest results in $\Theta(\log n)$ times the maximum individual effort of the optimal contest without a shortlist. Specifically, we have:
    $$\frac{S^{(1)}(2,n,1)}{S^{(1)}(n,n,1)} = \Theta(\log n).$$
\end{theorem}

\noindent \textbf{Total Effort.}
For the total effort objective, while the optimal contest with a shortlist is a complete simple contest as shown by Theorem \ref{thm:ConpleteSimpleContest}, the optimal shortlist size depends on the ability distribution. This means the optimal contests vary with different ability functions. To make a meaningful comparison, we focus on a fixed ability distribution $F$ and examine the gap in total effort between the optimal contest with a shortlist under $F$ and the $n$-contestant winner-take-all contest. Based on Lemma \ref{lem:bound on n,1} and \ref{lem:bound on m,m-1}, the comparison of these two contests is summarized in Theorem \ref{thm:TotalOPTVAN}.

\begin{theorem}\label{thm:TotalOPTVAN}
    Fixed any ability distribution $F$, under the total effort objective, the optimal contest with a shortlist can achieve $\Theta(n)$ times the total effort compared to the optimal contest without a shortlist. Specifically,
    $$\frac{S(m^*,n,m^*-1)}{S(n,n,1)} = \Theta(n),$$
    where $m^*$ is the optimal shortlist size for the ability distribution $F$. 
\end{theorem}

Additionally, with respect to total effort, we still can compare the fully shortlisted contest (the 2-contestant winner-take-all contest) with the optimal contest without a shortlist (the $n$-contestant winner-take-all contest), as presented in the following proposition.
\begin{proposition}\label{prop:TotalTWOVAN}
    For any ability distribution, under the total effort objective, the two-contestant winner-take-all contest results in $\Theta(\log n)$ times the total effort of the optimal contest without a shortlist. Specifically, we have:
    $$\frac{S^{(1)}(2,n,1)}{S^{(1)}(n,n,1)} = \Theta(\log n).$$
\end{proposition}

\section{Towards Practical Applications}\label{sec:practicalApp}

In previous sections, we use asymptotic analysis to characterize and solve the optimal contests, which has a potential application value. In this section, we further provide numerical results to support that our findings fit well even for small $n$, bringing them closer to practical implementation.
In this section, we focus on the total effort objective, the most widely used in practice.

\noindent \textbf{Finding the Optimal Contest.} The optimal contest is a complete simple contest (Proposition~\ref{thm:ConpleteSimpleContest}), and the corresponding shortlist size grows linearly with  $n$. The slope $k$ is determined by the ability distribution and can be identified by solving an equation (Theorem~\ref{thm:OptAsmLinear}, asymptotic).

Figure~\ref{fig:DistributionOpt} illustrates that this linear trend emerges even for very small values of  $n$, with the asymptotic ratio providing a close prediction. Therefore, the optimal contest for any given distribution can be determined efficiently

\begin{figure}[htb]
\centering
\begin{subfigure}[ht]{0.30\textwidth}
    \centering
    \includegraphics[width=\textwidth]{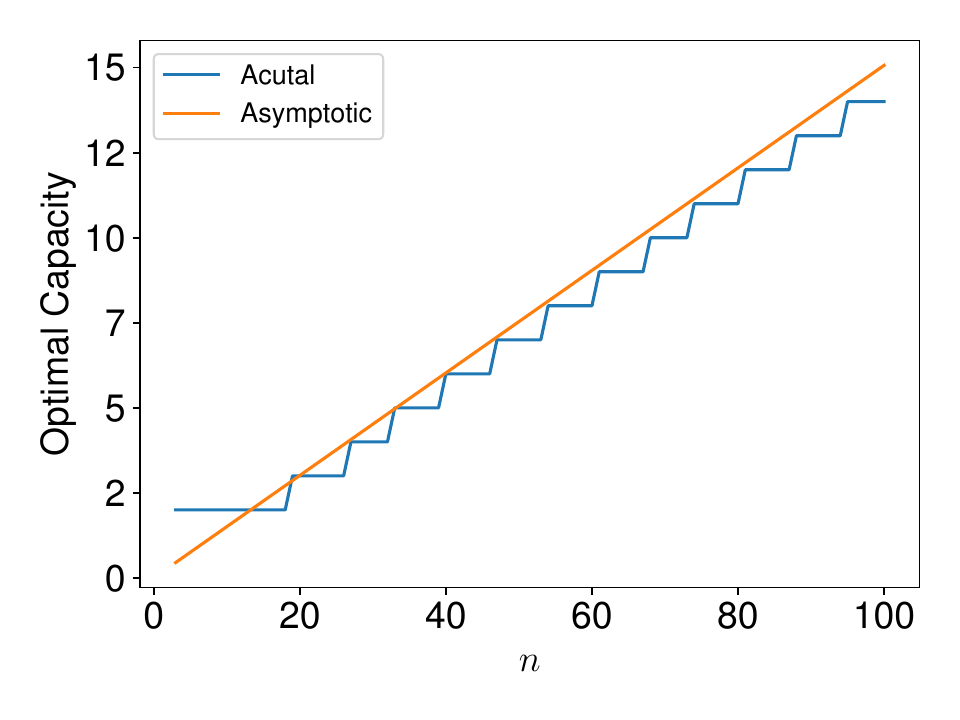}
    \subcaption{$F(x)=x$}
    \label{fig:disopt-a}
    \end{subfigure}
\begin{subfigure}[ht]{0.30\textwidth}
    \centering
    \includegraphics[width=\textwidth]{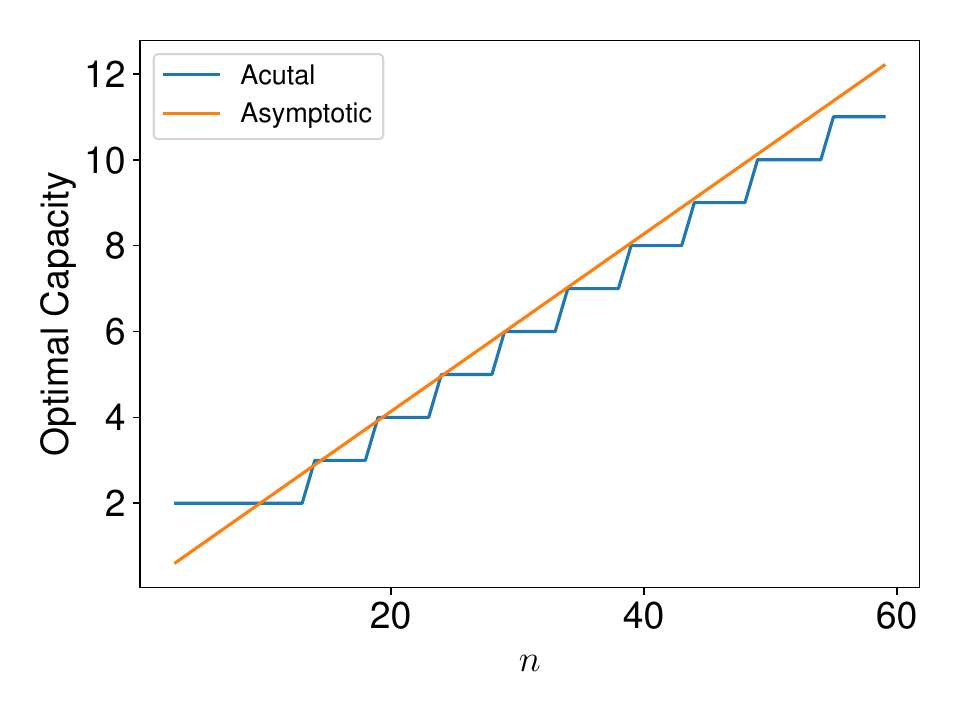}
    \subcaption{$F(x)=x^2$}
    \label{fig:disopt-b}
    \end{subfigure}
\begin{subfigure}[ht]{0.30\textwidth}
    \centering
    \includegraphics[width=\textwidth]{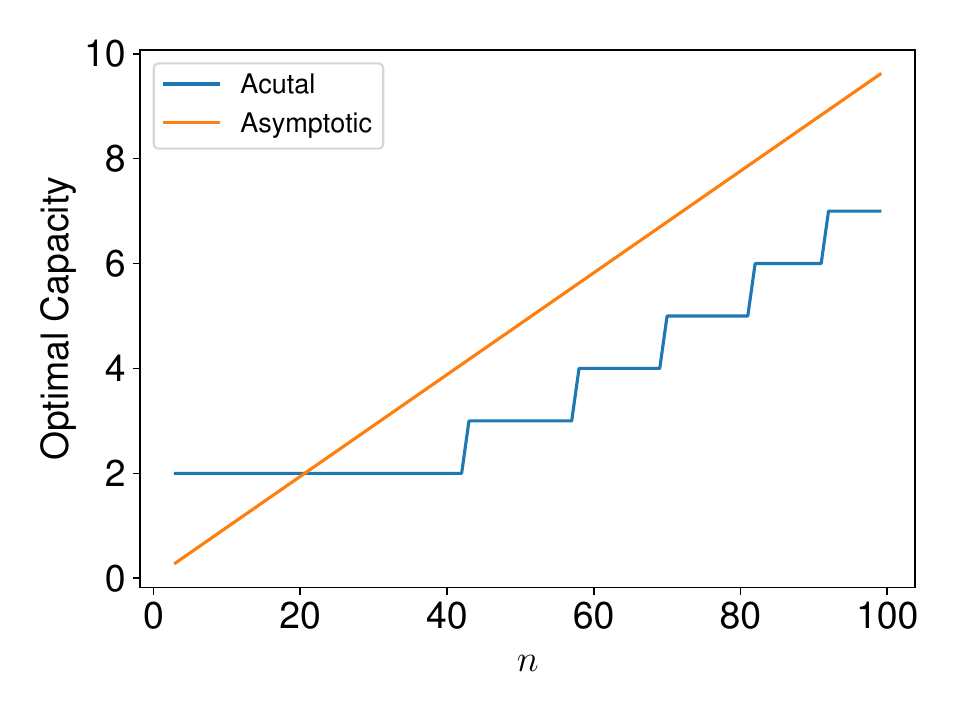}
    \subcaption{$F(x)=1-e^{-x}$}
    \label{fig:disopt-c}
    \end{subfigure}

\caption{The actual optimal size and $m^*$ predicted by asymptotic relation.}
\label{fig:DistributionOpt}
\end{figure}
\noindent \textbf{Universal Upper Bound of the Optimal Size.} There is no distribution such that the optimal shortlist size is larger than $31.62\%n$ (Theorem~\ref{thm:UniversalBound}, asymptotic).
We propose an $O(n)$ algorithm to determine the supremum of the optimal shortlist size across all distributions for any given $n$ (Proposition~\ref{prop:SupM}). Numerical results in Figure~\ref{fig:universal1} confirm that the asymptotic linear trend holds even for very small $n$. Therefore, contest designers can confidently eliminate nearly 68\% of contestants without compromising optimality.

\begin{figure}
    \centering
    \begin{minipage}{0.30\linewidth}
        \centering
        \includegraphics[width=\textwidth]{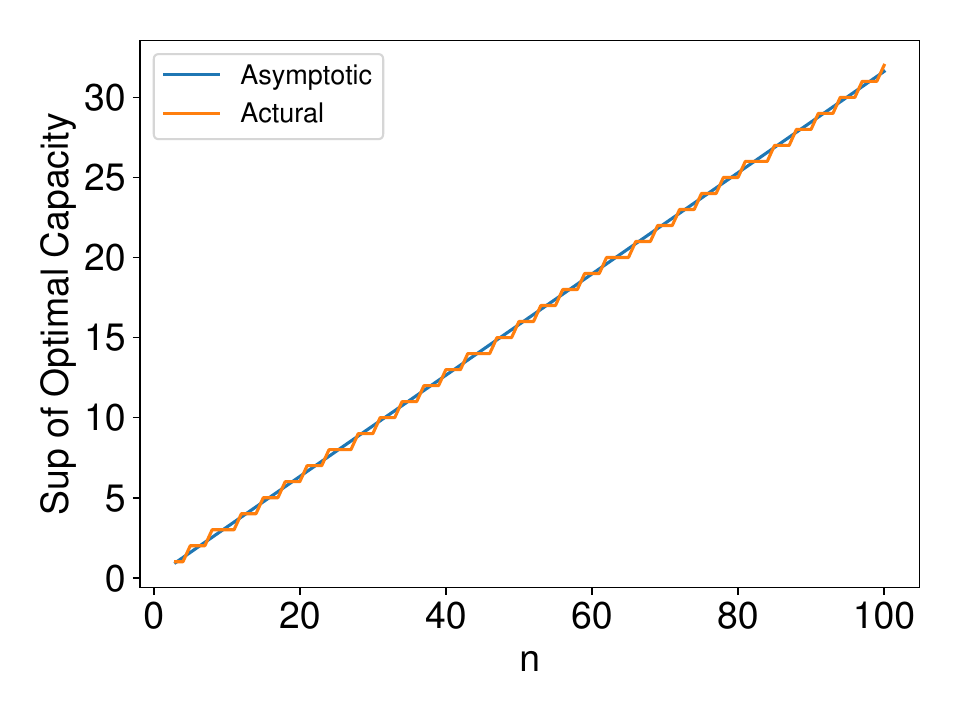}
        \caption{Supremum of optimal $m$.}
        \label{fig:universal1}
    \end{minipage}
    \begin{minipage}{0.55\linewidth}
        \centering
        \includegraphics[width=\textwidth]{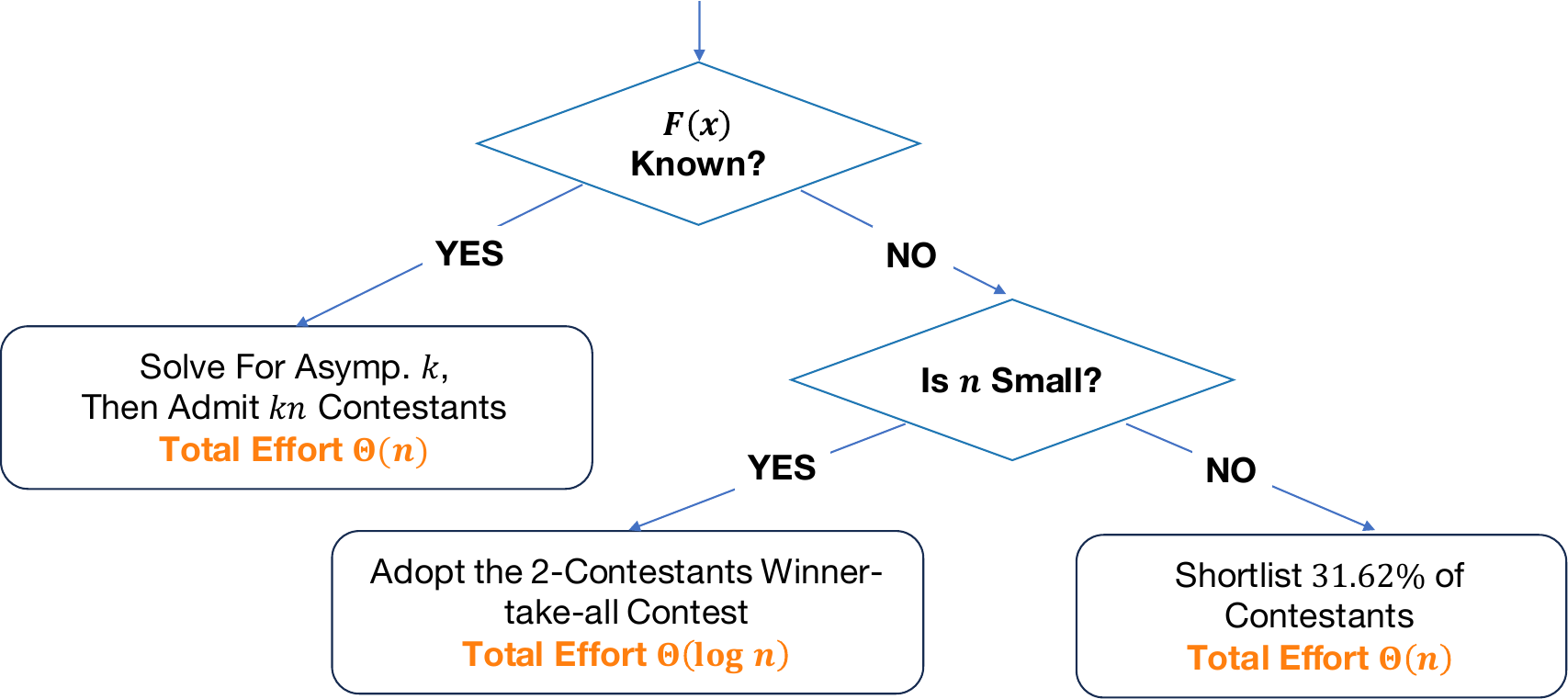}
        \caption{A flow chart for contest design in practice.}
        \label{fig:flowchart}
    \end{minipage}
\end{figure}


\noindent \textbf{Performance Enhancement.} For any distribution, the two-player winner-take-all contest is $\Theta(\log n)$ times better (Proposition~\ref{prop:TotalTWOVAN}) and the optimal contest with a shortlist, it is $\Theta(n)$ times better (Theorem~\ref{thm:TotalOPTVAN}, asymptotic), compared to the optimal one without a shortlist. Moreover, even when the distribution is unknown, the designer can still attain a $\Theta(n)$ improvement simply by shortlisting to 31.62\% of the contestants, compared to having no shortlist (Corollary~\ref{coro:shortlistAlways}).

In Figure~\ref{fig:DistributionOpt1}, numerical results demonstrate that the asymptotic approximation ratio holds even for small $n$, and the performance of the asymptotic optimal contest is nearly identical to that of the actual optimal contest. This indicates that our algorithm yields a contest design that is not only near-optimal and highly effective at small scales but also guarantees asymptotic optimality and a strong approximation ratio. 

\begin{figure}[h]
\begin{subfigure}[ht]{0.30\textwidth}
    \centering
    \includegraphics[width=\textwidth]{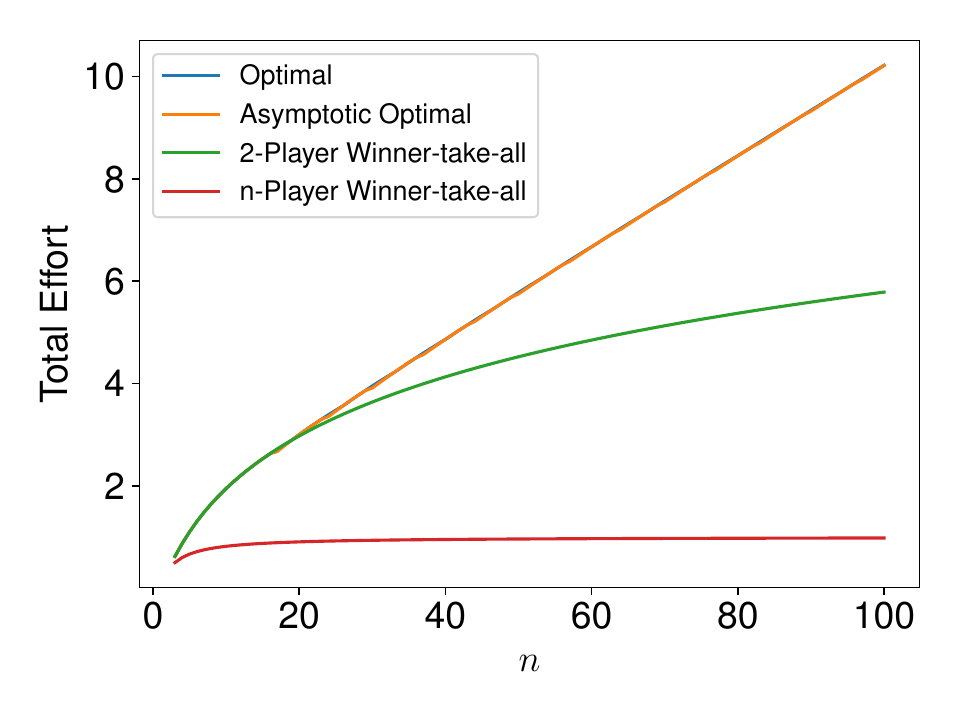}
    \subcaption{$F(x)=x$}
    \label{fig:disopt-a1}
    \end{subfigure}
\begin{subfigure}[ht]{0.30\textwidth}
    \centering
    \includegraphics[width=\textwidth]{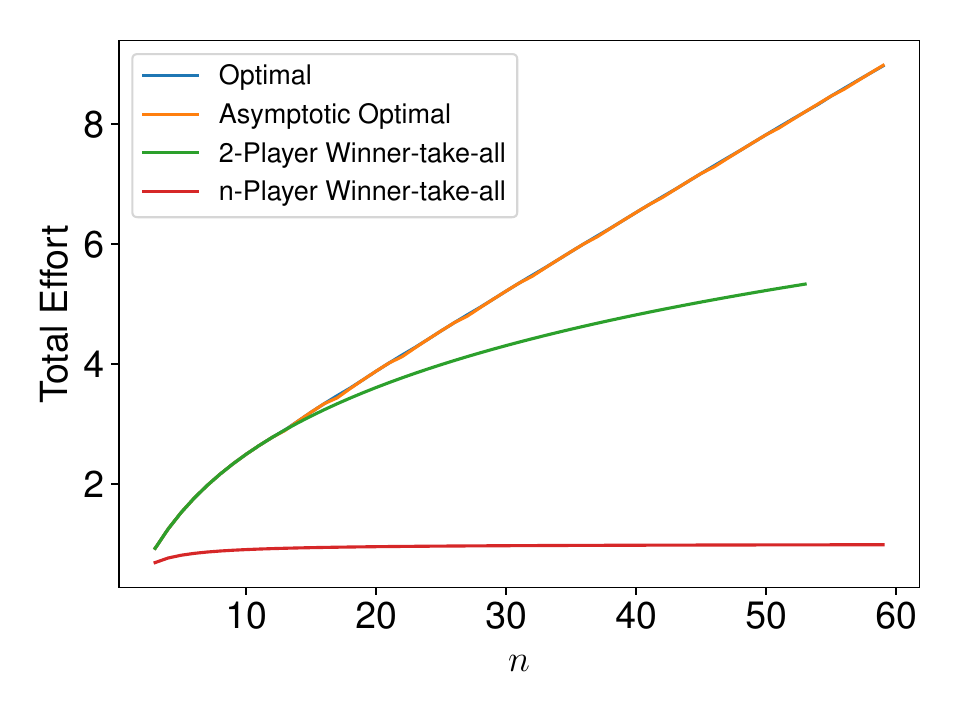}
    \subcaption{$F(x)=x^2$}
    \label{fig:disopt-b1}
    \end{subfigure}
\begin{subfigure}[ht]{0.30\textwidth}
    \centering
    \includegraphics[width=\textwidth]{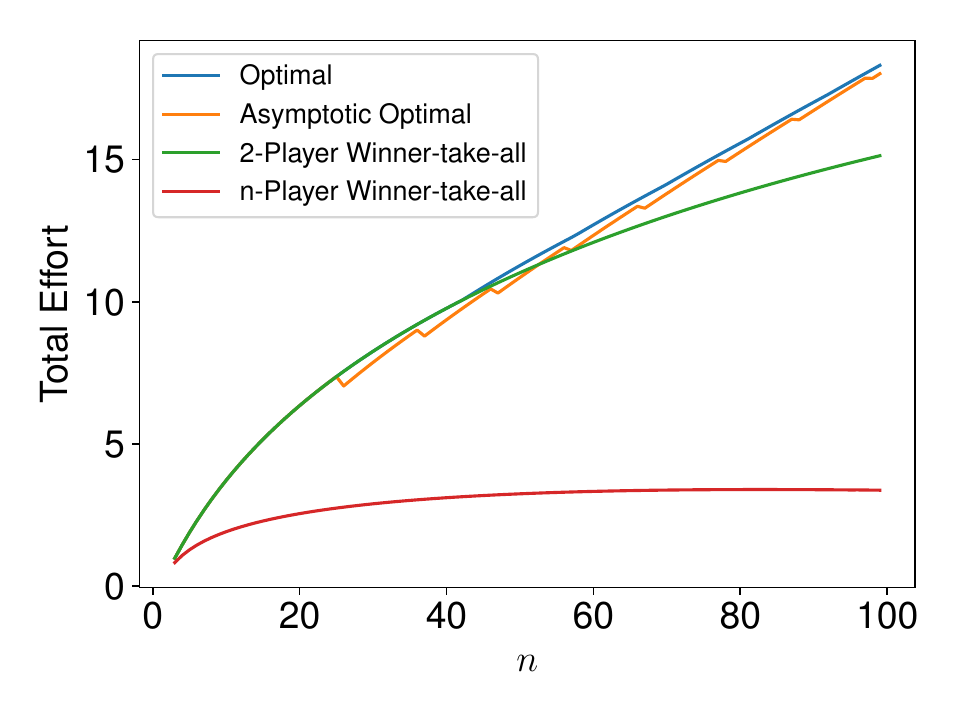}
    \subcaption{$F(x)=1-e^{-x}$}
    \label{fig:disopt-c1}
    \end{subfigure}
    
\caption{Total effort performance of different contest designs.}
\label{fig:DistributionOpt1}
\end{figure}

Finally, we summarize our findings into a practical guideline, as illustrated in Figure~\ref{fig:flowchart}.

\noindent \textbf{Contest Design Cheatsheet.} 
Is Distribution known? YES $\to$ Solve for the asymptotic slope $k \leq 31.62\%$, admit $kn$, and achieve almost-optimal performance, get $\Theta(n)$. NO $\to$ Is $n$ small? YES $\to$ Use 2-Contestant winner-take-all, get $\Theta(\log n)$. NO $\to$ Admit $31.62\%n$, get $\Theta(n)$.

\section{Conclusion and Future Works}
\label{sec:conclusion}

In this work, we study the optimal design of a rank-order contest with a shortlist for two objectives: the maximum individual effort and total effort. The designer must determine both the shortlist size and the prize structure, with only shortlisted contestants exerting costly effort to compete for prizes.
First, we fully characterize the unique symmetric Bayesian Nash equilibrium of admitted contestants. Next, we provide a detailed characterization of the optimal contests with a shortlist for both objectives. Finally, we compare the performance of optimal contests with and without a shortlist, establishing asymptotically tight bounds.

Several interesting directions remain for future research. First, when there are no restrictions on the cost function, the optimal contest with a shortlist remains unknown for the total effort. Second, investigating the approximation ratio in the general case, particularly for general cost functions, is a valuable problem. Lastly, How do contestant equilibria and optimal contest design, if there multiple contests with a shortlist?

\bibliographystyle{unsrtnat}
\bibliography{sample-bibliography}

\newpage

\appendix
\newpage
\section{Missing Proofs in Section \ref{sec:playerSBNE}}

\subsection*{Proof of Proposition~\ref{prop:posteriorBeliefs}}
\begin{proof}
The observer knows her own ability as well as the promotion status of others. Therefore, the joint posterior probability density \(\beta_1(\mathbf{x})\) is given by:
\[
\beta_1(\mathbf{x}) = \frac{\Pr\left[ \bigwedge_{i=2}^{m} X_i = x_i, \bigwedge_{j=m+1}^{n} X_j \leq \min_{k \in [m]}(X_k), X_1 = x_1 \right]}{\Pr\left[ \bigwedge_{j=m+1}^{n} X_j \leq \min_{k \in [m]}(X_k), X_1 = x_1 \right]}.
\]

For the denominator, we apply the law of total probability over \(X_2, \dots, X_m\):
\[
\Pr\left[ \bigwedge_{j=m+1}^n X_j \leq \min_{k \in [m]}(X_k), X_1 = x_1 \right] =
\sum_{(x_2, \dots, x_m) \in \mathcal{D}_X} \Pr\left[ \bigwedge_{j=m+1}^n X_j \leq \min_{k \in [m]}(x_k) \right] \Pr\left[ \bigwedge_{i=1}^m X_i = x_i \right].
\]

Expanding the probabilities:
\[
\sum_{(x_2, \dots, x_m) \in \mathcal{D}_X} \prod_{j=m+1}^n \Pr\left[X_j \leq \min_{k \in [m]}(x_k)\right] \prod_{i=1}^m f(x_i).
\]

Next, we classify the cases based on the minimum ability among the admitted candidates:

\underline{Case 1:} Observer's ability is the lowest among the admitted (\(x^{(1)} \geq x_1\)).

In this case, \(\min_{k \in [m]}(x_k) = x_1\), so:
\[
\prod_{j=m+1}^n \Pr\left[X_j \leq \min_{k \in [m]}(x_k)\right] = \prod_{j=m+1}^n \Pr\left[X_j \leq x_1\right].
\]

Thus, the integral becomes:
\[
\begin{aligned}
& \int_{x^{(1)} \geq x_1} \prod_{j=m+1}^n \Pr\left[X_j \leq x_1\right] \prod_{i=1}^m f(x_i) \, d\mathbf{x_{-1}} \\
= & F^{n-m}(x_1) f(x_1) \int_{x^{(1)} \geq x_1} \prod_{i=1}^m f(x_i) \, d\mathbf{x_{-1}} \\
= & F^{n-m}(x_1) f(x_1) \prod_{i=1}^m \int_{0}^{x_1} f(x_i) \, dx_i \\
= & F^{n-m}(x_1) (1 - F(x_1))^{m-1} f(x_1).
\end{aligned}
\]

\underline{Case 2:} Observer's ability is not the lowest among the admitted (\(x^{(1)} < x_1\)).

In this case, \(\min_{k \in [m]}(x_k) = x^{(1)}\), so:
\[
\prod_{j=m+1}^n \Pr\left[X_j \leq \min_{k \in [m]}(x_k)\right] = \prod_{j=m+1}^n \Pr\left[X_j \leq x^{(1)}\right].
\]

The integral becomes:
\[
\begin{aligned}
& \int_{x^{(1)} < x_1} \prod_{j=m+1}^n \Pr\left[X_j \leq x^{(1)}\right] \prod_{i=1}^m f(x_i) \, d\mathbf{x_{-1}} \\
= & (m-1) \int_{x_2 < x_1} \int_{\mathbf{x_{-1,2}} \geq x_2} \prod_{j=m+1}^n \Pr\left[X_j \leq x_2\right] \prod_{i=1}^m f(x_i) \, d\mathbf{x_{-1,2}} dx_2 \\
= & (m-1) f(x_1) \int_{0}^{x_1} F^{n-m}(x_2) \int_{\mathbf{x_{-1,2}} \geq x_2} \prod_{i=3}^m f(x_i) \, d\mathbf{x_{-1,2}} f(x_2) \, dx_2 \\
= & (m-1) f(x_1) \int_{0}^{x_1} F^{n-m}(x_2) \left(\prod_{i=3}^m \int_{x_2}^1 f(x_i) \, dx_i\right) f(x_2) \, dx_2 \\
= & (m-1) f(x_1) \int_{0}^{x_1} F^{n-m}(x_2) (1 - F(x_2))^{m-2} f(x_2) \, dx_2 \\
= & (m-1) f(x_1) \int_{0}^{F(x_1)} t^{n-m} (1-t)^{m-2} \, dt.
\end{aligned}
\]

Combining the results from both cases, the denominator becomes:
\[
\int_{x^{(1)} < x_1} \prod_{j=m+1}^n \Pr\left[X_j \leq x^{(1)}\right] \prod_{i=1}^m f(x_i) \, d\mathbf{x_{-1}} + \int_{x^{(1)} \geq x_1} \prod_{j=m+1}^n \Pr\left[X_j \leq x_1\right] \prod_{i=1}^m f(x_i) \, d\mathbf{x_{-1}}.
\]

This simplifies to:
\[
= \left[F^{n-m}(x_1)(1 - F(x_1))^{m-1} + (m-1) \int_{0}^{F(x_1)} t^{n-m} (1-t)^{m-2} \, dt\right] f(x_1).
\]

Or equivalantly:
\[
\binom{n-1}{m-1}^{-1}J(F, n, m, x_1) f(x_1).
\]

For the numerator, using the conditional probability formula:
\[
\begin{aligned}
& \Pr\left[\bigwedge_{i=2}^m X_i = x_i, \bigwedge_{j=m+1}^n X_j \leq \min_{k \in [m]}(X_k), X_1 = x_1\right] \\
= & \Pr\left[\bigwedge_{j=m+1}^n X_j \leq \min_{k \in [m]}(x_k) \mid \bigwedge_{i=1}^m X_i = x_i\right] \Pr\left[\bigwedge_{i=1}^m X_i = x_i\right] \\
= & \left(\prod_{j=m+1}^n \Pr\left[X_j \leq \min_{k \in [m]}(x_k)\right]\right) \prod_{i=1}^m f(x_i).
\end{aligned}
\]

We also analyze by cases. When observer's ability is the lowest among the admitted (\(x^{(1)} \geq x_1\)), the numerator becomes: 
\[
    \begin{aligned}
        & \left ( \prod_{j=m+1}^{n} Pr\left [ X_j \leq x_1 \right ] \right ) \prod_{i=1}^mf(x_i) \\
        = & F^{n-m}(x_1)\prod_{i=1}^mf(x_i)  
    \end{aligned}
\]

Similarly, when observer's ability is not the lowest among the admitted (\(x^{(1)} < x_1\)), the numerator becomes $F^{n-m}(x^{(1)})\prod_{i=1}^mf(x_i)$. 

Combining the numerator and denominator, we obtain:
\[
    \beta_1(\mathbf{x}) =   
    \begin{cases} 
    \frac{\binom{n-1}{m-1}F^{n-m}(x^{(1)})\prod_{i=2}^{m}f(x_i)}{J(F,n,m,x_1)} & \text{if } x^{(1)} \leq x_1, \\
    \frac{\binom{n-1}{m-1}F^{n-m}(x_1)\prod_{i=2}^{m}f(x_i)}{J(F,n,m,x_1)} & \text{if } x^{(1)} > x_1.
    \end{cases},
\]

which completes the proof.
\end{proof}

\begin{lemma} \label{lem:betaRep}
Let the incomplete Beta function be defined as 
\(
B_x(a, b) = \int_{0}^{x} t^{a-1}(1-t)^{b-1} \, dt,
\)
and the normalized incomplete Beta function as 
\(
I_x(a, b) = \frac{B_x(a, b)}{B(a, b)}
\),
where \(B(a, b)\) is the Beta function. For a random variable \(X \sim \text{Binomial}(n, p)\), the following equality holds:
\[
\Pr(X \leq k) = I_{1-p}(n-k, k+1) = 1 - I_p(k+1, n-k).
\]
\end{lemma}

\begin{lemma}[Interpretation of the Normalizer] \label{lem:normal} The normalization factor $J(F,n,m,x_i)$ is the prior probability of contestant $i$ advancing into the second stage, i.e., the following equality holds:
\[
    J(F,n,m,x_i) = \sum_{j=1}^{m}\binom{n-1}{j-1}F^{n-j}(x_i)(1-F(x_i))^{j-1},
\]
which also gives that $J(F,n,m,x_i)$ is an increasing function of shortlist capacity $m$.    
\end{lemma}
\begin{proof}
By definition, The expression for \( J(F, n, m, x_i) \) is given as:
\[
J(F, n, m, x_i) = \binom{n-1}{m-1}F^{n-m}(x_i)(1-F(x_i))^{m-1} + \binom{n-1}{m-1}(m-1)\int_0^{F(x_i)}t^{n-m}(1-t)^{m-2} \, dt.
\]

The second term on the right-hand side can be rewritten as:
\[
\begin{aligned}
    & \frac{(n-1)!}{(m-2)!(n-m)!} B_{F(x_i)}(n-m+1, m-1) \\
    = & \left[ \frac{\Gamma(n-m+1)\Gamma(m-1)}{\Gamma(n)} \right]^{-1} B_{F(x_i)}(n-m+1, m-1) \\
    = & \frac{B_{F(x_i)}(n-m+1, m-1)}{B(n-m+1, m-1)} \\
    = & I_{F(x_i)}(n-m+1, m-1),
\end{aligned}
\]

where \( B_{F(x_i)}(n-m+1, m-1) \) is the incomplete beta function. According to the relationship provided by Lemma~\ref{lem:betaRep}, the above expression corresponds to the value of \( \Pr(X \leq m-2) \), where \( X \sim Binomial(n-1, 1-F(x_i)) \). That is:
\[
\sum_{j=0}^{m-2} \binom{n-1}{j} F^{n-1-j}(x_i)(1-F(x_i))^j.
\]

Substituting back into the original expression and re-indexing terms, we obtain:
\[
J(F, n, m, x_i) = \sum_{j=1}^{m}\binom{n-1}{j-1}F^{n-j}(x_i)(1-F(x_i))^{j-1},
\]

which represents the probability that the ability of contestant \( i \) ranks among the top \( m \) contestants.

Since each term in the summation is positive, it follows that \( J(F, n, m, x_i) \) is monotonically increasing respect to shortlist capacity $m$, which completes the proof.
\end{proof}

\subsection*{Proof of Corollary~\ref{prop:marginalBelief}}
\begin{proof}
We perform a case analysis based on the value of \( z \). Since the denominator \( J(F,n,m,x_1) \) is independent of \( z \), we omit it in our discussion and focus on the marginalization of the numerator:

\underline{Case 1:} \( z \leq x_1 \). In this scenario, \( x^{(1)} \leq x_1 \). We further classify based on the relationship between \( z \) and \( x^{(1)} \), dividing the marginalization integral into two parts:

The first part, when \( z = x^{(1)} \), the marginalization of the numerator is:
\[
\begin{aligned}
    & \binom{n-1}{m-1} \int_{\mathbf{x_{-1,2}} \geq z} F^{n-m}(z) \prod_{k=2}^m f(x_k) \, d\mathbf{x_{-1,2}} \\ 
    & = \binom{n-1}{m-1} F^{n-m}(z) \left ( \int_{z}^{1} f(x) \, dx \right )^{m-2} f(z) \\
    & = \binom{n-1}{m-1} F^{n-m}(z) (1-F(z))^{m-2} f(z).
\end{aligned}
\]

The second part, when \( z > x^{(1)} \): Without loss of generality, we reorder the indices so that the competitor corresponding to \( x^{(1)} \) is indexed by $3$. The marginalization expression then becomes:
\[
\begin{aligned}
    & \binom{n-1}{m-1}(m-2) \int_0^z \left ( \underbrace{\int_{x_3}^{1} \cdots \int_{x_3}^{1}}_{m-3} F^{n-m}(x_3) \prod_{k=2}^m f(x_k) \, d\mathbf{x_{-1,2,3}} \right ) \, dx_3 \\
    = & \binom{n-1}{m-1}(m-2) f(z) \int_0^{z} F^{n-m}(x_3) (1-F(x_3))^{m-3} f(x_3) \, dx_3 \\
    = & \binom{n-1}{m-1}(m-2) f(z) \int_0^{F(z)} t^{n-m} (1-t)^{m-3} \, dt.
\end{aligned}
\]

By summing these two parts and factoring out the common terms, we obtain the expression for the marginalized numerator:
\[
\binom{n-1}{m-1} \left ( F^{n-m}(z)(1-F(z))^{m-2} + (m-2) \int_0^{F(z)} t^{n-m} (1-t)^{m-3} \, dt \right ) f(z).
\]

\underline{Case 2:} when \( z \geq x_1 \). Similarly, if the ability of contestant $2$ is not the lowest among the other advancing contestants, we re-index the contestant corresponding to \( x^{(1)} \) to 3. We then classify the discussion based on the relationship between \( x^{(1)} \) and \( x_1 \), splitting the marginalization integral into two region:

The first region, when \( x^{(1)} \leq x_1 \):

In this case, the marginalization of the numerator becomes:
\[
\begin{aligned}
    & \binom{n-1}{m-1}(m-2)\int_0^{x_1} \left( \underbrace{\int_{x_3}^{1} \cdots \int_{x_3}^{1}}_{m-3} F^{n-m}(x_3) \prod_{k=2}^m f(x_k) \, d\mathbf{x_{-1,2,3}}\right) \, dx_3 \\
    = & \binom{n-1}{m-1}(m-2)f(z)\int_0^{F(x_1)} t^{n-m}(1-t)^{m-3} \, dt
\end{aligned}
\]

The second region, when \( x^{(1)} > x_1 \):

We further classify based on the relationship between \( x^{(1)} \) and \( z \). When \( x_1 \leq x^{(1)} < z \), the marginalization of the numerator is given by:
\[
\begin{aligned}
    & \binom{n-1}{m-1}(m-2)\int_{x_1}^{z} \left( \underbrace{\int_{x_3}^{1} \cdots \int_{x_3}^{1}}_{m-3} F^{n-m}(x_1) \prod_{k=2}^m f(x_k) \, d\mathbf{x_{-1,2,3}}\right) \, dx_3 \\
    = & \binom{n-1}{m-1}(m-2)f(z)F^{n-m}(x_1)\int_{x_1}^{z} (1-t)^{m-3}\, dt
\end{aligned}
\]

When \( z \) is the smallest ability among the remaining contestants, i.e., \( z = x^{(1)} \), the marginalization of the numerator becomes:
\[
\begin{aligned}
    & \binom{n-1}{m-1}\int_{\mathbf{x_{-1,2}} \geq z} F^{n-m}(x_1)  \prod_{k=2}^m f(x_k) \, d\mathbf{x_{-1,2}}\\
    = & \binom{n-1}{m-1}f(z)F^{n-m}(x_1)(1-F(z))^{m-2} \\
    = & \binom{n-1}{m-1}(m-2)f(z)F^{n-m}(x_1)\int_z^{1} (1-t)^{m-3} \, dt
\end{aligned}
\]

Adding the two components together, we obtain the marginalization of the second part's numerator: 
\[
\begin{aligned}
    & \binom{n-1}{m-1}(m-2)f(z)F^{n-m}(x_1)\int_0^{1} (1-t)^{m-3} \, dt \\
    = & \binom{n-1}{m-1} F^{n-m}(x_1)(1-F(x_1))^{m-2}
\end{aligned}
\]

In summary, we derive the posterior probability density function:
\[
\beta_1(z) =   
\begin{cases} 
\frac{\binom{n-1}{m-1}\left( F^{n-m}(z)(1-F(z))^{m-2}+(m-2)\int_0^{F(z)}t^{n-m}(1-t)^{m-3}\, dt \right)f(z)}{J(F,n,m,x_1)} & \text{if } z \leq x_1, \\
\frac{\binom{n-1}{m-1}\left(  F^{n-m}(x_1)(1-F(x_1))^{m-2} +(m-2)\int_0^{F(x_1)} t^{n-m}(1-t)^{m-3}\, dt \right) f(z)}{J(F,n,m,x_1)} & \text{if } z > x_1.
\end{cases}
\]

Since each term in the expression is continuous, it is easy to verify that \(\lim_{z \rightarrow x_i^{+}} \beta_1(z) = \beta_1(x_i)\), thus \(\beta_1\) is continuous.

We then compute the expression for \(\Pr_{\beta_i}\). When \(z \leq x_i\), the integral of the numerator of \(\beta_1(z)\) over \((0, z]\) is:
\[
\begin{aligned}
    & \binom{n-1}{m-1} \left[ \int_0^{z} F^{n-m}(t)(1-F(t))^{m-2} f(t)\, dt + (m-2) \int_0^{z} \int_0^{F(t)} p^{n-m}(1-p)^{m-3}\, dp \, f(t) \, dt \right] \\
    = & \binom{n-1}{m-1} \left[ \int_0^{F(z)} t^{n-m}(1-t)^{m-2} \, dt + (m-2) \int_0^{z} B_{F(t)}(n-m+1, m-2) \, dF(t) \right] \\
    = & \binom{n-1}{m-1} \left[ B_{F(z)}(n-m+1, m-1) + (m-2) \int_0^{z} B_{F(t)}(n-m+1, m-2) \, dF(t) \right],
\end{aligned}
\]
where the last equality follows from the definition of incomplete beta function \(B_x(a,b) = \int_0^x t^{a-1}(1-t)^{b-1} \, dt\) as defined in Lemma~\ref{lem:betaRep}.

When \(z > x_i\), only the last term \(f(z)\) in the expression for \(\beta_1(z)\) is related to \(z\), thus:
\[
\begin{aligned}
    \Pr_{\beta_1}(X_2 \leq z) & = \Pr_{\beta_1}( X_2 \leq x_1) + \int_{x_1}^z \beta_1(t) \, dt \\
    & = \Pr_{\beta_1}(X_2 \leq x_1) + \int_{x_1}^z f(t) \, dt \frac{\beta_1(z)}{f(z)} \\
    & = \Pr_{\beta_1}(X_2 \leq x_1) + (F(z) - F(x_1)) \frac{\beta_1(z)}{f(z)}.
\end{aligned}
\]

In summary, we obtain the expression for the posterior cumulative probability distribution:
\[
\Pr_{\beta_1}(X_2 \leq z) =   
\begin{cases} 
\frac{\binom{n-1}{m-1} \left[ B_{F(z)}(n-m+1, m-1) + (m-2) \int_{0}^{z} B_{F(t)}(n-m+1, m-2) \, dF(t) \right]}{J(F,n,m,x_1)} & \text{if } z \leq x_1, \\
\Pr_{\beta_1}(X_2 \leq x_1) + \frac{\binom{n-1}{m-1} (F(z) - F(x_1)) \left[ F^{n-m}(x_1)(1-F(x_1))^{m-2} + (m-2) B_{F(x_1)}(n-m+1, m-2) \right]}{J(F,n,m,x_1)} & \text{if } z > x_1.
\end{cases}
\]

Since \(\beta_1\) is continuous, its integral \(\Pr_{\beta_i}\) is also continuous, completing the proof.
\end{proof}

\begin{lemma}\label{lem:StoDom}
    For two continuous functions \( f(x), g(x) \geq 0 \), if the following conditions are satisfied:
    \begin{enumerate}
        \item \( g(x) \) is non-decreasing.
        \item \( \int_0^{+\infty} f(x) \, dx = \int_0^{+\infty} g(x) \, dx = \int_0^{+\infty} g(x) f(x) \, dx = 1 \),
    \end{enumerate}
    then for all \( a \geq 0 \), the inequality 
    \[
    \int_0^{a} f(x) \, dx \geq \int_0^a g(x) f(x) \, dx
    \]
    holds.
\end{lemma}

\begin{proof}
    Define 
    \[
    \varphi(a) = \int_0^a \big(1 - g(x)\big) f(x) \, dx.
    \]
    
    Differentiating \(\varphi(a)\), we have
    \[
    \varphi'(a) = \big(1 - g(a)\big) f(a).
    \]
    
    Since \( g(x) \) is non-decreasing and \(\int_0^{+\infty} g(x) \, dx = 1\), there exists a finite point \( a' \neq +\infty \) such that \(\varphi'(a) = 0\). Moreover, for \( a < a' \), \(\varphi'(a) > 0\), and for \( a > a' \), \(\varphi'(a) < 0\). This implies that \(\varphi(a)\) is a unimodal function, achieving its minimum only at the endpoints of the domain.

    Evaluating \(\varphi(a)\) at the endpoints, we have:
    \[
    \varphi(0) = 0,
    \]
    and
    \[
    \varphi(+\infty) = \int_0^{+\infty} \big(1 - g(x)\big) f(x) \, dx = \int_0^{+\infty} f(x) \, dx - \int_0^{+\infty} g(x) f(x) \, dx = 1 - 1 = 0.
    \]
    
    Therefore \(\varphi(a)\) is nonnegative for all \( a \geq 0 \), and we conclude that
    \[
    \int_0^a \big(1 - g(x)\big) f(x) \, dx \geq 0.
    \]
    
    This implies
    \[
    \int_0^a f(x) \, dx \geq \int_0^a g(x) f(x) \, dx,
    \]
    which completes the proof.
\end{proof}

\subsection*{Proof of Proposition~\ref{prop:stoDomPos}}
\begin{proof}
    Define the function:
    \[
    q_{x_i}(z) :=
    \begin{cases} 
    \frac{F^{n-m}(z)(1-F(z))^{m-2}+(m-2)\int_0^{F(z)}t^{n-m}(1-t)^{m-3}\, dt}{\binom{n-1}{m-1}^{-1}J(F,n,m,x_i)} & \text{if } z \leq x_i, \\
    \frac{F^{n-m}(x_i)(1-F(x_i))^{m-2} +(m-2)\int_0^{F(x_i)} t^{n-m}(1-t)^{m-3}\, dt}{\binom{n-1}{m-1}^{-1}J(F,n,m,x_i)} & \text{if } z > x_i.
    \end{cases}
    \]

    The posterior belief \(\beta_i\) can be rewritten as \(\beta_i(z) = q_{x_i}(z)f(z)\).

    To apply Lemma~\ref{lem:StoDom}, we will prove that \(q_{x_i}(z)\) is non-decreasing and \(\int_0^{+\infty}q_{x_i}(z)\,dz = 1\).

    First, consider monotonicity. When \(z > x_i\), \(q_{x_i}(z)\) is independent of \(z\), so \(q'_{x_i}(z) = 0\). When \(z \leq x_i\):
    \[
    \begin{aligned}
        q'_{x_i}(z) = & (n-m)F^{n-m-1}(z)(1-F(z))^{m-2}f(z) \\
                      & - (m-2)F^{n-m}(z)(1-F(z))^{m-3}f(z) + (m-2)F^{n-m}(z)(1-F(z))^{m-3}f(z) \\
                    = & (n-m)F^{n-m-1}(z)(1-F(z))^{m-2}f(z) \geq 0,
    \end{aligned}
    \]
    thus \(q_{x_i}'(z) \geq 0\) holds, $q_{x_i}(z)$ is non-decreasing, as desired.

    Next, consider the integral identity. Since the denominator of \(q_{x_i}(z)\) is independent of \(z\), this is equivalent to proving:
    \begin{equation}\label{eq:StoDomInt}
        \int_0^{+\infty}\binom{n-1}{m-1}^{-1}J(F,n,m,x_i)q_{x_i}(z)\,dz=\binom{n-1}{m-1}^{-1}J(F,n,m,x_i).
    \end{equation}

    Expanding the integral based on the piecewise definition of \(q_{x_i}(z)\):
    \[
    \begin{aligned}
        & \int_0^{x_i}\binom{n-1}{m-1}^{-1}J(F,n,m,x_i)q_{x_i}(z)\,dz + \int_{x_i}^{+\infty}\binom{n-1}{m-1}^{-1}J(F,n,m,x_i)q_{x_i}(z)\,dz \\
        = & \int_0^{F(x_i)}t^{n-m}(1-t)^{m-2}\, dt +(m-2)\int_0^{x_i}\int_0^{F(z)}t^{n-m}(1-t)^{m-3}\, dt\, dF(z) \\
        & + F^{n-m}(x_i)(1-F(x_i))^{m-2}+(m-2)(1-F(x_i))\int_0^{F(x_i)}t^{n-m}(1-t)^{m-3}\, dt.
    \end{aligned}
    \]

    Both sides of Equation~\eqref{eq:StoDomInt} are functions of \(x_i\). To prove these functions are equal on \([0,+\infty)\), differentiate both sides with respect to \(x_i\). On the left-hand side:
    \[
    \begin{aligned}
        & F^{n-m}(x_i)(1-F(x_i))^{m-2} + (m-2)\int_0^{F(x_i)}t^{n-m}(1-t)^{m-3}\, dt \\
        & + (n-m)F^{n-m-1}(x_i)(1-F(x_i))^{m-1} - (m-1)F^{n-m}(x_i)(1-F(x_i))^{m-2} \\
        & + (m-2)F^{n-m}(x_i)(1-F(x_i))^{m-2} -(m-2)\int_0^{F(x_i)}t^{n-m}(1-t)^{m-3}\, dt \\
        = & (n-m)F^{n-m-1}(x_i)(1-F(x_i))^{m-1}.
    \end{aligned}
    \]

    For the right-hand side, by definition:
    \[
    J(F, n, m, x_i) = \binom{n-1}{m-1}F^{n-m}(x_i)(1-F(x_i))^{m-1} + \binom{n-1}{m-1}(m-1)\int_0^{F(x_i)}t^{n-m}(1-t)^{m-2} \, dt.
    \]

    Differentiating with respect to \(x_i\) yields:
    \[
    \begin{aligned}
        & (n-m)F^{n-m-1}(x_i)(1-F(x_i))^{m-1} - (m-1)F^{n-m}(x_i)(1-F(x_i))^{m-2} \\
        & + (m-1)F^{n-m}(x_i)(1-F(x_i))^{m-2} \\
        = & (n-m)F^{n-m-1}(x_i)(1-F(x_i))^{m-1}.
    \end{aligned}
    \]

    Thus, for the Equation~\eqref{eq:StoDomInt}, \(\text{LHS}'_{x_i} = \text{RHS}'_{x_i}\). Additionally, at \(x_i = 0\), \(\text{LHS} = \text{RHS} = 0\). This two conditions gives that \(\text{LHS} = \text{RHS}\) holds, i.e., \(\int_0^{+\infty}q_{x_i}(z)\, dz = 1\).

    By Lemma~\ref{lem:StoDom}, we obtain, for all \(z \geq 0\):
    \[
    \begin{aligned}
        \int_0^z q_{x_i}(t)f(t) \,dt & \leq \int_0^{z}f(t)\, dt, \\
        \int_0^{z} \beta_i(t)\, dt & \leq F(z), \\
        \Pr_{\beta_i}(X_j \leq z) & \leq \Pr_{f}(X_j \leq z),
    \end{aligned}
    \]which completes the proof.
\end{proof}

\subsection*{Proof of Proposition~\ref{prop:StoDomAbi}}
\begin{proof}
    For the convenience of discussion, we define:
    \[
    \begin{aligned}
        I(F,n,m,x) :&= \binom{n-1}{m-1}^{-1}J(F,n,m,x) \\
        & = F^{n-m}(x)(1-F(x))^{m-1}+(m-1)\int_0^{F(x)}t^{n-m}(1-t)^{m-2}\, dt.
    \end{aligned}
    \]

    Then, $I(F,n,m,x)$'s derivative respect to $x$ is expressed as:
    \begin{equation}\label{eq:derivativeI}
    \begin{aligned}
        I'(F,n,m,x) = & (n-m)F^{n-m-1}(x)(1-F(x))^{m-1}f(x) \\
        & -(m-1)F^{n-m}(x)(1-F(x))^{m-2}f(x) \\
        & +(m-1)F^{n-m}(x)(1-F(x))^{m-2}f(x) \\
        = &(n-m)F^{n-m-1}(x)(1-F(x))^{m-1}f(x).
    \end{aligned}
    \end{equation}

    With the help of $I(F,n,m,x_j)$, the formula for $q_{x_j}(p)$ can be simplified to:
    \[
    q_{x_j}(p) =
    \begin{cases} 
    \frac{I(F,n-1,m-1,p)}{I(F,n,m,x_j)} & \text{if } p \leq x_j, \\
    \frac{I(F,n-1,m-1,x_j)}{I(F,n,m,x_j)} & \text{if } p > x_j.
    \end{cases}
    \]
    
    By definition, for contestant $j$ (so does contestant $i$):
    \[
    \Pr_{\beta_j}(X_k \leq z) = \int_0^{z}q_{x_j}(p)f(p)\, dp.
    \]

    Since $q_{x_j}(\cdot)$ is piece-wise, we classify the cases based on the value of $z$. 

    \underline{Case 1:} When $z \leq x_j$:

    In this case, only the denominator of $q_{x_j}(p)$ depends on $x_j$. Equation~\eqref{eq:derivativeI} shows that $I'(F,n,m,x)\geq 0$ for all $x\geq0$, $I(F,n,m,x)$ decreases with $x$. Therefore, $q_{x_j}(p) \geq q_{x_i}(p)$ holds for any $p \geq 0$, so we have $\int_0^{z}q_{x_i}(p)f(p)\, dp \leq \int_0^{z}q_{x_j}(p)f(p)\, dp$, i.e.,$\Pr_{\beta_i}(X_k \leq z) \leq \Pr_{\beta_j}(X_k \leq z)$, as desired. 

    \underline{Case 2:} When $z > x_j$:

    In this case, it suffices to prove that $\Pr_{\beta_j}(X_k \leq z)$ is an decreasing function of her ability $x_j$, which, by Leibniz integral rule, is to show that:
    \[
    \int_0^{z}\frac{\partial q_{x_j}(p)}{\partial x_j} f(p)\,dp \leq 0.
    \]

    We then proceed by dividing the integration interval into two parts, $(0,x_j]$ and $(x_j,z]$. Since after taking derivative, the denominator of $\partial q_{x_j}(p)/ \partial x_j$ is the same positive term for both interval and irrelevant to integration variable $z$, so we omit it in the discussion. 

    The first part, $\int_0^{x_j}\frac{\partial q_{x_j}(p)}{\partial x_j} f(p)\,dp$, the numerator of of $\partial q_{x_j}(p)/ \partial x_j$ is:
    \[
    \begin{aligned}
        & -I'(F,n,m,x_i)I(F,n-1,m-1,p) \\
        = & -(n-m)F^{n-m-1}(x_j)(1-F(x_j))^{m-1}f(x_j)\left [ F^{n-m}(p)(1-F(p))^{m-2}+\int_0^{F(p)}t^{n-m}(1-t)^{m-3}\,dt\right ]
    \end{aligned}
    \]

    The integration over $(0,x_j]$ then becomes:
    \begin{multline*}
        -(n-m)F^{n-m-1}(x_j)(1-F(x_j))^{m-1}f(x_j) \\
        \cdot \left [ \int_0^{F(x_j)}t^{n-m}(1-t)^{m-2}\, dt+ \int_0^{x} \int_0^{F(p)}t^{n-m}(1-t)^{m-3}\,dt f(p) \, dp\right ]
    \end{multline*}

    The second part, $\int_{x_j}^{z}\frac{\partial q_{x_j}(p)}{\partial x_j} f(p)\,dp$, the numerator of $\partial q_{x_j}(p)/ \partial x_j$ is: 
    \[
    \begin{aligned}
        & I'(F,n-1,m-1,x_j)I(F,n,m,x_j)-I'(F,n,m,x_j)-I'(F,n-1,m-1,x_j) \\
        = & (n-m)F^{n-m-1}(x_j)(1-F(x_j))^{m-2}f(x_j)\left [ F^{n-m}(x_j)(1-F(x_j))^{m-1}+\int_0^{F(x_j)}t^{n-m}(1-t)^{m-2}\,dt\right ] \\
        & -(n-m)F^{n-m-1}(x_j)(1-F(x_j))^{m-1}f(x_j)\left [ F^{n-m}(x_j)(1-F(x_j))^{m-2}+\int_0^{F(x_j)}t^{n-m}(1-t)^{m-3}\,dt\right ] \\
        = & (n-m)F^{n-m-1}(x_j)(1-F(x_j))^{m-2}f(x_j)
        \\
        & \cdot \left [\int_0^{F(x_j)}t^{n-m}(1-t)^{m-2}\,dt - (1-F(x_j))\int_0^{F(x_j)}t^{n-m}(1-t)^{m-3}\,dt\right ]
    \end{aligned}
    \]

    The integration over $(x_j,z]$ then becomes:
    \begin{multline*}
    (n-m)F^{n-m-1}(x_j)(1-F(x_j))^{m-2}f(x_j) \\
    \cdot \left [ (F(z)-F(x_j))\int_0^{F(x_j)}t^{n-m}(1-t)^{m-2}\,dt-(F(z)-F(x_j))(1-F(x_j))\int_0^{F(x_j)}t^{n-m}(1-t)^{m-3}\,dt\right]
    \end{multline*}

    For the terms inside the square brackets, we have:
    \[
    \begin{aligned}
        \leq & (F(z)-F(x_j))\int_0^{F(x_j)}t^{n-m}(1-t)^{m-2}\,dt \\
        \leq & (1-F(x_j))\int_0^{F(x_j)}t^{n-m}(1-t)^{m-2}\,dt
    \end{aligned}
    \]

    Combining two parts, we obtain that $I^2(F,n,m,x_j)\int_0^{z}\frac{\partial q_{x_j}(p)}{\partial x_j} f(p)\,dp$:
    \[
    \begin{aligned}
        \leq &  (n-m)F^{n-m-1}(x_j)(1-F(x_j))^{m-1}f(x_j) \\
        & \cdot \left[ \int_0^{F(x_j)}t^{n-m}(1-t)^{m-2}\,dt -\int_0^{F(x_j)}t^{n-m}(1-t)^{m-2}\, dt - \int_0^{x} \int_0^{F(p)}t^{n-m}(1-t)^{m-3}\,dt f(p)\, dt\right] \\
        = & (n-m)F^{n-m-1}(x_j)(1-F(x_j))^{m-1}f(x_j)\int_0^{x} \int_0^{F(p)}t^{n-m}(1-t)^{m-3}\,dt f(p)\, dt \leq 0.
    \end{aligned}
    \]

    Therefore, when $z > x_j$, $\int_0^{z}\frac{\partial q_{x_j}(p)}{\partial x_j} f(p)\,dp \leq 0$, holds for any ability value $x_j$, so we have $\int_0^{z}q_{x_i}(p)f(p)\, dp \leq \int_0^{z}q_{x_j}(p)f(p)\, dp$, i.e.,$\Pr_{\beta_i}(X_k \leq z) \leq \Pr_{\beta_j}(X_k \leq z)$, as desired. 

    Finally, after discussion by cases, we conclude that for all $z\geq0$, it holds that $\Pr_{\beta_i}(X_k \leq z) \leq \Pr_{\beta_j}(X_k \leq z)$, which completes the proof.
\end{proof}

\subsection*{Proof of Proposition~\ref{prop:ThreatenDesc}}
\begin{proof}
    It is suffice to prove $\Pr_{\beta_1}(X_{(l)} \leq x_1) \geq \Pr_{f}(X_{(l)}\leq x_1)$ for the first claim. We start by solving the expression for $\Pr_{\beta_1}(X_{(l)})$. 

    As stated in the proposition, the observer is re-indexed as $1$, we then appoint the top $l$ contestant (which induces a binomial factor $\binom{m-1}{l}$), re-index the $l^{\text{th}}$ strongest contestant as $2$ (which induces a coefficient $l$), and re-index the weakest contestant as $3$ (which induces a coefficient $m-l-1$). 
    
    Since $x^{(1)} \leq X_{(l)} \leq x_1$, therefore we have $\Pr_{\beta_1}(X_{(l)})$ equals: 
    \[
    \binom{m-1}{l}l(m-l-1)\int_0^{x_1}\underbrace{\int_{x_2}^1 \cdots \int_{x_2}^1}_{l-1} \int_0^{x_2} \underbrace{\int_{x_3}^{x_2} \cdots \int_{x_3}^{x_2}}_{m-l-2} \frac{\binom{n-1}{m-1}F^{n-m}(x_3)\prod_{i=2}^{m}f(x_i)}{J(F,n,m,x_1)} \, d\mathbf{x_{-1}}. 
    \]

    We then proceed by dealing with these integration from inside to outside. We omit coefficient and constant denominator $J(F,n,m,x_1)$ in the procedure.

    First, deal with last $m-l-1$ contestants:
    \[
    \begin{aligned}
    & \int_0^{F(x_2)} t^{n-m}(F(x_2)-t)^{m-l-2}\, dt \\
    = & \frac{(n-m)!(m-l-2)!}{(n-l-1)!} F(x_2)^{n-l-1},
    \end{aligned}
    \]where the equality gives by Lemma~\ref{lem:betaUpper}. 

    Then, the integration over top $l$ contestants gives:
    \[
    \begin{aligned}
        \int_0^{F(x_1)}t^{n-l-1}(1-t)^{l-1} \, dt 
    \end{aligned}
    \]

    Combining the coefficients, they becomes:
    \[
    \begin{aligned}
        & \frac{(m-1)!}{l!(m-l-1)!}l(m-l-1)\frac{(n-1)!}{(m-1)!(n-m)!}\frac{(n-m)!(m-l-2)!}{(n-l-1)!} \\
        = & \frac{(n-1)!}{(l-1)!(n-l-1)!} \\
        = & (n-1) \binom{n-2}{l-1}
    \end{aligned}
    \]

    Recovering the denominator, we get:
    \[
    \Pr_{\beta_1}(X_{(l)}) = \frac{(n-1)\binom{n-2}{l-1}\int_0^{F(x_1)} t^{n-l-1}(1-t)^{l-1}\, dt}{J(F,n,m,x_1)}.
    \]

    Specifically, when $m=n$, i.e., without shortlisting, $J(F,n,m,x_1)=1$, then $\Pr_{\beta_1}(X_{(l)})$ exactly correspond to the $l^{\text{th}}$ order-statistics (in our notation, the order is the inverse of the tradition) of the prior distribution, that is, $\Pr_{\beta_1(m)}(X_{(l)} \leq x_1) = \Pr_{f}(X_{(l)}\leq x_1)$ when $m=n$.

    Finally, Lemma~\ref{lem:normal} gives that $J(F,n,m,x_1)$ increases with $m$, then $\Pr_{\beta_1}(X_{(l)})$ decreases with $m$, equivalently: 
    \[
    \Pr_{\beta_1(m)}(X_{(l)} > x_1) \leq \Pr_{\beta_1(m')}(X_{(l)} > x_1) \leq \Pr_{f}(X_{(l)} > x_1),
    \]for all $l<m<m'\leq n$, which completes the proof.
\end{proof}

\begin{lemma}\label{lem:betaUpper}
For any $x > 0$, $a,b \in \mathbb{N}$, the following equality holds:
    \[
    \int_0^{x}t^{a-1}(x-t)^{b-1}\, dt = \frac{(a-1)!(b-1)!}{(a+b-1)!}x^{a+b-1}
    \]
\end{lemma}
\begin{proof}
    \[
    \int_0^{x}t^{a-1}(x-t)^{b-1}\, dt = x^{a+b-1} \int_0^x (\frac{t}{x})^{a-1}(\frac{t}{x})^{b-1}(x^{-1} \, dt ) = x^{a+b-1} \int_{0}^{1}t^{a-1}(1-t)^{b-1}\, dt
    \]
    
    Recall that for beta function $B(a,b)$, it holds that:
    \[
    B(a,b) = \int_{0}^{1}t^{a-1}(1-t)^{b-1}\, dt = \frac{\Gamma(a)\Gamma(b)}{\Gamma(a+b)}
    \]

    This gives:
    \[
    \int_0^{x}t^{a-1}(x-t)^{b-1}\, dt = \frac{\Gamma(a)\Gamma(b)}{\Gamma(a+b)} x^{a+b-1} = \frac{(a-1)!(b-1)!}{(a+b-1)!}x^{a+b-1}
    \]
    Where last equality holds since $\Gamma(m) = (m-1)!$ for any $m\in \mathbb{N}$, which completes the proof.
\end{proof}

\subsection*{Proof of Proposition~\ref{prop:winProb}}
\begin{proof}
To facilitate the discussion, we re-index the admitted contestants as \(\mathcal{I} = [m]\), and we designate contestant \(i\) with \(1\) in the new index. Since there is a one-to-one mapping between $e_i$ and $\gamma_i$, choosing an effort \(e_i\) is equivalent to reporting an ability \(\gamma_i\). In the following discussion, we focus on the latter interpretation, which simplifies notations.

Contestant \(i\) (with the original index) sees her probability of achieving rank \(l \neq m\) is:
\[
\begin{aligned}
& \binom{m-1}{l-1} \Pr\left [ \bigwedge_{j=2}^{l}X_j > \gamma_i, \bigwedge_{k=l+1}^m X_k < \gamma_i \mid x_i, \mathcal{I} \right ] \\
= & \binom{m-1}{l-1} \underbrace{\int_{\gamma_i}^{1} \cdots \int_{\gamma_i}^{1}}_{l-1} \underbrace{\int_{0}^{\gamma_i} \cdots \int_{0}^{\gamma_i}}_{m-l} \beta_i(\mathbf{x}) \, d\mathbf{x_{-i}}
\end{aligned}
\]

At this point, \(x^{(1)}\) appears among the last \(m-l\) contestants. Since the reported ability \(\gamma_i\) affects the expression of the integrand \(\beta_i(\mathbf{x})\), we classify the cases based on the value of \(\gamma_i\):

\underline{Case 1:} The reported ability is lower than the true ability, i.e., \(\gamma_i \leq x_i\). Here, \(x^{(1)} \leq x_i\), so:
\[
\binom{m-1}{l-1} \underbrace{\int_{\gamma_i}^{1} \cdots \int_{\gamma_i}^{1}}_{l-1} \underbrace{\int_{0}^{\gamma_i} \cdots \int_{0}^{\gamma_i}}_{m-l} \frac{\binom{n-1}{m-1}F^{n-m}(x^{(1)})\prod_{k=2}^{m}f(x_k)}{J(F,n,m,x_i)} \, d\mathbf{x_{-i}}
\]

Without loss of generality, assume contestant \(m\) has the lowest ability among the admitted contestants:
\[
\begin{aligned}
& \binom{m-1}{l-1} (m-l) \int_{0}^{\gamma_i} \underbrace{\int_{\gamma_i}^{1} \cdots \int_{\gamma_i}^{1}}_{l-1} \underbrace{\int_{0}^{\gamma_i} \cdots \int_{0}^{\gamma_i}}_{m-l-1} \frac{\binom{n-1}{m-1}F^{n-m}(x^{(1)})\prod_{k=2}^{m}f(x_k)}{J(F,n,m,x_i)} \, d\mathbf{x_{-i,m}} \, dx_m \\
= & \binom{m-1}{l-1} (m-l) [1-F(\gamma_i)]^{l-1} \int_{0}^{\gamma_i} \underbrace{\int_{0}^{\gamma_i} \cdots \int_{0}^{\gamma_i}}_{m-l-1} \frac{\binom{n-1}{m-1}F^{n-m}(x_m)\prod_{k=l+1}^{m}f(x_k)}{J(F,n,m,x_i)} \, d\mathbf{x_{-[l]\cup{\{m\}}}} \, dx_m \\
= & \binom{m-1}{l-1} (m-l) [1-F(\gamma_i)]^{l-1} \int_{0}^{\gamma_i} \frac{\binom{n-1}{m-1}F^{n-m}(x_m)[F(\gamma_i)-F(x_m)]^{m-l-1}f(x_m)}{J(F,n,m,x_i)} \, dx_m \\
\end{aligned}
\]

Extracting terms independent of the integral, we obtain:
\[
\frac{\binom{n-1}{m-1} \binom{m-1}{l-1} (m-l) [1-F(\gamma_i)]^{l-1} \int_{0}^{F(\gamma_i)} t^{n-m}[F(\gamma_i)-t]^{m-l-1} \, dt}{J(F,n,m,x_i)}
\]

Substituting the equality from Lemma~\ref{lem:betaUpper}:
\[
\begin{aligned}
& \frac{\binom{n-1}{m-1} \binom{m-1}{l-1} (m-l) [1-F(\gamma_i)]^{l-1} \frac{(n-m)!(m-l-1)!}{(n-l)!}F(\gamma_i)^{n-l}}{J(F,n,m,x_i)} \\
= & \frac{\frac{(n-1)!}{(n-m)!(m-1)!} \frac{(m-1)!}{(l-1)!(m-l)!} [1-F(\gamma_i)]^{l-1} \frac{(n-m)!(m-l)!}{(n-l)!}F(\gamma_i)^{n-l}}{J(F,n,m,x_i)} \\
= & \frac{\frac{(n-1)!}{(l-1)!(n-l)!} [1-F(\gamma_i)]^{l-1}F(\gamma_i)^{n-l}}{J(F,n,m,x_i)} \\
= & \frac{\binom{n-1}{l-1}F(\gamma_i)^{n-l}(1-F(\gamma_i))^{l-1}}{J(F,n,m,x_i)}
\end{aligned}
\]

\underline{Case 2:} The reported ability value by the contestant is higher than the actual ability value, i.e., \(\gamma_i > x_i\). We need to decompose the original integral into two additive parts by considering the relationship between \(x_i\) and \(x^{(1)}\):

\begin{enumerate}
\item \underline{First Part:} When \(x^{(1)} > x_i\), the integral over this interval is:
   \[
   \begin{aligned}
       & \underbrace{\int_{\gamma_i}^{1} \cdots \int_{\gamma_i}^{1}}_{l-1} \underbrace{\int_{x_i}^{\gamma_i} \cdots \int_{x_i}^{\gamma_i}}_{m-l} \frac{\binom{n-1}{m-1}F^{n-m}(x_i)\prod_{k=2}^{m}f(x_k)}{J(F,n,m,x_i)} \, d\mathbf{x_{-i}} \\
       = & \frac{\binom{n-1}{m-1}[1-F(\gamma_i)]^{l-1}F^{n-m}(x_i)[F(\gamma_i)-F(x_i)]^{m-l}}{J(F,n,m,x_i)}
   \end{aligned}
   \]
\item \underline{Second Part:} When \(x^{(1)} \leq x_i\), similarly, assuming the \(m\)-th contestant has the lowest ability value among those advancing, the integral over this interval can be expressed as:
   \[
   \begin{aligned}
       & (m-l) \int_{0}^{x_i} \underbrace{\int_{x_m}^{\gamma_i} \cdots \int_{x_m}^{\gamma_i}}_{m-l-1} \underbrace{\int_{\gamma_i}^1 \cdots \int_{\gamma_i}^1}_{l-1} \frac{\binom{n-1}{m-1} F^{n-m}(x^{(1)}) \prod_{k=2}^m f(x_k)}{J(F,n,m,x_i)} \, d\mathbf{x_{-i}} \, dx_m \\
       = & (m-l) [1-F(\gamma_i)]^{l-1} \int_{0}^{x_i} \underbrace{\int_{x_m}^{\gamma_i} \cdots \int_{x_m}^{\gamma_i}}_{m-l-1} \frac{\binom{n-1}{m-1} F^{n-m}(x_m) \prod_{k=l+1}^m f(x_k)}{J(F,n,m,x_i)} \, d\mathbf{x_{-[l]\cup{\{m\}}}} \, dx_m \\
       = & (m-l) [1-F(\gamma_i)]^{l-1} \int_{0}^{x_i} \frac{\binom{n-1}{m-1} F^{n-m}(x_m) [F(\gamma_i) - F(x_m)]^{m-l-1}f(x_m)}{J(F,n,m,x_i)} \, dx_m 
   \end{aligned}
   \]
   
   Extracting the terms independent of the integral, we obtain:
   \[
   \frac{\binom{n-1}{m-1}(m-l)[1-F(\gamma_i)]^{l-1} \int_{0}^{F(x_i)} t^{n-m}[F(\gamma_i)-t]^{m-l-1}\, dt}{J(F,n,m,x_i)} 
   \]
\end{enumerate}

By adding the two parts together, incorporating the binomial coefficient, and factoring out common terms, we derive the final expression:\[
\frac{\binom{n-1}{m-1}\binom{m-1}{l-1}(1-F(\gamma_i))^{l-1}\left[F^{n-m}(x_i)(F(\gamma_i)-F(x_i))^{m-l}+(m-l)\int_{0}^{F(x_i)}t^{n-m}(F(\gamma_i)-t)^{m-l-1}\, dt\right]}{J(F,n,m,x_i)}
\]

contestant \( i \) sees her probability of achieving rank \( l = m \) is given by:
\[
\Pr\left [ \bigwedge_{j=2}^{m} X_j > \gamma_i \mid x_i, \mathcal{I} \right ]
\]

We consider different cases based on the value of \(\gamma_i\):

\underline{Case 1:} When \(\gamma_i > x_i\), we have \(x^{(1)} > x_i\). The expression becomes:
\[
\begin{aligned}
    & \underbrace{\int_{\gamma_i}^{1} \cdots \int_{\gamma_i}^{1}}_{m-1} \frac{\binom{n-1}{m-1} F^{n-m}(x_i) \prod_{k=2}^{m} f(x_k)}{J(F,n,m,x_i)} \, d\mathbf{x_{-i}} \\
    = & \binom{n-1}{m-1} F^{n-m}(x_i) [1-F(\gamma_i)]^{m-1}
\end{aligned}
\]

\underline{Case 2:} When \(\gamma_i \leq x_i\), we decompose the integration space into two parts:

\begin{enumerate}
\item \underline{First Part:} When \(x^{(1)} > x_i\), the integral is:
\[
\begin{aligned}
    & \underbrace{\int_{x_i}^{1} \cdots \int_{x_i}^{1}}_{m-1} \frac{\binom{n-1}{m-1} F^{n-m}(x_i) \prod_{k=2}^{m} f(x_k)}{J(F,n,m,x_i)} \, d\mathbf{x_{-i}} \\
    = & \binom{n-1}{m-1} F^{n-m}(x_i) [1-F(x_i)]^{m-1}
\end{aligned}
\]
\item \underline{Second Part:} When \(x^{(1)} \leq x_i\), the integral is:
\[
\begin{aligned}
    & (m-1) \int_{\gamma_i}^{x_i} \underbrace{\int_{x_2}^{1} \cdots \int_{x_2}^{1}}_{m-2} \frac{\binom{n-1}{m-1} F^{n-m}(x_i) \prod_{k=2}^{m} f(x_k)}{J(F,n,m,x_i)} \, d\mathbf{x_{-[l]\cup{\{m\}}}} \, dx_m \\
    = & \binom{n-1}{m-1}(m-1) \int_{F(\gamma_i)}^{F(x_i)} t^{n-m} (1-t)^{m-2} \, dt
\end{aligned}
\]
\end{enumerate}

Combining both parts and factoring out common terms, we obtain the final expression:
\[
\frac{\binom{n-1}{m-1} \left [ F^{n-m}(x_i) (1-F(x_i))^{m-1} + (m-1) \int_{F(\gamma_i)}^{F(x_i)} t^{n-m} (1-t)^{m-2} \, dt \right ]}{J(F,n,m,x_i)}
\]

Based on the above discussion, we have derived the expression for \(P_{(i,l)}(\gamma_i \mid x_i)\). Next, we will verify the continuity of \(P_{(i,l)}\). It suffices to prove that the function is right-continuous at \(x_i\).

We then compute the right limit of $P_{(i,l)}$ at \(x_i\). Since the denominator is independent of \(\gamma_i\), we focus on the limit of the numerator:

When \(l = m\), the numerator is:
\[
\binom{n-1}{m-1}F^{n-m}(x_i)(1-F(x_i))^{m-1}
\]

When \(l < m\), the numerator is:
\[
\binom{n-1}{m-1}\binom{m-1}{l-1}(1-F(x_i))^{l-1}(m-l)\int_{0}^{F(x_i)}t^{n-m}(F(x_i)-t)^{m-l-1}\, dt
\]

Substituting the equality from Lemma~\ref{lem:betaUpper}:
\[
\begin{aligned}
    & \binom{n-1}{m-1}\binom{m-1}{l-1}(1-F(x_i))^{l-1}(m-l)\frac{(n-m)!(m-l-1)!}{(n-l)!}F(x_i)^{n-l} \\ 
    = & \frac{(n-1)!}{(n-m)!(m-1)!} \frac{(m-1)!}{(l-1)!(m-l)!}(1-F(x_i))^{l-1}\frac{(n-m)!(m-l)!}{(n-l)!}F(x_i)^{n-l} \\
    = & \binom{n-1}{l-1}(1-F(x_i))^{l-1}F(x_i)^{n-l} 
\end{aligned}
\]

In summary, we obtain:
\[
\begin{aligned}
\lim_{\gamma_i \rightarrow x_i^{+}}P_{(i,l)}(\gamma_i \mid x_i)  & = \begin{cases} 
\frac{\binom{n-1}{l-1}(1-F(x_i))^{l-1}F(x_i)^{n-l} }{J(F,n,m,x_i)} & \text{if } l < m, \\
\frac{\binom{n-1}{m-1}F^{n-m}(x_i)(1-F(x_i))^{m-1}}{J(F,n,m,x_i)} & \text{if } l = m.
\end{cases} \\
& = P_{(i,l)}(x_i \mid x_i) 
\end{aligned}
\]

Therefore, \(P_{(i,l)}(\gamma_i \mid x_i)\) is continuous throughout \(\gamma_i \geq 0\), thus completing the proof.
\end{proof}

\begin{corollary}[Derivative of Winning Probability]\label{coro:DerivProb}
When admitted contestant $i$ reports an ability $\gamma_i$ given her ability $x_i$, the derivative function of her perceived probability of getting rank $l$ respect to $\gamma_i$, denoted as $P'_{(i,l)}(\gamma_i \mid x_i)$, exists. Specifically, the derivative at $x_i$ can be expressed as: 
\[
P'_{(i,l)}(x_i\mid X_i=x_i) =   
\begin{cases} 
\frac{\binom{n-1}{l-1}\left [ (n-l)(1-F(x_i))-(l-1)F(x_i)\right ]F(x_i)^{n-l-1}(1-F(x_i))^{l-2}f(x_i)}{J(F,n,m,x_i)} & \text{if } l < m, \\
- \frac{\binom{n-1}{m-1}(m-1)F^{n-m}(x_i)(1-F(x_i))^{m-2}f(x_i)}{J(F,n,m,x_i)} & \text{if } l = m.
\end{cases}
\]
We denote $P'_{(i,l)}$ as a shorthand for $P'_{(i,l)}(x_i\mid X_i=x_i)$. 
\end{corollary}
\begin{proof}[Proof of Corollary~\ref{coro:DerivProb}]
From Proposition~\ref{prop:winProb}, it is easy to see \( P_{(i,l)} \) has derivatives when \(\gamma_i \neq x_i\). We only need to prove the existence of the derivative at \(\gamma_i = x_i\), which means showing that the left-hand and right-hand derivatives at \(\gamma_i = x_i\) are equal. Since the denominator \( J(F,n,m,x_i) \) is independent of \(\gamma_i\), we omit it in our discussion. We will consider the cases \( l < m \) and \( l = m \) separately.

\underline{Case 1:} \( l < m \) :

When \(\gamma_i \leq x_i\), the derivative of the numerator is:
\[
\begin{aligned}
    & \left[ \binom{n-1}{l-1} F(\gamma_i)^{n-l} (1-F(\gamma_i))^{l-1} \right]' \\
    = & \binom{n-1}{l-1} \left[ (n-l)(1-F(\gamma_i)) - (l-1)F(\gamma_i) \right] F(\gamma_i)^{n-l-1} (1-F(\gamma_i))^{l-2} f(\gamma_i)
\end{aligned}
\]

Substituting \(\gamma_i = x_i\), the numerator of the left-hand derivative \( P'^{-}_{(i,l)}(x_i \mid x_i) \) is:
\[
\binom{n-1}{l-1} \left[ (n-l)(1-F(x_i)) - (l-1)F(x_i) \right] F(x_i)^{n-l-1} (1-F(x_i))^{l-2} f(x_i)
\]

When \(\gamma_i > x_i\), the numerator of the expression for \( P_{(i,l)} \) expands to:
\[
\binom{n-1}{m-1} \binom{m-1}{l-1} F^{n-m}(x_i) (1-F(\gamma_i))^{m-1} + \binom{n-1}{m-1} \binom{m-1}{l-1} (m-l) \int_{0}^{F(x_i)} t^{n-m} (F(\gamma_i)-t)^{m-l-1} \, dt
\]

We handle each part separately.
After differentiating the first part, we obtain:
\begin{multline*}
    \binom{n-1}{m-1} \binom{m-1}{l-1} \left[ (m-l)(1-F(\gamma_i)) - (l-1)(F(\gamma_i)-F(x_i)) \right] \\
    \cdot F^{n-m}(x_i) (1-F(\gamma_i))^{l-2} (F(\gamma_i) - F(x_i))^{m-l-1} f(\gamma_i)
\end{multline*}

This expression is continuous, so the right-hand limit at \(\gamma_i = x_i\) can be obtained by direct substitution. When \( l < m-1 \), \((F(\gamma_i) - F(x_i))^{m-l-1} = 0\), making this part's right-hand limit zero. When \( l = m-1 \), \((F(\gamma_i) - F(x_i))^{m-l-1} = 1\), and the right-hand limit then becomes:
\begin{equation} \label{eq:derEq1}
    \binom{n-1}{m-1} \binom{m-1}{l-1} F^{n-l-1}(x_i) (1-F(x_i))^{l-1} f(x_i)
\end{equation}

In the second part, we extract the term \( F(\gamma_i) \) from the integrand in the last term:
\[
\begin{aligned}
    \int_{0}^{F(x_i)} t^{n-m} (F(\gamma_i) - t)^{m-l-1} \, dt &= F^{n-l}(\gamma_i) \int_{0}^{F(x_i)} \left(\frac{t}{F(\gamma_i)}\right)^{n-m} \left(1 - \frac{t}{F(\gamma_i)}\right)^{m-l-1} \left(F(\gamma_i)^{-1} dt\right) \\
    &= F^{n-l}(\gamma_i) \int_0^{\frac{F(x_i)}{F(\gamma_i)}} t^{n-m} (1-t)^{m-l-1} \, dt
\end{aligned}
\]

Thus, the expression of this part becomes:
\[
\binom{n-1}{m-1} \binom{m-1}{l-1} (m-l) F^{n-l}(\gamma_i) (1-F(\gamma_i))^{l-1} \int_0^{\frac{F(x_i)}{F(\gamma_i)}} t^{n-m} (1-t)^{m-l-1} \, dt
\]

This expression includes the product of three functions of \(\gamma_i\). After differentiation, it turns into the sum of three components:

The first component involves \([F^{n-l}(\gamma_i)]'\):
\[
\binom{n-1}{m-1} \binom{m-1}{l-1} (n-l) (m-l) F^{n-l-1}(\gamma_i) (1-F(\gamma_i))^{l-1} f(\gamma_i) \int_0^{\frac{F(x_i)}{F(\gamma_i)}} t^{n-m} (1-t)^{m-l-1} \, dt
\]

The limit of the last term at $x_i$ is \(B(n-m+1, m-l)\). Substituting into the Lemma~\ref{lem:betaUpper} and combining other coefficients, we have:
\[
\begin{aligned}
    & \frac{(n-1)!}{(m-1)!(n-m)!} \frac{(m-1)!}{(l-1)!(m-l)!} (n-l) (m-l) \frac{(n-m)!(m-l-1)!}{(n-l)!} \\
    = & \frac{(n-1)!}{(l-1)!(n-l)!} (n-l) = \binom{n-1}{l-1}(n-l)
\end{aligned}
\]

Thus, the right-hand limit of this component is:
\[
\binom{n-1}{l-1} (n-l) F^{n-l-1}(x_i) (1-F(x_i))^{l-1} f(x_i)
\]

The second component involves \([(1-F(\gamma_i))^{l-1}]'\):
\[
-\binom{n-1}{m-1} \binom{m-1}{l-1} (m-l)(l-1) F^{n-l}(\gamma_i) (1-F(\gamma_i))^{l-2} f(\gamma_i) \int_0^{\frac{F(x_i)}{F(\gamma_i)}} t^{n-m} (1-t)^{m-l-1} \, dt
\]

Similarly, substituting the value of \(B(n-m+1, m-l)\) and simplifying, the right-hand limit of this component becomes:
\[
-\binom{n-1}{l-1} (l-1) F^{n-l}(x_i) (1-F(x_i))^{l-2} f(x_i)
\]

The third component involves the expression:
\[
\left[ \int_0^{\frac{F(x_i)}{F(\gamma_i)}} t^{n-m}(1-t)^{m-l-1} \, dt \right]'.
\]

By applying the differentiation rule for integrals with variable upper limits, we obtain:
\[
\begin{aligned}
    \left[ \int_0^{\frac{F(x_i)}{F(\gamma_i)}} t^{n-m}(1-t)^{m-l-1} \, dt \right]' 
    &= - \left( \frac{F(x_i)}{F(\gamma_i)} \right)^{n-m} \left( 1-\frac{F(x_i)}{F(\gamma_i)} \right)^{m-l-1} \frac{F(x_i)}{F^2(\gamma_i)} f(\gamma_i).
\end{aligned}
\]

Now consider the behavior of this term under different $l$:
\begin{enumerate}
\item Case \( l < m-1 \):  
   In this case, we have:
   \[
   \lim_{\gamma_i \to x_i^+} \left( 1-\frac{F(x_i)}{F(\gamma_i)} \right)^{m-l-1} = 0.
   \]
   Consequently, the right-hand limit of this component is \( 0 \).
\item Case \( l = m-1 \):
   Here, the factor \( \left( 1-\frac{F(x_i)}{F(\gamma_i)} \right)^{m-l-1} \) simplifies to \( 1 \). Thus, the right-hand limit of the derivative with respect to the upper limit becomes:
   \[
   -\frac{f(x_i)}{F(x_i)}.
   \]
   The corresponding expression for this component's right-hand limit is:
   \begin{equation} \label{eq:derEq2}
   -\binom{n-1}{m-1} \binom{m-1}{l-1} F^{n-l-1}(x_i) (1-F(x_i))^{l-1} f(x_i).
   \end{equation}
\end{enumerate}

In both cases (\( l < m-1 \) and \( l = m-1 \)), the sum of the two parts from Equation~\eqref{eq:derEq1} and Equation~\eqref{eq:derEq2} equals \( 0 \). Therefore, no additional classification of \( l \) is required.

Finally, combining the results from all parts, we have the right-hand limit of the numerator:
\[
\binom{n-1}{l-1} \left[ (n-l)(1-F(\gamma_i)) - (l-1)F(\gamma_i) \right] F(\gamma_i)^{n-l-1}(1-F(\gamma_i))^{l-2} f(\gamma_i).
\]

Thus, we conclude that for \( l < m \), the following equality holds, as desired:
\[
P'^{-}_{(i,l)}(x_i \mid x_i) = P'^{+}_{(i,l)}(x_i \mid x_i)
\]

\underline{Case 2:} \( l = m \):

If \( \gamma_i \leq x_i \), the derivative's numerator is:
\[
\begin{aligned}
    & \binom{n-1}{m-1} \left[ F^{n-m}(x_i)(1-F(x_i))^{m-1} + (m-1) \int_{F(\gamma_i)}^{F(x_i)} t^{n-m}(1-t)^{m-2} \, dt \right]' \\ 
    = & \binom{n-1}{m-1}(m-1) \left[ \int_{0}^{F(x_i)} t^{n-m}(1-t)^{m-2} \, dt - \int_{0}^{F(\gamma_i)} t^{n-m}(1-t)^{m-2} \, dt \right]' \\ 
    = & -\binom{n-1}{m-1}(m-1) \left[ \int_{0}^{F(\gamma_i)} t^{n-m}(1-t)^{m-2} \, dt \right]' \\
    = & -\binom{n-1}{m-1}(m-1) F^{n-m}(\gamma_i)(1-F(\gamma_i))^{m-2} f(\gamma_i)
\end{aligned}
\]

Substituting \(\gamma_i = x_i\) gives the numerator of left-hand derivative \( P'^{-}_{(i,l)}(x_i \mid x_i) \) as:
\[
-\binom{n-1}{m-1}(m-1) F^{n-m}(x_i)(1-F(x_i))^{m-2} f(x_i)
\]

If \( \gamma_i > x_i \), the derivative's numerator is:
\[
\begin{aligned}
    & \left[ \binom{n-1}{m-1} F^{n-m}(x_i)(1-F(\gamma_i))^{m-1} \right]' \\
    = & -\binom{n-1}{m-1}(m-1) F^{n-m}(x_i)(1-F(\gamma_i))^{m-2} f(\gamma_i)
\end{aligned}
\]

Taking the right-hand limit at $x_i$, we obtain:
\[
-\binom{n-1}{m-1}(m-1) F^{n-m}(x_i)(1-F(x_i))^{m-2} f(x_i)
\]

Thus, we have shown that when \( l = m \), \( P'^{-}_{(i,l)}(x_i \mid x_i) = P'^{+}_{(i,l)}(x_i \mid x_i) \).

In conclusion, for any \( l \leq m \), \( P_{(i,l)}(\gamma_i \mid x_i) \) is differentiable everywhere. Therefore, the derivative \( P'_{(i,l)}(\gamma_i \mid x_i) \) exists and at \( \gamma_i = x_i \), and it is given by:
\[
P'_{(i,l)}(x_i \mid X_i = x_i) =   
\begin{cases} 
\frac{\binom{n-1}{l-1} \left[ (n-l)(1-F(x_i))-(l-1)F(x_i) \right] F(x_i)^{n-l-1}(1-F(x_i))^{l-2} f(x_i)}{J(F,n,m,x_i)} & \text{if } l < m, \\
- \frac{\binom{n-1}{m-1}(m-1) F^{n-m}(x_i)(1-F(x_i))^{m-2} f(x_i)}{J(F,n,m,x_i)} & \text{if } l = m.
\end{cases}
\]
This completes the proof.
\end{proof}

\subsection*{Proof of Theorem~\ref{thm:contestantSBNE}}
\begin{proof}
Under a symmetric Bayesian Nash Equilibrium, each individual's effort \( b(x_i) \) is her best response given her ability \( x_i \), i.e., for all \( i \in [n] \):
\[
b(x_i) \in \mathop{\arg \max}_{e_i} u_i := \sum_{l=1}^{m} V_l P_{(i,l)} - \frac{g(e_i)}{x_i}
\]

Since \( e_i \) and \( \gamma_i \) are in one-to-one correspondence, the expression can be rewritten as:
\[
b(x_i) \in \left\{ b(\gamma_i) \mid \mathop{\arg \max}_{\gamma_i} \sum_{l=1}^{m} V_l P_{(i,l)} - \frac{g(b(\gamma_i))}{x_i} \right\}
\]

Taking the derivative of \( u_i \) with respect to \( \gamma_i \) gives the first-order condition:
\[
\sum_{l=1}^{m} V_l P'_{(i,l)}(\gamma_i \mid X_i = x_i) = \frac{g'(b(\gamma_i)) b'(\gamma_i)}{x_i}
\]

In equilibrium, everyone exerts effort according to \( b(x_i) \), so \( \gamma_i = \gamma(e_i) = \gamma(b(x_i)) = x_i \). Substituting this into the equation and moving the denominator on the right to the left side:
\[
\sum_{l=1}^{m} V_l P'_{(i,l)}(x_i \mid X_i = x_i)x_i = g'(b(x_i)) b'(x_i)
\]

Since the equilibrium condition holds for any realization of ability \( x_i > 0 \), we replace \( x_i \) with a variable \( t \) and integrate both sides over \((0, x_i]\):
\[
\begin{aligned}
\int_{0}^{x_i} \sum_{l=1}^{m} V_l P'_{(i,l)}(t \mid X_i = t)t \, dt &  = \int_{0}^{x_i} g'(b(t)) b'(t) \, dt \\
 & = g(b(x_i)) 
\end{aligned}
\]

Applying the inverse of the cost function \( g \): 
\[
b(x_i) = g^{-1}\left( \int_{0}^{x_i} \sum_{l=1}^{m} V_l P'_{(i,l)}(t \mid X_i = t)t \, dt \right)
\]

By plugging the expression for $P'_{(i,j)}$ into the summation \(\sum_{l=1}^{m}V_lP'_{(i,l)}(t|X_i=t)\):
\[
\begin{aligned} 
    & \frac{\sum_{l=1}^{m-1}\binom{n-1}{n-l}(n-l)F^{n-l-1}(t)(1-F(t))^{l-1}V_l - \sum_{l=1}^m\binom{n-1}{l-1}(l-1)F^{n-l}(t)(1-F(t))^{l-2}V_l}{J(F,n,m,t)f^{-1}(t)} \\
    = & \frac{\sum_{l=1}^{m-1}\binom{n-1}{n-l}(n-l)F^{n-l-1}(t)(1-F(t))^{l-1}V_l - \sum_{l=2}^m\binom{n-1}{l-1}(l-1)F^{n-l}(t)(1-F(t))^{l-2}V_l}{J(F,n,m,t)f^{-1}(t)}
\end{aligned}
\]

Combining like terms, we get:
\[
\begin{aligned}
& \frac{\sum_{l=1}^{m-1}\left [ \binom{n-1}{n-l}(n-l)V_l - \binom{n-1}{l}lV_{l+1}\right ]F^{n-l-1}(t)(1-F(t))^{l-1}f(t)}{J(F,n,m,t)} \\
= & \frac{\sum_{l=1}^{m-1}\binom{n-1}{l-1}(n-l)(V_l-V_{l+1})F^{n-l-1}(t)(1-F(t))^{l-1}f(t)}{J(F,n,m,t)}
\end{aligned}
\]

Therefore, the equilibrium effort must satisfy:
\begin{equation}\label{eq:EquEffort}
        b(x_i) = g^{-1}\left(\int_{0}^{x_i}\frac{\sum_{l=1}^{m-1}\binom{n-1}{l-1}(n-l)(V_l-V_{l+1})F^{n-l-1}(t)(1-F(t))^{l-1}f(t)}{J(F,n,m,t)} t\, dt \right)
\end{equation}

Since we consider rank-order prize structures, i.e., \( V_l \geq V_{l+1}\), the integrand is almost strictly increasing everywhere. Therefore, non-negativity and monotonicity of \( b(x_i) \) are guaranteed. Therefore, we obtain the unique symmetric Bayesian Nash Equilibrium, which completes the proof.
\end{proof}

\section{Missing Proofs in Section \ref{sec:optimal design}}

\subsection*{Proof of Corollary~\ref{coro:Consolation}}
\begin{proof}
    The contribution of \( V_m \) to the contestant $i$'s equilibrium effort is given by:
    \[
    V_m \int_{0}^{x_i} t P'_{(i,m)}(t \mid X_i = t) \, dt
    \]
    
    For \( \forall t > 0 \), we have:
    \[
    P'_{(i,m)}(t \mid X_i = t) = - \frac{\binom{n-1}{m-1}(m-1)F^{n-m}(t)(1-F(t))^{m-2}f(t)}{J(F,n,m,t)} \leq 0
    \]
    
    Furthermore, for a non-zero measure subset of the domain, it holds that \( P'_{(i,m)}(t \mid X_i = t) < 0 \).

    Therefore,
    \begin{align*}
    b(x_i) & = g^{-1} \left( \sum_{l=1}^{m} V_l \int_{0}^{x_i} t P'_{(i,l)}(t \mid X_i = t) \, dt \right) \\
           & < g^{-1} \left( \sum_{l=1}^{m-1} V_l \int_{0}^{x_i} t P'_{(i,l)}(t \mid X_i = t) \, dt \right)
    \end{align*}
    Setting \( V_m = 0 \) will increase the effort for each admitted contestant, which completes the proof. 
\end{proof}

\subsection*{Proof of Corollary~\ref{coro:EmptyPrize}}
\begin{proof}
    For \( m > k+1 \), by Lemma~\ref{lem:normal}, we have:
    \[
    0 < \frac{J(F,n,k+1,t)}{J(F,n,m,t)} < 1
    \]

    From Equation~\eqref{eq:EquEffort} and the rank-order prize structure, we know that \(\sum_{l=1}^{m} V_l P'_{(i,l)}(t \mid X_i = t) \geq 0\). Additionally, for \( l \geq k+1 \), \( V_l = 0 \). Therefore, we have:
    \begin{align*}
        \sum_{l=1}^{k+1} V_l P'_{(i,l)}(t \mid X_i = t, k+1) &= \sum_{l=1}^{k} V_l P'_{(i,l)}(t \mid X_i = t, k+1) \\
        &= \sum_{l=1}^{k} V_l P'_{(i,l)}(t \mid X_i = t, m) \frac{J(F,n,m,t)}{J(F,n,k+1,t)} \\
        &= \sum_{l=1}^{m} V_l P'_{(i,l)}(t \mid X_i = t, m) \frac{J(F,n,m,t)}{J(F,n,k+1,t)} \\
        &> \sum_{l=1}^{m} V_l P'_{(i,l)}(t \mid X_i = t, m)
    \end{align*}

    Multiplying both sides by \( t \) and integrating over \( (0, x_i] \), then applying the inverse cost function \( g^{-1} \), we obtain:
    \begin{align*}
        g^{-1}\left( \int_{0}^{x_i} \sum_{l=1}^{k+1} V_l P'_{(i,l)}(t \mid X_i = t, k+1)t \, dt \right) & > g^{-1}\left( \int_{0}^{x_i} \sum_{l=1}^{m} V_l P'_{(i,l)}(t \mid X_i = t, m)t \, dt \right)
    \end{align*}
    Thus, \( b(x_i; k+1) > b(x_i; m) \), meaning that merely providing more advancement slots will decrease the effort of any participant who was originally advancing, which completes the proof.
\end{proof}

\begin{lemma}\label{lem:binomDesc}
    Let $X \sim Binomial(n,p)$, then $\Pr(X = m \mid X \leq m)$ decreases as \(m\) increases.  
\end{lemma}
\begin{proof}
    First, observe that:
    \begin{align*}
    \frac{\Pr(X = m)}{\Pr(X = m+1)} & = \frac{\binom{n}{m} p^{m} (1-p)^{n-m}}{\binom{n}{m+1} p^{m+1} (1-p)^{n-m-1}} \\
                                & = \frac{m+1}{n-m} \cdot \frac{p}{1-p}
    \end{align*}
    This ratio increases as \(m\) increases.

    Therefore, we have:
    \begin{align*}
        \Pr(X \leq m) & = \sum_{j=0}^{m} \Pr(X = j) \\
                    & \leq \frac{m+1}{n-m} \cdot \frac{p}{1-p} \sum_{j=0}^{m+1} \Pr(X = j) \\
                    & = \frac{m+1}{n-m} \cdot \frac{p}{1-p} \cdot \Pr(X \leq m+1)
    \end{align*}
    Since \(\Pr(X = m \mid X \leq m) = \frac{\Pr(X = m)}{\Pr(X \leq m)}\), consider the ratio of consecutive terms:
    \begin{align*}
        \frac{\Pr(X = m \mid X \leq m)}{\Pr(X = m+1 \mid X \leq m+1)} & = \frac{\Pr(X = m)}{\Pr(X = m+1)} \cdot \frac{\Pr(X \leq m+1)}{\Pr(X \leq m)} \\
            & > \left( \frac{m+1}{n-m} \cdot \frac{p}{1-p} \right) \left( \frac{n-m}{m+1} \cdot \frac{1-p}{p} \right) \\
            & = 1 
    \end{align*}
    Since probabilities are non-negative, \(\Pr(X = m \mid X \leq m)\) decreases as \(m\) increases.
\end{proof}


\subsection*{Proof of Proposition~\ref{prop:DesignGuideline}}
\begin{proof}
    We first discuss the ex-post utility when the designer only cares about the effort from ranking $i$. The utility is expressed as:
    \[
    \begin{aligned}
        u_d(\vec{x})  & = c_i b(x_{(i)}) \\
             & = c_ig^{-1}\left(\int_{0}^{x_{(1)}}\frac{\sum_{l=1}^{m-1}\binom{n-1}{l-1}(n-l)(V_l-V_{l+1})F^{n-l-1}(t)(1-F(t))^{l-1}f(t)}{J(F,n,m,t)} t\, dt \right)
    \end{aligned}
    \]
    
    Since $g(\cdot)$ is an monotonically increasing function, so does the inverse $g^{-1}(\cdot)$. Also, $c_i$ is a positive constant, then we only have to optimize over the inner expression. 

    The designer's optimization problem then can be written as :
    \[
    \begin{aligned}
        \mathop{\arg \max}_{m,\vec{V}} \quad & \sum_{l=1}^{m-1}(V_l-V_{l+1})\int_{0}^{x_{(i)}}\frac{\binom{n-1}{l-1}(n-l)F^{n-l-1}(t)(1-F(t))^{l-1}}{\sum_{j=1}^{m}\binom{n-1}{j-1}F^{n-j}(t)(1-F(t))^{j-1}}f(t) t \, dt \\
        s.t. \quad & \sum_{l=1}^{m} V_l \leq B \\
            & V_l  \geq V_{l+1} \geq 0,
    \end{aligned}
    \]
    where we use Lemma~\ref{lem:normal} to simplifies the expression of $J(F,n,m,t)$. 

    Let $Z_l=l(V_l-V_{l+1})$, we know from Corollary~\ref{coro:Consolation} that the optimal contest design should satisfy $V_m=0$, hence the above optimization problem becomes:
    \begin{equation}\label{eq:ExpostHighestObj}
    \begin{aligned}
        \mathop{\arg \max}_{m,\vec{Z}} \quad & \sum_{l=1}^{m-1}Z_l\int_{0}^{x_{(1)}}\frac{\binom{n-1}{l-1}\frac{n-l}{l}F^{n-l-1}(t)(1-F(t))^{l-1}}{\sum_{j=1}^{m}\binom{n-1}{j-1}F^{n-j}(t)(1-F(t))^{j-1}}f(t) t \, dt \\
        s.t. \quad & \sum_{l=1}^{m-1} Z_l \leq B \\
            & Z_l \geq 0,
    \end{aligned}
    \end{equation}
    where it be easily verified that $\sum_{l=1}^{m-1} Z_l = \sum_{l=1}^{m} V_l -mV_m\leq B-0=B$. 

    Since for any given $m$, the coefficient for $Z_l$ is a positive constant that determined by the realization of $i^\text{th}$ order statistics of ability $x_{(1)}$, the optimal solution for the problem must have some $Z_{l^*}=B$, i.e., a non-trivial simple contest with $V_1=V_2=\ldots=V_{l^*}=\frac{B}{l^*}$, as desired. 

    Next, we consider the ex-post utility when the cost function is linear. The utility is given as: 
    \[
    \begin{aligned}
        u_d(\vec{x})  & =\vec{c}\cdot e(\vec{x}) \\
             & = \sum_{i=1}^{m}c_ik^{-1}\int_{0}^{x_i}\frac{\sum_{l=1}^{m-1}\binom{n-1}{l-1}(n-l)(V_l-V_{l+1})F^{n-l-1}(t)(1-F(t))^{l-1}f(t)}{J(F,n,m,t)} t\, dt \\
             & = \sum_{l=1}^{m-1}(V_l-V_{l+1}) \left [ \sum_{i=1}^{m} c_ik^{-1} \left ( \int_{0}^{x_{i}}\frac{\binom{n-1}{l-1}(n-l)F^{n-l-1}(t)(1-F(t))^{l-1}}{\sum_{j=1}^{m}\binom{n-1}{j-1}F^{n-j}(t)(1-F(t))^{j-1}}f(t) t \, dt \right ) \right ],
    \end{aligned}
    \]
    which is still a linear combination of $(V_l-V_{l+1})$s. 

    Again, we obtain the optimization problem:
    \[
    \begin{aligned}
        \mathop{\arg \max}_{m,\vec{Z}} \quad & \sum_{l=1}^{m-1}Z_l \left [ \sum_{i=1}^{m} c_ik^{-1} \left ( \int_{0}^{x_{i}}\frac{\binom{n-1}{l-1}\frac{n-l}{l}F^{n-l-1}(t)(1-F(t))^{l-1}}{\sum_{j=1}^{m}\binom{n-1}{j-1}F^{n-j}(t)(1-F(t))^{j-1}}f(t) t \, dt \right ) \right ] \\
        s.t. \quad & \sum_{l=1}^{m-1} Z_l \leq B \\
            & Z_l \geq 0,
    \end{aligned}
    \] the optimal solution of which is still a simple contest, as desired. 

    As for the ax-ante utilities, it follows directly that the optimal contest are simple contests, since the linearity of expectation operator and cost function enable us to extract out $(V_l-V_{l+1})$s, the optimization objective of the designer is still a non-negative combination of $Z_l$s, thus the optimal contest must be a non-trivial simple contest. 

    To sum up, we now conclude that if the designer's utility is a non-negative linear combination of admitted contestants' effort under the equilibrium, whether ex-ante or ex-post, as long as the designer only cares about effort of single ranking or the cost function is linear, the optimal contests are always non-trivial simple contests, which completes the proof.  
\end{proof}

\subsection*{Proof of Theorem~\ref{thm:ExpostHighestEffort}}
\begin{proof}
    From Equation~\ref{eq:ExpostHighestObj} in Proposition~\ref{prop:DesignGuideline}, the designer's maximization objective becomes:
    \[
    \begin{aligned}
        & \int_{0}^{x_{1}}\frac{\binom{n-1}{l-1}\frac{n-l}{l}F^{n-l-1}(t)(1-F(t))^{l-1}}{\sum_{j=1}^{m}\binom{n-1}{j-1}F^{n-j}(t)(1-F(t))^{j-1}}f(t) t \, dt \\
        = & \int_{0}^{x_{1}}\frac{\binom{n-1}{l}F^{n-l-1}(t)(1-F(t))^{l-1}}{\sum_{j=1}^{m}\binom{n-1}{j-1}F^{n-j}(t)(1-F(t))^{j-1}}f(t) t \, dt \\
        = & \int_{0}^{x_{1}}\frac{\binom{n-1}{l}F^{n-l-1}(t)(1-F(t))^{l}}{\sum_{j=1}^{m}\binom{n-1}{j-1}F^{n-j}(t)(1-F(t))^{j-1}} \frac{f(t)}{1-F(t)} t \, dt,
    \end{aligned}
    \]
    where the numerator of the first term in the integrand indicate the probability of a contestant with ability $t$ to get rank $l+1$ among all contestants, while the denominator denotes her chance of getting admitted into the contest with shortlist capacity $m$.

    To find the optimal simple contest, we have to identify the best $m$. Since the designer only cares about the highest effort $b(x_1)$, from Corollary~\ref{coro:EmptyPrize}, we know immediately that the optimal contest must satisfy $m^*=l^{*}+1$, where $l^*$ is the number of prize in the optimal contest.

    Substituting the relation of $m$ and $l$ when optimal, the optimization problem now becomes:
    \[
    \begin{aligned}
         \mathop{\arg \max}_{m} &  \int_{0}^{x_{1}}\frac{\binom{n-1}{m-1}F^{n-m}(t)(1-F(t))^{m-1}}{\sum_{j=0}^{m-1}\binom{n-1}{j}F^{n-j-1}(t)(1-F(t))^{j}}\frac{f(t)}{1-F(t)} t \, dt \\
        s.t. & \quad m \in [n]\backslash\{1\}.
    \end{aligned}
    \]

    The first term in the integrand of the objective corresponds to $\Pr(X=m-1\mid X \leq m-1)$, where $X\sim Binomial(n-1,1-F(t))$. It follows from Lemma~\ref{lem:binomDesc} that for all $t \geq0$, this term attains its maximum at $m=2$. So the objective also attains its maximum when $m=2,l=1$.

    Finally, we conclude that the optimal contest that maximize ex-post highest effort is a 2-contestant winner-take-all contest, the corresponding effort (or designer's utility $u_d$) expresses as: 
    $$
    \begin{aligned}
        & g^{-1}\left ( B \int_{0}^{x_{1}} \frac{(n-1)F^{n-2}(t)f(t)t}{F^{n-1}(t)+(n-1)F^{n-2}(t)(1-F(t))} \, dt \right ) \\
        = & g^{-1}\left ( B \int_{0}^{x_{1}} \frac{(n-1)f(t)t}{F(t)+(n-1)(1-F(t))} \, dt \right )
    \end{aligned}
    $$
    which completes the proof.
\end{proof}

\subsection*{Proof of Corollary~\ref{coro:OptimalMaximumEffort}}
\begin{proof}
    To find the optimal contest that maximize ex-ante highest effort, the problem can be written as:
    \[
    \mathop{\arg \max}_{m,\vec{V}} \quad\mathbb{E}_{x \sim X_{(1)}}[b(x;m,\vec{V})],
    \]
    where we use $b(\cdot;m,\vec{V})$ to denote the symmetric Bayesian Nash Equilibrium that uniquely determined by contest configuration $(m,\vec{V})$. 

    Theorem~\ref{thm:ExpostHighestEffort} states that the 2-contestant winner-take-all contest maximize $b(x;m,\vec{V})$ for any given $x$, i.e., for any other contest $(m', \vec{V}')$ and $\forall x\geq0$, we have $b(x;2,Be_1) \geq b(x;m',\vec{V}')$. 

    Therefore, $\mathbb{E}_{x\sim X_{(1)}}[b(x;2,Be_1)]\geq \mathbb{E}_{x\sim X_{(1)}}[b(x;m',\vec{V}')]$, i.e., the 2-contestant winner-take-all contest maximizes ex-ante highest effort, which completes the proof.
\end{proof}

\subsection*{Proof of Theorem \ref{thm: opt for total}}
\begin{proof}
    This theorem can be concluded by Proposition \ref{thm:ConpleteSimpleContest} and \ref{thm:OptAsmLinear}
\end{proof}

\subsection*{Proof of Proposition~\ref{thm:ConpleteSimpleContest}}
\begin{proof}
    Since we consider total effort under linear cost function, Proposition~\ref{prop:DesignGuideline} gives that the optimal contest is a non-trivial simple contest, with capacity $m$ and $l$ equal prizes.

    We will show that when $l$ is fixed, total effort decrease with $m$ for $l+1\leq m\leq n$. 

    We begin by introducing several notations, the first is order statistics. We use $X_{(k)} \sim F_{(k,n)}$ to denote the $k^{\text{th}}$ highest ability among all contestants, then:
    $$f_{(k,n)}(x) = n \binom{n-1}{k-1} (1-F(x))^{k-1} F(x)^{n-k} f(x) $$
    $$F_{(k,n)}(x) = n \binom{n-1}{k-1} \int_0^x (1-F(t))^{k-1} F(t)^{n-k} f(t) \,dt $$.

    Additionally, define:
    \[
    A(m, n, x) = \sum_{j=1}^m \binom{n-1}{j-1} F(x)^{n-j} (1-F(x))^{j-1},
    \]when $m=n$, $A(m, n, x) = 1$; when $m<n$, $\frac{dA}{dx} = (n-m) \binom{n-1}{m-1} F(x)^{n-m-1} (1-F(x))^{m-1} f(x) = f_{(m,n-1)}(x)$, therefore $A(m, n, x) = F_{(m,n-1)}(x) $.

    Then use these notation to rewrite the expression of total effort:
    $$
\begin{aligned}   
S(m,n,l) = &\sum_{i=1}^m n\binom{n-1}{i-1} \int_0^{\infty} \int_0^{x} \frac{\frac{1}{l} \cdot f_{(l,n-1)}(t) \cdot t}{A(m,n,t)}\, dt \cdot F(x)^{n-i}(1-F(x))^{i-1}f(x)\,dx \\= &\sum_{i=1}^m n\binom{n-1}{i-1} \int_0^{\infty} \int_t^{\infty} F(x)^{n-i}(1-F(x))^{i-1}f(x)\,dx \cdot \frac{\frac{1}{l} \cdot f_{(l,n-1)}(t) \cdot t}{A(m,n,t)} \,dt \\=& \frac{n}{l} \int_0^{\infty} \int_t^{\infty} \frac{f_{l,n-1}(t) \cdot t}{A(m,n,t)} \cdot A(m,n,x)f(x)\,dx \,dt \\=& \frac{n}{l} \int_0^{\infty} \frac{f_{l,n-1}(t)t}{A(m,n,t)} \,dt \int_t^{\infty} A(m,n,x)f(x)\, dx \\
=& \frac{n}{l} \int_0^{\infty} f_{l,n-1}(t)t  \int_t^{\infty} \frac{A(m,n,x)f(x)}{A(m,n,t)}\, dx \, dt
\end{aligned}
$$

    Now, define $h(m,t)$ to denote the inner integration:
    \[
    h(m, t) = \frac{\int_t^\infty A(m, n, x) f(x) dx}{A(m, n, t)} = \frac{\int_t^\infty\sum_{j=1}^m \binom{n-1}{j-1} F(x)^{n-j} (1-F(x))^{j-1}f(x)dx}{\sum_{j=1}^m \binom{n-1}{j-1} F(t)^{n-j} (1-F(t))^{j-1} }
    \]

    Consider the fraction of $j^\text{th}$ corresponding terms, we denote as $g(j,t)$:
    \[
    g(j,t)=\frac{\int_t^\infty \binom{n-1}{j-1} F(x)^{n-j} (1-F(x))^{j-1}f(x)dx}{\binom{n-1}{j-1} F(t)^{n-j} (1-F(t))^{j-1} }=\frac{\int_q^1 u^{n-j}(1-u)^{j-1}du}{q^{n-j}(1-q)^{j-1}},
    \]where we let $u=F(x)$ and $q=F(t)$ (symbol $q$ is then override).

    We will show that $g(j,t)$ is decreasing with $j$. 
    $$\frac{g(j,t)}{g(j-1,t)}=\frac{\int_q^1 u^{n-j}(1-u)^{j-1}du}{\int_q^1 u^{n-j+1}(1-u)^{j-2}du}\cdot \frac{q^{n-j+1}(1-q)^{j-2}}{q^{n-j}(1-q)^{j-1}}=\frac{\int_q^1 u^{n-j+1}(1-u)^{j-2}\cdot \frac{1-u}{u}du}{\int_q^1 u^{n-j+1}(1-u)^{j-2}du}\cdot\frac{q}{1-q}$$

    Because $\frac{1-u}{u}$ decrease with $u$, also it holds in the integration interval that $u\ge q$, so we have $\frac{1-u}{u}\le \frac{1-q}{q}$. Furthermore,
    $$\frac{g(j,t)}{g(j-1,t)}\le \frac{\int_q^1 u^{n-j+1}(1-u)^{j-2}\cdot \frac{1-q}{q}du}{\int_q^1 u^{n-j+1}(1-u)^{j-2}du}\cdot\frac{q}{1-q}=1,$$ (the equal sign binds only when integration of $u^{n-j+1}(1-u)^{j-2}$ over $u>q$ is $0$, thus, generally, we can assume that the inequality strictly holds). 

    Hence, $g(m,t)< g(m-1,t) < \ldots < g(1,t)$, showing that  $g(j,t)$ is decreasing with $j$. 

    We now prove by induction that $h(m,t)< h(m-1,t)$.

    \underline{Base Case:}
    First, we verify the base case. 
    $$h(1,t)=g(1,t)=\frac{\frac{1}{n}(1-q^n)}{q^{n-1}}$$
$$
\begin{aligned}
h(2,t)= & \frac{\sum_{j=1}^2 \binom{n-1}{j-1}\int_q^1 u^{n-j}(1-u)^{j-1}du}{\sum_{j=1}^2 \binom{n-1}{j-1} q^{n-j} (1-q)^{j-1} }\\
= &\frac{\sum_{j=1}^2 \binom{n-1}{j-1} q^{n-j} (1-q)^{j-1}\cdot g(j,t)}{\sum_{j=1}^2 \binom{n-1}{j-1} q^{n-j} (1-q)^{j-1} } \\ < & g(1,t) =h(1,t)
\end{aligned}
$$

Similarly, we can show that $h(2,t)>g(2,t)$.

    \underline{Inductive Hypothesis:}
    Assume that $h(j,t)< h(j-1,t)$ holds for all $j\leq m-1$.

    \underline{Induction Step:}
    We now prove that $h(m,t)< h(m-1,t)$. 

    $$
\begin{aligned}
& h(m,t)  \\ = &\frac{\sum_{j=1}^{m-1} \binom{n-1}{j-1}\int_q^1 u^{n-j}(1-u)^{j-1}du+\binom{n-1}{m-1}\int_q^1 u^{n-m}(1-u)^{m-1}du}{\sum_{j=1}^{m-1} \binom{n-1}{j-1} q^{n-j} (1-q)^{j-1}+\binom{n-1}{m-1} q^{n-m} (1-q)^{m-1}} \\ =&\frac{h(m-1,t)\cdot \sum_{j=1}^{m-1} \binom{n-1}{j-1} q^{n-j} (1-q)^{j-1}+g(m,t)\cdot \binom{n-1}{m-1} q^{n-m} (1-q)^{m-1}}{\sum_{j=1}^{m-1} \binom{n-1}{j-1} q^{n-j} (1-q)^{j-1}+\binom{n-1}{m-1} q^{n-m} (1-q)^{m-1}} \\ <&\frac{h(m-1,t)\cdot \sum_{j=1}^{m-1} \binom{n-1}{j-1} q^{n-j} (1-q)^{j-1}+g(m-1,t)\cdot \binom{n-1}{m-1} q^{n-m} (1-q)^{m-1}}{\sum_{j=1}^{m-1} \binom{n-1}{j-1} q^{n-j} (1-q)^{j-1}+\binom{n-1}{m-1} q^{n-m} (1-q)^{m-1}} \\ <&\frac{h(m-1,t)\cdot \sum_{j=1}^{m-1} \binom{n-1}{j-1} q^{n-j} (1-q)^{j-1}+h(m-1,t)\cdot \binom{n-1}{m-1} q^{n-m} (1-q)^{m-1}}{\sum_{j=1}^{m-1}\binom{n-1}{j-1} q^{n-j} (1-q)^{j-1}+\binom{n-1}{m-1} q^{n-m} (1-q)^{m-1}} \\=&h(m-1,t)
\end{aligned}
$$

Similarly, it can be shown that $h(m,t)>g(m,t)$. Thus, by the principle of mathematical induction, $g(j,t)<h(j,t)<h(j-1,t)$ holds for any $j=1,\ldots,m$.

Now the total effort becomes: 
$$
\begin{aligned}
S(m, n, l) & = \frac{n}{l}\int_0^\infty f_{(l,n-1)}(t)t h(m,t)dt \\ &< \frac{n}{l}\int_0^\infty f_{(l,n-1)}(t)t h(m-1,t)dt
\\ &= S(m-1,n,l)
\end{aligned}
$$

Therefore, $S(m,n,l)$ decrease with $m$, also $m>l$, thus, we conclude with $S(m,n,l)<S(l+1,n,l)$. 

So far we have showed that in total effort setting, the optimal simple contest must satisfies $l=m-1$, hence it is a complete simple contest, which complete the proof.
\end{proof}

\subsection*{Proof of Lemma~\ref{lem:betaRepTotalEffort}}
\begin{proof}
    Let $q=F(t)$ ($q$ therefore no longer stand for quantile), we then define $a(j,q)=\int_q^1\binom{n-1}{j-1}u^{n-j}(1-u)^{j-1}\,du$ and $b(j,q)=\binom{n-1}{j-1}q^{n-j}(1-q)^{j-1}$.

    Ex-ante total effort now rewrites as:
    \[
    S(m)=\int_0^1\frac{F^{-1}(q)}{1-q}b(m,q)\frac{\sum_{j=1}^{m}a(j,q)}{\sum_{j=1}^{m}b(j,q)}\,dq
    \]

    We first simplify the summation term in the denominator, i.e., $B(m,q)=\sum_{j=1}^{m}b(j,q)$, after taking derivative, we get:
    $$
\begin{aligned}
B(m,q)' = &(n-1)q^{n-2}+\sum_{j=2}^m(\binom{n-1}{j-1}(n-j)q^{n-j-1}(1-q)^{j-1}-\binom{n-1}{j-1}(j-1)q^{n-j}(1-q)^{j-2}) \\ = & \sum_{j=2}^m (-\binom{n-1}{j-1}(j-1)q^{n-j}(1-q)^{j-2}+\binom{n-1}{j-1-1}(n-j+1)q^{n-j+1-1}(1-q)^{j-1-1})\\&+\binom{n-1}{m-1}(n-m)q^{n-m-1}(1-q)^{m-1} \\= &\binom{n-1}{m-1}(n-m)q^{n-m-1}(1-q)^{m-1}    
\end{aligned}
$$

Thus, we can use integration to express $B(m,q)$:
$$B(m,q)=B(m,0)+\int_0^qB(m,u)'\,du=\int_0^qB(m,u)'\, du=\int_0^q\beta(u,n-m,m)\,du.$$

By definition, $b(m,q)=\frac{1}{n}\beta(q,n-m+1,m)$, so:
$$A(m,q)=\sum_{j=1}^{m}a(j,q)=\int_q^1B(m,u)\, du=\int_q^1\int_0^u \beta(x,n-m,m)\, dx\, du.$$

Rearrange the order of integration. Since the integration region is a trapezoid, the outer integral can be split into two parts:
$$
\begin{aligned} 
A(m,q)&=\int_0^q\int_q^1\beta(x,n-m,m)\,du\,dx+\int_q^1\int_x^1 \beta(x,n-m,m)\,du\,dx  \\&= \int_0^q(1-q)\beta(x,n-m,m)\,dx+\int_q^1(1-x) \beta(x,n-m,m)\,dx  
\end{aligned}
$$

Thereby, total effort can be simplified using beta distribution:
\[
\begin{aligned}
S(m) &=n\int_0^1\frac{F^{-1}(q)}{1-q}b(m,q)\frac{\sum_{j=1}^{m}a(j,q)}{\sum_{j=1}^{m}b(j,q)}dq 
\\&=n\int_0^1\frac{F^{-1}(q)}{1-q}\frac{1}{n}\beta(q,n-m+1,m)\\&\phantom{a}\cdot\frac{\int_0^q(1-q)\beta(x,n-m,m)dx+\int_q^1(1-x) \beta(x,n-m,m)dx}{\int_0^q\beta(u,n-m,m)du}dq \\&=\int_0^1\frac{F^{-1}(q)}{1-q}\beta(q,n-m+1,m)\big((1-q)+\frac{\int_q^1(1-x) \beta(x,n-m,m)dx}{\int_0^q\beta(u,n-m,m)du}\big)dq \\&=\int_0^1F^{-1}(q)\beta(q,n-m+1,m)dq\\&\phantom{a}+\int_0^1\frac{F^{-1}(q)}{1-q}\beta(q,n-m+1,m)\frac{\int_q^1(1-x) \beta(x,n-m,m)dx}{\int_0^q\beta(u,n-m,m)du}dq \\&=\int_0^1F^{-1}(q)\beta(q,n-m+1,m)dq\\ & \phantom{a}+\int_0^1\frac{F^{-1}(q)q}{1-q}\frac{n}{n-m}\beta(q,n-m,m)\frac{\frac{m}{n}\int_q^1\beta(x,n-m,m+1)dx}{\int_0^q\beta(u,n-m,m)du}dq \\&=\int_0^1F^{-1}(q)\beta(q,n-m+1,m)dq
\\ & \phantom{a}+\int_0^1\frac{F^{-1}(q)q}{1-q}\frac{m}{n-m}\beta(q,n-m,m)\frac{\int_q^1\beta(x,n-m,m+1)dx}{\int_0^q\beta(u,n-m,m)du}dq
\\ &= \int_0^1F^{-1}(q)\beta(q,n-m+1,m)dq\\&\phantom{a}+\int_0^1 F^{-1}(q)\frac{q}{1-q}\frac{m}{n-m}\frac{\beta(q,n-m,m)}{\int_0^q\beta(u,n-m,m)du}\int_q^1\beta(x,n-m,m+1)\,dx\,dq,
\end{aligned}
\]which completes the proof.
\end{proof}

\begin{lemma}\label{lem:1}
    Assume $\beta(q,n-m+1,m+1)$ attain it maximum at $q=\mu$, then when $\mid q-\mu\mid > \delta$, we have $\beta(q,n-m+1,m+1)\le \epsilon$, where $\delta = \sqrt{\frac{k(1-k)}{n}\big(\ln n+\ln \frac{1}{\epsilon^2}+\ln \frac{1}{2\pi k(1-k)}\big)}=\Theta(\sqrt{\frac{\ln n}{n}})$.
\end{lemma}
\begin{proof}
    We begin by finding the point where probability mass of beta distribution attain its maximum. Taking derivative: 
$$\beta(q,\alpha,\beta)'=\beta(q,\alpha,\beta)(\frac{\alpha-1}{q}-\frac{\beta-1}{1-q}).$$

    Then it can be concluded that the Beta distribution function first increases and then decreases. At \(\mu = \frac{\alpha - 1}{\alpha + \beta - 2}\), the derivative equals zero, and the function attains its maximum value, $\beta(q,n-m,m)_{\max}\simeq\frac{\sqrt{n}}{\sqrt{2\pi k(1-k)}}$, where $k=n/m$, when $n$ is large.

For the convenience of discussion, we consider $\beta(q,n-m+1,m+1)$, thus, $\mu=\frac{n-m}{n}=1-k$, 
$$
\begin{aligned}
&\beta(\mu-c,n-m+1,m+1) \\&=\frac{(n+1)!}{(n-m)!m!}(\mu-c)^{n-m}(1-\mu+c)^{m} \\&=\frac{(n+1)\sqrt{2\pi n}(\frac{n}{e})^n}{\sqrt{2\pi (n-m)}(\frac{n-m}{e})^{n-m}\sqrt{2\pi (m)}(\frac{m}{e})^{m}}(\mu-c)^{n-m}(1-\mu+c)^{m}\\&=\frac{\sqrt{n}(\frac{1-k-c}{1-k})^{n-m}(\frac{k+c}{k})^m}{\sqrt{2\pi k(1-k)}} \\&=\frac{\sqrt{n}}{\sqrt{2\pi k(1-k)}}\exp\{(n-m)\ln (1-\frac{c}{1-k})+m\ln (1+\frac{c}{k})\} \\&=\frac{\sqrt{n}}{\sqrt{2\pi k(1-k)}}\exp\{n(1-k)(-\frac{c}{1-k}-\frac{1}{2}\big(\frac{c}{1-k})^2\big) + \big( nk (\frac{c}{k}-\frac{1}{2}(\frac{c}{k})^2)\big) \}\\&=\frac{\sqrt{n}}{\sqrt{2\pi k(1-k)}}\exp\{ \frac{-nc^2}{2k(1-k)}\},   
\end{aligned}
$$
where the second equality derives from Sterling's formula, and the second to last equality derives from Taylor expansion.

To satisfy $\beta(\mu-c,n-m+1,m+1)\le \epsilon$, we have to make the following inequality stands:
$$\ln (\frac{\sqrt{n}}{\sqrt{2\pi k(1-k)}})-\frac{nc^2}{2k(1-k)}\le \ln \epsilon$$

Therefore, 
$$c\ge \delta=\sqrt{\frac{k(1-k)}{n}\big(\ln n+\ln \frac{1}{\epsilon^2}+\ln \frac{1}{2\pi k(1-k)}\big)}=\Theta(\sqrt{\frac{\ln n}{n}}).$$

We aim to ensure that the convergence rate is independent of \( k \). To achieve this, it suffices to choose \( c \) greater than the maximum value of the right-hand side (with respect to \( k \)). Since \( k(1-k) \leq \frac{1}{4} \), it follows that:
\[
\frac{k(1-k)}{n} \ln n \leq \frac{1}{4} \ln n.
\]

Additionally, as \( k(1-k) \to 0 \), we have: 
\[
k(1-k) \ln \frac{1}{2\pi k(1-k)} \to 0.
\]

Therefore, there exists a sufficiently large \( N \) such that for all \( n > N \), the maximum value of the right-hand side does not exceed \( \sqrt{\frac{\ln n}{n}} \). Consequently, as long as the deviation of \( q \) from \( \mu \) exceeds \( \delta = \sqrt{\frac{\ln n}{n}} \), it holds that \( \beta(q, n-m+1, m+1) \leq \epsilon \), as desired.

The analysis is similar for the other side. This completes the proof.
\end{proof}

\begin{lemma}\label{lem:2}
    when $n \rightarrow \infty$, we have:
    $$\int_0^1F^{-1}(q)\beta(q,n-m+1,m)\,dq = F^{-1}(1-k).$$
\end{lemma}
\begin{proof}
    Provided that $F^{-1}(q)\le L$ is bounded and continuous. Therefore for any $\epsilon>0$, there exists $\delta_1>0$, such that when $|q-\mu|=|q-1+k|\le \delta_1$, we have $|F^{-1}(q)-F^{-1}(1-k)|\le \epsilon /2$.

    From Lemma~\ref{lem:1}, for $\epsilon_1=\frac{\epsilon}{2L}$, exists $\delta_2=\Theta(\sqrt{\frac{\ln n}{n}})$, such that $\beta(q,n-m+1,m)<\epsilon_1, if |q-\mu|> \delta_2$.

When $n\rightarrow +\infty$, $\frac{\ln n}{n}=0$, therefore exist $N_1$, such that when $n>N_1$, $\delta_2=\Theta(\sqrt{\frac{\ln n}{n}})<\delta_1$. 

Let $\delta=\delta_1$, therefore:
$$
\begin{aligned}
& \int_0^1F^{-1}(q)\beta(q,n-m+1,m)dq \\ &=\int_{\mu-\delta}^{\mu+\delta}F^{-1}(q)\beta(q,n-m+1,m)dq+ \int_0^{\mu-\delta}F^{-1}(q)\beta(q,n-m+1,m)dq \\ & \phantom{a}+\int_{\mu+\delta}^1F^{-1}(q)\beta(q,n-m+1,m)dq  \\ &\le (1-2\delta)\epsilon_1L+ (F^{-1}(1-k)+\frac{\epsilon}{2})\int_{\mu-\delta}^{\mu+\delta}\beta(q,n-m+1,m)dq \le F^{-1}(1-k)+\epsilon    
\end{aligned}
$$

Similarly, we can prove that:
$$\int_0^1F^{-1}(q)\beta(q,n-m+1,m)dq \ge F^{-1}(1-k)-\epsilon $$

Thus, when $n\rightarrow \infty$, it holds that:
$$\int_0^1F^{-1}(q)\beta(q,n-m+1,m)dq = F^{-1}(1-k),$$
which completes the proof.
\end{proof}

\begin{lemma}\label{lem:3}
    When $q<\mu=1-k$ and $n\rightarrow \infty$, it holds that:
    $$\frac{\beta(q,n-m,m)}{n\int_0^q\beta(x,n-m,m)dx}=\frac{1}{q}-\frac{k}{q(1-q)}.$$
\end{lemma}
\begin{proof}
    We start by making discrete approximation for the numerator. Diving $[0,q]$ into $N$ intervals evenly, $q_j=\frac{j}{N}q,j=0,1,2,\cdots,N$, let $n=N^2$, the fraction becomes:
$$
\begin{aligned}
\frac{\beta(q,n-m,m)}{n\int_0^q\beta(x,n-m,m)dx} & =\frac{\sum_{j=1}^{N-1}(\beta(q_{j+1}-\beta(q_j)))}{n\sum_{j=1}^{N-1}\beta(q_{j})\frac{q}{N}}\\ &= \frac{\sum_{j=1}^{N-1}\beta(q_j)'\frac{q}{N}}{n\sum_{j=1}^{N-1}\beta(q_{j})\frac{q}{N}}\\&=\frac{\sum_{j=1}^{N-1}\beta(q_j)(\frac{1}{q_j}-\frac{k}{q_j(1-q_j)})}{\sum_{j=1}^{N-1}\beta(q_{j})}
\end{aligned}
$$

This fraction can be viewed as weighted average over $N-1$ fraction terms, where the $j^{\text{th}}$ fraction term is $\frac{1}{q_j}-\frac{k}{q_j(1-q_j)}$, and its weight is $\frac{\beta(q_{j})}{\sum_{j=1}^{N-1}\beta(q_{j})}$.

We can compare the weight of two consecutive terms:
$$
\begin{aligned}
\frac{\beta(q_{j+1})}{\beta(q_j)} & =1+\frac{\beta(q_{j+1})-\beta(q_j)}{\beta(q_j)}\\ &=1+n\big(\frac{1}{q_j}-\frac{k}{q_j(1-q_j)}\big)\frac{q}{N}\\ &>1+N\big(\frac{1}{q}-\frac{k}{q(1-q)}\big)\\ &=\lambda_N
\end{aligned}
$$

Then the wight of last term:
$$\frac{\beta_{N-1}}{\sum_{j=1}^{N-1}\beta(q_{j})}\ge\frac{\beta_{N-1}}{\beta_{N-1}\sum_{j=1}^{N-1}\frac{1}{\lambda_N^{j-1}}}\ge \frac{1}{1-\frac{1}{\lambda_N}}=\frac{\lambda_N}{\lambda_N-1}$$

When \( q \leq 1-k-\delta \), we have:
\[
\lambda_N = 1 + \frac{n}{N} \frac{1}{q} \left( 1 - \frac{k}{1-q} \right) \geq 1 + \frac{n}{N} \frac{1}{1-k-\delta} \frac{\delta}{k+\delta}.
\]

Using the definition of \(\delta = \sqrt{\frac{\ln n}{n}}\), this can be further bounded as:
\[
\lambda_N \geq 1 + \frac{4 \delta n}{N} = 1 + \frac{4n \sqrt{\frac{\ln n}{n}}}{N} = 1 + 4\sqrt{\ln n}.
\]

As \( N \to \infty \), it follows that \( \lambda_N \to \infty \). Consequently, in the expression:
\[
\frac{\sum_{j=1}^{N-1} \beta(q_j) \left( \frac{1}{q_j} - \frac{k}{q_j(1-q_j)} \right)}{\sum_{j=1}^{N-1} \beta(q_j)},
\]

the weight of the final term approaches 1. Therefore, the entire ratio converges to the value of the final term. Therefore, we have:
\[
\lim_{n \to \infty} \frac{\beta(q, n-m, m)}{n \int_0^q \beta(x, n-m, m) \, dx} = \lim_{N \to \infty} \left( \frac{1}{q_{N-1}} - \frac{k}{q_{N-1}(1-q_{N-1})} \right)=\frac{1}{q} - \frac{k}{q(1-q)}.
\]

Moreover, the rate of convergence is independent of \( k \). The proof is then completed.
\end{proof}

\subsection*{Proof of Lemma~\ref{lem:AsyRep}}
\begin{proof}
    Recall from Lemma~\ref{lem:betaRepTotalEffort} the beta representation of ex-ante total effort:
    \[
    \begin{aligned}
        S(m,n) = & 
        \int_0^1F^{-1}(q)\beta(q,n-m+1,m)\,dq \\
        & +\int_0^1F^{-1}(q)\frac{q}{1-q} \frac{m}{n-m}\frac{\beta(q,n-m,m)}{\int_0^q\beta(x,n-m,m)\,dx}\int_q^1\beta(x,n-m,m+1)\,dx\, dq.
    \end{aligned}
    \]

    Let:
    \[
    \phi(k) = \int_0^{1-k} F^{-1}(q) \frac{q}{1-q} \frac{k}{1-k} \left( \frac{1}{q} - \frac{k}{q(1-q)} \right) \, dq,
    \]then we aim to show that for any \(\epsilon > 0\), there exists \(N > 0\) such that for all \(n > N\), the inequality \(
\left| \frac{S(m,n)}{n} - \phi(k) \right| \leq \epsilon, \forall k \in (0,1), \) holds.


    The first term in \(\frac{S(m,n)}{n}\) is \(\frac{1}{n} \int_0^1 F^{-1}(q) \beta(q, n-m+1, m) \, dq. \) From the asymptotic properties of the beta distribution, as \(n \to \infty\), \(\beta(q, n-m+1, m)\) becomes concentrated near \(q = 1-k\), and the integral's contribution outside this region vanishes. Furthermore, since \(F^{-1}(q)\) is bounded and \(\beta(q, n-m+1, m)\) scales with \(n\), we have:
\[
\lim_{n \to \infty} \frac{1}{n} \int_0^1 F^{-1}(q) \beta(q, n-m+1, m) \, dq = \lim_{n \to \infty} \frac{F^{-1}(1-k)}{n} = 0.
\]
Thus, the first term vanishes asymptotically.

    The second integration term in the $S(m,n)$ is $\frac{1}{n}\int_0^1F^{-1}(q)\frac{q}{1-q} \frac{m}{n-m}\frac{\beta(q,n-m,m)}{\int_0^q\beta(x,n-m,m)\,dx}\int_q^1\beta(x,n-m,m+1)\,dx\, dq$, which can be separated into three parts based on the integration interval.

    \underline{Part 1:} When $q<1-k-\delta$,

    \(\beta(q, n-m, m)\) is negligible because the beta distribution is concentrated near the $\delta$ neighborhood of \(q = 1-k\) (Lemma~\ref{lem:1}) with $\delta = \Theta(\sqrt{\frac{\ln n}{n}})$, the the convergence speed is independent of $k$. Additionally, for \(q < 1-k-\delta\), \(\int_q^1 \beta(x, n-m, m+1) \, dx \approx 1\). This combines with Lemma~\ref{lem:3}, gives:
    \[
    \begin{aligned}
        & \lim_{n \to \infty} \frac{1}{n}\int_0^{1-k-\delta}F^{-1}(q)\frac{q}{1-q} \frac{m}{n-m}\frac{\beta(q,n-m,m)}{\int_0^q\beta(x,n-m,m)\,dx}\int_q^1\beta(x,n-m,m+1)\,dx\, dq \\
        = & \lim_{n \to \infty} \frac{1}{n}\int_0^{1-k-\delta}F^{-1}(q)\frac{q}{1-q} \frac{m}{n-m}\frac{\beta(q,n-m,m)}{\int_0^q\beta(x,n-m,m)\,dx}\, dq \\
        = & \lim_{n \to \infty} \int_0^{1-k-\delta}F^{-1}(q)\frac{q}{1-q} \frac{m}{n-m}\left (\frac{1}{q} -\frac{k}{q(1-q)} \right)\, dq \\
        = & \int_0^{1-k}F^{-1}(q)\frac{q}{1-q} \frac{k}{1-k}\left (\frac{1}{q} -\frac{k}{q(1-q)} \right)\, dq.
    \end{aligned}
    \]

    \underline{Part 2:} When $q\in (1-k-\delta, 1-k+\delta)$, we denote this integration as $S_2$.


\[
\frac{\beta(q, n-m, m)}{n \int_0^q \beta(x, n-m, m) dx} = \frac{n-m}{nq} \frac{b(m, q)}{B(m, q)} \leq \frac{1-k}{kq}.
\]

Now, for the term \( S_2 \), we have the following bounds:
\[
S_2 \leq \int_{1-k-\delta}^{1-k+\delta} F^{-1}(q) \frac{q}{1-q} \frac{k}{1-k} \frac{1-k}{kq} \int_q^1 \beta(x, n-m, m+1) dx.
\]

Simplifying further, we find:
\[
S_2 \leq \int_{1-k-\delta}^{1-k+\delta} \frac{F^{-1}(q)}{1-q} dq \leq \frac{F^{-1}(1-k+\delta)}{k-\delta} \cdot 2\delta.
\]

Since \( \frac{F^{-1}(1-k+\delta)}{k-\delta} \) is bounded, and with \( \delta = \Theta\left(\sqrt{\frac{\ln n}{n}}\right) \), we conclude that: \(\lim_{n \to \infty} S_2 = 0.\)
    

    \underline{Part 3:} When $q>1-k+\delta$,

\[
\begin{aligned}
    S_3 & \leq \int_{1-k+\delta}^1 F^{-1}(q) \frac{q}{1-q} \frac{m}{n-m} \frac{\beta(q, n-m, m)}{n(1-\epsilon)} \epsilon \, dq 
\\ & = \int_{1-k+\delta}^1 F^{-1}(q) \frac{m}{n-m} \frac{n-m}{m-1} \beta(q, n-m+1, m-1) \frac{\epsilon}{n(1-\epsilon)} \, dq,
\\ & \leq\int_{1-k+\delta}^1 F^{-1}(q) \frac{2\epsilon^2}{n(1-\epsilon)} \, dq.
\end{aligned}
\]

Since \( F^{-1}(q) \) is bounded, we have:
\[
S_3 \leq \frac{2\epsilon^2}{n(1-\epsilon)} \int v f(v) \, dv = \frac{2\epsilon^2}{n(1-\epsilon)} E(v),
\]where \( E(v) \) denotes the expectation of \( v \), and it is assumed to be bounded. 

Taking the limit as \( n \to \infty \), We get
\(\lim_{n \to \infty} S_3 = 0,
\)

    

    To sum up, we conclude that \(
\lim_{n \to \infty} \frac{S(m,n)}{n} = \phi(k),
\) which completes the proof.
\end{proof}

\subsection*{Proof of Proposition~\ref{thm:OptAsmLinear}}
\begin{proof}
    Recall from Lemma~\ref{lem:AsyRep} that: 
    $$\lim_{n\rightarrow +\infty}\frac{S(m,n)}{n}=\int_0^{1-k}F^{-1}(q)\frac{q}{1-q}\frac{k}{1-k}(\frac{1}{q}-\frac{k}{q(1-q)})dq,$$
    we then denote $\phi(k):=\lim_{n\rightarrow +\infty}\frac{S(m,n)}{n}$, which is a function of $k=m/n$.

To find the optimal $m^*(n)$, it is suffice to find the $k$ such that $\phi(k)$ attains its maximum. Since $\phi'(0)=\int_0^1F^{-1}(1-q)\frac{1}{q}dq>0$, $\phi'(k_2)<0$ (where $k_2 \approx 0.3162$),  there exist an $0<k^*<k_2$, such that $\phi'(k^*)=0$, and $\phi(k)$ attains maximum. 

Such $k^*$ can be found by solving the following equation:
$$\frac{d \phi}{dk}=\frac{1}{(1-k)^2}\int_k^1 F^{-1}(1-q)(\frac{1}{q}-(2k-k^2)\frac{1}{q^2})dq=0.$$

This is equivalent to:
$$\int_k^1 F^{-1}(1-q)(\frac{1}{q}-(2k-k^2)\frac{1}{q^2})dq=0,$$
which completes the proof. 
\end{proof}

\subsection*{Proof of Theorem~\ref{thm:UniversalBound}}
\begin{proof}
The proof consists of two parts. First, we demonstrate that \(k_2\) is an upper bound, and then we show that this upper bound is tight by constructing a binding distribution.

\underline{Part 1:} \(k_2\) is an upper bound.

Based on Theorem~\ref{thm:OptAsmLinear}, the optimal \(k\) must satisfy:
\[
\int_k^1 F^{-1}(1-q)\left(\frac{1}{q}-(2k-k^2)\frac{1}{q^2}\right)dq=0.
\]

\(r(q)=\frac{1}{q}-(2k-k^2)\frac{1}{q^2}\) is zero at \(q=2k-k^2\), and \(r(q)<0\) for \(k \le q < 2k-k^2\), \(r(q)>0\) for \(1 \ge q > 2k-k^2\). We then Let \(q_0=2k-k^2\).

Since \(F^{-1}(1-q)\) is a decreasing function of \(q\), for any \(k\):
\[
\begin{aligned}
&\int_k^1 F^{-1}(1-q)\left(\frac{1}{q}-(2k-k^2)\frac{1}{q^2}\right)dq \\
&=\int_k^{q_0} F^{-1}(1-q)\left(\frac{1}{q}-(2k-k^2)\frac{1}{q^2}\right)dq+\int_{q_0}^1 F^{-1}(1-q)\left(\frac{1}{q}-(2k-k^2)\frac{1}{q^2}\right)dq \\
&\le F^{-1}(1-q_0)\int_k^{q_0}\left(\frac{1}{q}-(2k-k^2)\frac{1}{q^2}\right)dq+F^{-1}(1-q_0)\int_{q_0}^1\left(\frac{1}{q}-(2k-k^2)\frac{1}{q^2}\right)dq \\
&= F^{-1}(1-q_0)\int_k^1\left(\frac{1}{q}-(2k-k^2)\frac{1}{q^2}\right)dq \\
&= F^{-1}(1-q_0)\left((2-k)(k-1)-\ln k \right)
\end{aligned}
\]

When \(k\) is the solution \(k_2\) of the equation \(\ln k = (2-k)(k-1)\), the right-hand side of the above inequality is zero. This shows that at \(k=k_2\), the derivative of \(\phi(k)\) is non-positive.

Next, we will show that the optimal $k$ will not appear after $k_2$. 

For \(k>k_2\), let \(c(k)=(2-k)(k-1)-\ln k\), then \(c'(k)=-2k+3-\frac{1}{k}=-\frac{1}{k}\left((2k-1)(k-1)\right)\).

Thus, \(c'(k)\) is first negative and then positive, with a zero at \(k=0.5\), implying that \(c(k)\) first decreases and then increases. Since \(c(k_2)=0\) and \(c(1)=0\), for \(k_2< k < 1\), \(c(k)<0\), and $\phi'(k) < F^{-1}(1-q_0)c(k) < 0$, so \(S(n,m)\) will not achieve its maximum for \(k>k_2\).

Therefore, there exists a linear upper bound for the optimal \(m\) in terms of \(n\), i.e., \(\lim_{n \rightarrow \infty} \frac{m^*(n)}{n} \leq k_2\), where \(k_2\) is the non-trivial solution of the equation \(\ln k = (2-k)(k-1)\).

\underline{Part 2:} \(k_2\) is binding.

We now construct a probability distribution such that the corresponding optimal \(k^*\) approaches \(k_2\) asymptotically.

Consider \(f(v)\), which is the probability mass function of a Beta distribution \(f(v; \alpha, \beta)\). As \(\alpha + \beta \to \infty\), \(f(v)\) concentrated around a single point \(\mu = \frac{\alpha - 1}{\alpha + \beta - 2}\). Let \(F^{-1}(1-q)\) be written as \(v(q)\), where \(v(q)\) is a decreasing function of \(q\). The absolute value of the derivative of \(v(q)\) with respect to \(q\) is given by \(|v'(q)| = \frac{1}{f(v)}\). Since \(f(v)\) is uni-modal, increasing and then decreasing as \(q\) varies from \(k\) to 1, \(v'(q)\) is always negative, and \(|v'(q)|\) first decreases and then increases.

By Lemma~\ref{lem:1}, for any \(\epsilon > 0\), if \(|x-\mu| > \delta = \sqrt{\frac{k(1-k)}{n} \left( \ln n + \ln \frac{1}{\epsilon^2} + \ln \frac{1}{2 \pi k(1-k)} \right)} = \Theta\left(\sqrt{\frac{\ln n}{n}}\right)\), then \(f(x) < \epsilon\). This implies that the Beta distribution \(f(v; \alpha, \beta)\) integrates to \(\epsilon_1\) and \(\epsilon_2\) over the intervals \([\mu + \delta, 1]\) and \([0, \mu - \delta]\), respectively, with \(\epsilon_1 + \epsilon_2 < \epsilon\) holds.

So, for \(v(q)\), when \(q > \epsilon_1\), \(v(q) < \mu + \delta\); and when \(q > 1 - \epsilon_2\), \(v(q) < \mu - \delta\). As \(n = \alpha + \beta\) grows large, \(\delta = \Theta\left(\frac{\ln n}{n}\right) \to 0\). Thus, there exists an \(N\) such that when \(\alpha + \beta > N\), \(\delta < \epsilon\) is satisfied.

The integral \(\int_k^1 F^{-1}(1-q)\left(\frac{1}{q} - (2k-k^2)\frac{1}{q^2}\right) dq\) can now be simplified as:
\[
\int_k^{1-\epsilon_2} \mu \left(\frac{1}{q} - (2k-k^2)\frac{1}{q^2}\right) dq = \mu \int_k^{1-\epsilon_2} \left(\frac{1}{q} - (2k-k^2)\frac{1}{q^2}\right) dq.
\]

When \(k = k_2\), this integral approaches 0. This implies that for such a Beta distribution \(f(v; \alpha, \beta)\), \(k^* \to k_2 \approx 31.62\%\), therefore upper bound $k_2$ is binding, which completes the proof. 
\end{proof}

\section{Missing Proofs in Section \ref{sec: compare}}

\subsection*{Proof of Lemma~\ref{lem:QuantileRep}}
\begin{proof}
    We start from the normal expression of $S(m,n,l)$:
    \begin{align*}
    & \mathbb{E}_{X_{(1)}, \ldots, X_{(m)}} \left [\sum_{i=1}^{m}\int_{0}^{x_{(i)}}\frac{\binom{n-1}{l}F^{n-l-1}(t)(1-F(t))^{l}}{\sum_{j=1}^{m}\binom{n-1}{j-1}F^{n-j}(t)(1-F(t))^{j-1}}\frac{f(t)}{1-F(t)} t\, dt \right ]\\
    = &\sum_{i\in[n]}\mathbb{E}_{x_i}[\Pr[x_i>=x_{(m)}]\int_{0}^{x_{i}}\frac{\binom{n-1}{l}F^{n-l-1}(t)(1-F(t))^{l}}{\sum_{j=1}^{m}\binom{n-1}{j-1}F^{n-j}(t)(1-F(t))^{j-1}}\frac{f(t)}{1-F(t)} t\,dt]\\
    = &n\int_0^{+\infty} F_{(m-1,n-1)}(x_i)f(x_i)\int_{0}^{x_{i}}\frac{\binom{n-1}{l}F^{n-l-1}(t)(1-F(t))^{l}}{\sum_{j=1}^{m}\binom{n-1}{j-1}F^{n-j}(t)(1-F(t))^{j-1}}\frac{f(t)}{1-F(t)} t\,dt\,dx_i\\
    = &n\int_0^{+\infty} F_{(m-1,n-1)}(x_i)f(x_i)\int_{1-F(x_{i})}^{1}\frac{\binom{n-1}{l}(1-q)^{n-l-1}q^{l}q^{-1}}{\sum_{j=1}^{m}\binom{n-1}{j-1}(1-q)^{n-j}q^{j-1}}v(q)\,dq\,dx_i\\
    = &n\int_0^{1} \sum_{j=1}^{m}\binom{n-1}{j-1}q_i^{j-1}(1-q_i)^{n-j}\int_{q_i}^{1}\frac{\binom{n-1}{l}(1-q)^{n-l-1}q^{l}q^{-1}}{\sum_{j=1}^{m}\binom{n-1}{j-1}(1-q)^{n-j}q^{j-1}}v(q)\,dq \,dq_i\\
    = &n\int_0^{1} \frac{\binom{n-1}{l}(1-q)^{n-l-1}q^{l}q^{-1}}{\sum_{j=1}^{m}\binom{n-1}{j-1}(1-q)^{n-j}q^{j-1}}\int_{0}^{q}\sum_{j=1}^{m}\binom{n-1}{j-1}q_i^{j-1}(1-q_i)^{n-j}\,dq_i\,v(q)\,dq,
\end{align*}where $F_{(m-1,n-1)}(x)$ is cumulative probability function following the proof of Theorem~\ref{thm:ConpleteSimpleContest}.

    Next, we use $G_{(m,l)}(q)$ to denote the distribution-free part:
    \[
    G_{(m,l)}(q)=\frac{\binom{n-1}{l}(1-q)^{n-l-1}q^{l-1}}{\sum_{j=1}^{m}\binom{n-1}{j-1}(1-q)^{n-j}q^{j-1}}\int_{0}^{q}\sum_{j=1}^{m}\binom{n-1}{j-1}p^{j-1}(1-p)^{n-j}\,dp.
    \]

    Then, by change of integration sequence, we obtain:
    \[
    \begin{aligned}
        S(l,n,m) = & \int_0^{1} G_{l,m}(q)v(q)\,dq
        \\ = & \int_0^{1} G_{l,m}(q)\int_{q}^1|v'(t)|\, dt\,dq
        \\ = &\int_0^1|v'(q)|\int_0^qG_{l,m}(t)\,dt\,dq,
    \end{aligned}
    \]as desired. The same derivation applies to $S^{(1)}(m,n,l)$, resulting in $S^{(1)}(m,n, l)= n\int_0^1|v'(q)|\int_0^qG^{(1)}_{(m,l)}(t)\,dt\,dq$ and $G_{l,m}^{(1)}(t):=\frac{\binom{n-1}{l}(1-t)^{n-l-1}t^{l-1}}{\sum_{j=1}^{m}\binom{n-1}{j-1}(1-t)^{n-j}t^{j-1}}\int_{0}^{t}(1-p)^{n-1}\,dp,$ which completes the proof.
\end{proof}


\begin{lemma}
\label{lemma:1/qint approximation}
    Define $$\zeta(m, n, q) = \sum_{j=1}^m \binom{n-1}{j-1} (1-q)^{n-j}q^{j-1},$$
    For any $1\leq m\leq n$, $q\in[0,1]$, it holds that 
    $$ \frac14\min\{1,\frac{m}{nq}\}\leq \frac1{q}\int_0^{q}\zeta(m, n, t)dt\leq \min\{1,\frac{m}{nq}\}.$$
\end{lemma}
\begin{proof}
    Firstly, we have $\zeta(m,n,q)\in[0,1]$ for all $q\in[0,1]$, and $\int_0^1\zeta(m, n, t)dt=\frac{m}{n}$.
    We immediately have $\frac1{q}\int_0^{q}\zeta(m, n, t)dt\leq\frac1{q}\int_0^{q}1dt\leq=1$ and $\frac1{q}\int_0^{q}\zeta(m, n, t)dt\leq\frac1{q}\int_0^{1}\zeta(m, n, t)dt\leq=\frac{m}{nq}$, so the upper bound holds.

    For the lower bound, observe that $\zeta(m,n,q)$ is non-increasing in $q$ on $[0,1]$. We can then discuss the following three cases:

    \underline{Case 1:} $\zeta(m,n,q)\geq\frac{1}{2}$. 
    In this case, we have $\frac1{q}\int_0^{q}\zeta(m, n, t)dt\geq \zeta(m, n, q)dt\geq\frac12$.

    \underline{Case 2:} $\zeta(m,n,q)<\frac12$ and $m=1$. We can calculate $\zeta(m,n,\frac{m}{n})=(1-\frac1{n})^{n-1}\geq\frac12$, which implies that $q>\frac{m}{n}$. It holds that $\frac1{q}\int_0^{q}\zeta(m, n, t)dt\geq \frac1{q}\int_0^{\frac{m}{n}}\frac12dt=\frac{m}{2nq}$.
    
    \underline{Case 3:} $\zeta(m,n,q)<\frac12$ and $m\geq 2$. Since $\zeta(m, n, q)=\Pr[\mathrm{Binomial}(n-1,q)\leq m-1]$, by the property of binomial distribution we have that $\zeta(m, n, \frac{m-1}{n-1})\geq\frac12$, which implies that $q>\frac{m-1}{n-1}$.
    Therefore, we have $\frac1{q}\int_0^{q}\zeta(m, n, t)dt\geq \frac1{q}\int_0^{\frac{m-1}{n-1}}\frac12dt=\frac1{2q}\frac{m-1}{n-1}>\frac{m}{4nq}$.

    In summary, we have $\frac1{q}\int_0^{q}\zeta(m, n, t)dt\geq \frac14\min\{1,\frac{m}{nq}\}$, which completes the proof.
\end{proof}

\begin{lemma}\label{lem:Hn1}
    $H_{(n,1)}(q)=\Theta(\min(nq^2,\frac1{n}))$
\end{lemma}
\begin{proof}
   Firstly, for any $q\leq \frac1{2n}$, we prove that $\frac12(e^{-\frac12}-\frac12)nq^2\leq H_{(n,1)}(q)\leq \frac12nq^2$. We can calculate
   \begin{align*}
       &H_{(n,1)}(q)=\int_0^q (n-1) (1-t)^{n-2} t \, dt\\
=& -(1-t)^{n-1} t \Big|_0^q - \int_0^q \left( -(1-t)^{n-1} \right) \, dt\\
=& -(1-q)^{n-1} q - \frac{1}{n} (1-t)^n \Big|_0^q\\
=& -(1-q)^{n-1} q - \frac{1}{n} (1-q)^n + \frac{1}{n}\\
=& \frac{1}{n} \left( 1 - (1-q)^n - n(1-q)^{n-1} q \right).
   \end{align*}
    By Taylor expansion of $x^n$ at $x=1$, we have: $(1-q)^n=1-nq+\frac{n(n-1)\xi_1^{n-2}}{2}q^2$ for some $\xi_1\in[1-q,1]$.
    
   Therefore, we can calculate $1 - (1-q)^n - n(1-q)^{n-1} q
   =nq-\frac{n(n-1)\xi_1^{n-2}}{2}q^2-n(1-q)^{n-1} q
   =nq(1-\frac{n-1}{2}\xi_1^{n-2}q-(1-q)^{n-1})$. 
   
   Since $\xi_1\in[1-q,1]$, we have: 
   $$
   \begin{aligned}
       nq(1-(1-q)^{n-1}-\frac{n-1}{2}q) & \leq 1 - (1-q)^n - n(1-q)^{n-1} q\\ & \leq nq(1-(1-q)^{n-1}-\frac{n-1}{2}(1-q)^{n-2}q).
   \end{aligned}
   $$
   
   We have $1-(1-q)^{n-1}=(n-1)\xi_2^{n-2}q$ for some $\xi_2\in[1-q,1]$ by Taylor expansion of $x^{n-1}$ at $1$, so $(1-q)^{n-2}(n-1)q\leq 1-(1-q)^{n-1}\leq (n-1)q$. Then we have:
   
   $$
   \begin{aligned}
       nq((1-q)^{n-2}(n-1)q-\frac{n-1}{2}q) & \leq 1 - (1-q)^n - n(1-q)^{n-1} q \\ &\leq nq((n-1)q-\frac{n-1}{2}(1-q)^{n-2}q).
   \end{aligned}
   $$
   
   Since $q\leq\frac1{2n}$, $(1-q)^{n-2}=(1+\frac{1}{2n-1})^{-(n-2)}\geq e^{-\frac12}>\frac12$:

   For the lower bound, $nq((1-q)^{n-2}(n-1)q-\frac{n-1}{2}q)\geq nq(e^{-\frac12}-\frac12)(n-1)q\geq \frac12(e^{-\frac12}-\frac12)n^2q^2$.
   
   For the upper bound, $nq((n-1)q-\frac{n-1}{2}(1-q)^{n-2}q)\leq nq\frac12(n-1)q\leq\frac12n^2q^2$.

   So we have $\frac12(e^{-\frac12}-\frac12)n^2q^2\leq 1 - (1-q)^n - n(1-q)^{n-1} q\leq\frac12n^2q^2$, and then $\frac12(e^{-\frac12}-\frac12)nq^2\leq H_{(n,1)}(q)\leq\frac12nq^2$.

   Next, for any $q\geq\frac1{2n}$, we prove that $\frac18(e^{-\frac12}-\frac12)\frac{1}{n}\leq H_{(n,1)}(q)\leq \frac1n$.
   For the lower bound, since $H_{(n,1)}(q)$ is increasing in $q$, we have $H_{(n,1)}(q)\geq H_{(n,1)}(\frac1{2n})\geq \frac18(e^{-\frac12}-\frac12)\frac{1}{n}$.
   For the upper bound, we have $H_{(n,1)}(q)=\frac1n\left( 1 - (1-q)^n - n(1-q)^{n-1} q\right)\leq \frac1n$. We then conclude with $H_{(n,1)}(q)=\Theta(\min(nq^2,\frac1{n}))$.
\end{proof}

\begin{lemma}\label{lem:H21}
    $H_{(2,1)}(q)=\Theta(\min\{nq^2,\frac{\log(nq)+1}{n}\})$.
\end{lemma}
\begin{proof}
We can calculate
\begin{align*}
H_{(2,1)}(q) &= \int_0^q (n-1) (1-t)^{n-2} t \frac{ \frac{1}{t} \int_0^t \left( (1-x)^{n-1} + (n-1) (1-x)^{n-2} x \right) dx }{(1-t)^{n-1} + (n-1) (1-t)^{n-2} t} \, dt \\
&= \int_0^q (n-1) (1-t)^{n-2} t \frac{\Theta\left( \min\left( 1, \frac{2}{nt} \right) \right)}{(1-t)^{n-1} + (n-1) (1-t)^{n-2} t} \, dt\\
&= \int_0^q \frac{\Theta(1) \min\left( 1, \frac{1}{nt} \right) }{\frac{(1-t)}{(n-1)t}  + 1} \, dt\\
&= \Theta(1) \int_0^q \frac{\min\left(t, \frac{1}{n} \right)}{\frac{1}{n-1}(1+(n-2)t)} \, dt \\
&= \Theta(1) \int_0^q \frac{\min\left( t, \frac{1}{n} \right)}{\Theta(1)(\frac{1}{n} + t)} \, dt \\
&= \Theta(1) \int_0^q \frac{\min\left( t, \frac{1}{n} \right)}{\frac{1}{n} + t} \, dt, \\
\end{align*}
where the second equality holds is by Lemma \ref{lemma:1/qint approximation}.

For $q \leq \frac{1}{n}$, $\int_0^q \frac{\min\left( t, \frac{1}{n} \right)}{\frac{1}{n} + t} \, dt=\int_0^q\frac{t}{t+\frac1n}dt=\int_0^q\Theta(n)tdt=\Theta(nq^2)$

For $q \geq \frac{1}{n}$, $\int_{\frac1n}^q \frac{\min\left( t, \frac{1}{n} \right)}{\frac{1}{n} + t} \, dt=\int_{\frac1n}^q\frac{\frac1n}{t+\frac1n}dt=\Theta(1/n)\int_{\frac1n}{q}\frac1tdt=\Theta(\frac{\ln(qn)}{n})$, and then $H_{(2,1)}(q)=H_{(2,1)}(\frac1n)+\int_{\frac1n}^q \frac{\min\left( t, \frac{1}{n} \right)}{\frac{1}{n} + t} \, dt=\Theta(\frac{1}{n}+\frac{\log(nq)}{n})$.
\end{proof}

\subsection*{Proof of Lemma~\ref{lem:bound on n,1}}
\begin{proof}
    From the quantile representation of total effort (Lemma~\ref{lem:QuantileRep}) and analysis on $H_{(n,1)}$ (Lemma~\ref{lem:Hn1}) we have:
    \[
    \begin{aligned}
        S(n,n,1) & = n\int_0^1|v'(q)|H_{(m,1)}(q)\,dq \\
        & = n\int_0^1|v'(q)| \Theta(\min(nq^2,\frac1{n}))\, dq \\
        &= \int_0^1|v'(q)| \Theta(\min(n^2q^2,1))\, dq \\
        &= \int_{0}^{\frac{1}{n}}|v'(q)| \Theta(n^2q^2)\, dq+\int_{\frac{1}{n}}^{1}|v'(q)| \Theta(1)\, dq \\
        &= \Theta(n^2)\int_{0}^{\frac{1}{n}}|v'(q)|q^2\, dq+\Theta(1)\int_{\frac{1}{n}}^{1}|v'(q)|\, dq.
    \end{aligned}
    \]
    
    Since $|v'(q)|$ is bounded, i.e., $L'\leq|v'(q)|\leq L$, then $\int_{\frac{1}{n}}^{1}|v'(q)|\, dq = \Theta(1+n^{-1}) = \Theta(1)$. Similarly, $\int_{0}^{\frac{1}{n}}|v'(q)|q^2\, dq=\Theta(n^{-3})$, Therefore we have:
    \[
    S(n,n,1)= \Theta(n^{2})\Theta(n^{-3})+\Theta(1)\Theta(1)=\Theta(1).
    \]

    As for $S^{(1)}(2,n,1)$, by Lemma \ref{lemma:1/qint approximation} we have 
$H_{l,m}^{(1)}(t)=\Theta(1)\int_0^t\frac{\binom{n-1}{l}(1-q)^{n-l-1}q^{l}}{\sum_{j=1}^{m}\binom{n-1}{j-1}(1-q)^{n-j}q^{j-1}}\min\{1,\frac{1}{nq}\}dq$.

    Therefore, We have $H_{(n,1)}^{(1)}(t)=\Theta(1)\int_0^t(n-1)(1-q)^{n-2}q\min\{1,\frac{1}{nq}\}dq=\Theta(H_{(n,1)}(t))$. The same procedure immediately applies. Then we conclude that, $S^{(1)}(n,n,1)=\Theta(1)$ and $S(n,n,1)=\Theta(1)$, which completes the proof.
\end{proof}

\begin{lemma}\label{lem:logInt}
    \(\int_{1/n}^{1} \Theta(\log(nq)) \, dq = \Theta(\log n)\) 
\end{lemma}
\begin{proof}
    W.l.o.g. \(\Theta(\log(nq)) = C \log(nq)\), where \(C\) is a constant. Then we have:
    \[
    \int_{1/n}^{1} C \log(nq) \, dq.
    \]

This integral can be solved using change of variables. Let \(u = \log(nq)\). Then \( du = \frac{1}{q} \, dq \Rightarrow dq = q \, du.\)

When \(q = 1/n\), we have \(u = \log(n \cdot 1/n) = \log(1) = 0\). When \(q = 1\), we have \(u = \log(n \cdot 1) = \log(n)\). Substituting these into the integral, we obtain:
\[
\int_{1/n}^{1} C \log(nq) \, dq = \int_{0}^{\log(n)} C u \cdot \frac{e^u}{n} \, du = \frac{C}{n} \int_{0}^{\log(n)} u e^u \, du.
\]

Then integrate by parts. Let \(v = u\) and \(dw = e^u \, du\), so that \(dv = du\) and \(w = e^u\). Using the integration by parts formula \(\int v \, dw = vw - \int w \, dv\), we get \(\
\int u e^u \, du = u e^u - \int e^u \, du = u e^u - e^u + C.\)

Thus, the definite integral becomes:
\[
\int_{0}^{\log(n)} u e^u \, du = \left[ u e^u - e^u \right]_{0}^{\log(n)}.
\]

Evaluating the boundary terms, we have:
\[
\left( \log(n) \cdot n - n \right) - (0 - 1) = n \log(n) - n + 1.
\]

Substituting this result back, the original integral becomes:
\[
\frac{C}{n} (n \log(n) - n + 1) = C (\log(n) - 1 + \frac{1}{n}).
\]

For sufficiently large \(n\), the term \(\frac{1}{n}\) becomes negligible, leaving the dominant term as 
\(
C \log(n) - C.
\). Therefore, the asymptotic bound of the integral is:
\(
\Theta(\log(n))
\), as desired. This completes the proof.
\end{proof}

\subsection*{Proof of Lemma\ref{lem:bound on 2,1}}
\begin{proof}
    We first prove that $S(2,n,1) = \Theta(\log n).$ From the quantile representation of total effort (Lemma~\ref{lem:QuantileRep}) and analysis on $H_{(2,1)}$ (Lemma~\ref{lem:H21}) we have:
    \[
    \begin{aligned}
        S(n,n,1) & = n\int_0^1|v'(q)|H_{(2,1)}(q)\,dq \\
        & = n\int_0^1|v'(q)| \Theta(\min\{nq^2,\frac{\log(nq)+1}{n}\})\, dq \\
        &= \int_{0}^{\frac{1}{n}}|v'(q)| \Theta(n^2q^2)\, dq+\int_{\frac{1}{n}}^{1}|v'(q)| \Theta(\log(nq)+1)\, dq \\
        &= \Theta(n^2)\int_{0}^{\frac{1}{n}}|v'(q)|q^2\, dq+\int_{\frac{1}{n}}^{1}|v'(q)|\Theta(\log(nq))\, dq.
    \end{aligned}
    \]
    
    Since $|v'(q)|$ is bounded, i.e., $L'\leq|v'(q)|\leq L$, then $\int_{\frac{1}{n}}^{1}|v'(q)|\, dq = \Theta(1+n^{-1}) = \Theta(1)$. Also by Lemma~\ref{lem:logInt}, \(\int_{1/n}^{1} \Theta(\log(nq)) \, dq = \Theta(\log n)\), which gives that $\int_{\frac{1}{n}}^{1}|v'(q)|\Theta(\log(nq))\, dq=\Theta(\log n)$. Therefore we have:
    \[
    S(2,n,1)= \Theta(n^{2})\Theta(n^{-3})+\Theta(\log n)=\Theta(\log n).
    \]

    As for $S^{(1)}(2,n,1)$, by Lemma \ref{lemma:1/qint approximation} we have $H_{l,m}^{(1)}(t)=\Theta(1)\int_0^t\frac{\binom{n-1}{l}(1-q)^{n-l-1}q^{l}}{\sum_{j=1}^{m}\binom{n-1}{j-1}(1-q)^{n-j}q^{j-1}}\min\{1,\frac{1}{nq}\}dq$.

    Therefore, We have $H_{(2,1)}^{(1)}(t)=\Theta(1)\int_0^t\frac{(n-1)(1-q)^{n-2}q}{(1-q)^{n-1}+(n-1)(1-q)^{n-2}q}\min\{1,\frac{1}{nq}\}dq = \Theta(H_{(2,1)}(t))$. The same procedure immediately applies. Then we conclude that, $S^{(1)}(2,n,1)=\Theta(n)$ and $S(2,n,1)=\Theta(n)$, which completes the proof.
\end{proof}


\subsection*{Proof of Lemma~\ref{lem:bound on m,m-1}}
\begin{proof}
Recall from Lemma~\ref{lem:AsyRep} that, as $n\to \infty$:
\[
S(m,n,m-1)=F^{-1}(1-k)+n\int_0^{1-k}F^{-1}(q)\frac{q}{1-q}\frac{k}{1-k}(\frac{1}{q}-\frac{k}{q(1-q)})dq.
\]
Since by Theorem~\ref{thm:OptAsmLinear}, $k^*$ is a constant for arbitrary distributions. Then, asymptotically, $S(m^*,n,m^*-1)$ becomes a linear function of $n$, therefore $S(m^*,n,m^*-1)=\Theta(n)$, which completes the proof.
\end{proof}

\begin{proof}[Proof of Theorem~\ref{thm: 2,1 vs n,1 max effort}]
This follows directly from Lemma~\ref{lem:bound on 2,1} and Lemma~\ref{lem:bound on n,1}.
\end{proof}

\subsection*{Proof of Theorem~\ref{thm:TotalOPTVAN}}
\begin{proof}
    This follows directly from Lemma~\ref{lem:bound on n,1} and \ref{lem:bound on m,m-1} .
\end{proof}

\subsection*{Proof of Proposition~\ref{prop:TotalTWOVAN}}
\begin{proof}
    This follows directly from Lemma \ref{lem:bound on n,1} and \ref{lem:bound on 2,1}.
\end{proof}

\section{Missing Proofs in Section \ref{sec:practicalApp}}
\begin{lemma}\label{lem:IncFracSeq}
    If the following condition holds:
    \[
    \frac{a_1}{b_1} \leq \frac{a_2}{b_2} \leq \ldots \leq \frac{a_n}{b_n}.
    \]
    
    Let $A(x) = \sum_{i=1}^na_ix^{i}$, $B(x)=\sum_{i=1}^n b_ix_i$, then $A(x)/B(x)$ is increasing in $(0,+\infty)$. 
\end{lemma}

\begin{lemma}\label{lem:FracDesc}
    $H_{(m)}(q)/H_{(m')}(q)$ decreases with $q$ in $(0,1]$ for all $2 \leq m < m' \leq n$.
\end{lemma}
\begin{proof}
    We first prove that $H_{(m)}(q)/H_{(m+1)}(q)$ decreases with $q$. 

    It is suffice to prove that derivative of $H_{(m)}(q)/H_{(m+1)}(q)$ is non-positive:
    \[
    \begin{aligned}
        \frac{H'_{(m)}(q)H_{(m+1)}(q)-H_{(m)}(q)H'_{(m+1)}(q)}{(H_{(m+1)}(q))^2} & \leq 0 \\
        \frac{H_{(m+1)}(q)}{H'_{(m+1)}(q)}-\frac{H_{(m)}(q)}{H'_{(m)}(q)} & \leq 0 \\
        \frac{H_{(m+1)}(q)}{G_{(m+1)}(q)} & \leq \frac{H_{(m)}(q)}{G_{(m)}(q)},
    \end{aligned}
    \]or equivalently, $\frac{H_{(m)}(q)}{G_{(m)}(q)}$ decreases with $m$ for any $q\in (0,1]$.

    We have the expression of $\frac{H_{(m)}(q)}{G_{(m)}(q)}$:
    \begin{multline*}
        \frac{\sum_{j=1}^m\binom{n-1}{j-1}(1-q)^{n-j}q^{j-1}}{\binom{n-1}{m-1}(1-q)^{n-m}q^{m-2}\int_0^q\sum_{j=1}^m\binom{n-1}{j-1}(1-t)^{n-j}t^{j-1}\, dt} \\ \cdot \int_0^q \frac{\binom{n-1}{m-1}(1-x)^{n-m}x^{m-2}\int_0^x\sum_{j=1}^m\binom{n-1}{j-1}(1-t)^{n-j}t^{j-1}\, dt}{\sum_{j=1}^m\binom{n-1}{j-1}(1-x)^{n-j}x^{j-1}}\, dx
    \end{multline*}

    We proceed by prove that $\frac{H_{(m)}(q)}{D_{(m)}(q)} / \frac{H_{(m+1)}(q)}{D_{(m+1)}(q)} \geq 1$. Since the outer integration is hard to handle, we skip it by proving a stronger version of that claim: the inequality holds point-wise for $x$. 

    With some cancellation, this fraction can be simplified to \(D(q)/D(x)\), where:
    \[
        D(q) := \frac{q}{1-q} \frac{\sum_{j=1}^m\binom{n-1}{j-1}(1-q)^{n-j}q^{j-1}}{\int_0^{q}\sum_{j=1}^{m}\binom{n-1}{j-1}(1-t)^{n-j}t^{j-1}\, dt} \frac{\int_0^{q}\sum_{j=1}^{m+1}\binom{n-1}{j-1}(1-t)^{n-j}t^{j-1}\, dt}{\sum_{j=1}^{m+1}\binom{n-1}{j-1}(1-q)^{n-j}q^{j-1}},
    \]
    then it is suffice to prove that $D(q)$ increase with $q$, since $x\leq q$. 

    To simplify expression, we further introduce several notations:
    \[
    \begin{aligned}
        A(m,q) & =\sum_{j=1}^m\binom{n-1}{j-1}(1-q)^{n-j}q^{j-1} = 1-(n-1)\binom{n-2}{m-1} \int_0^{q}(1-t)^{n-m-1}t^{m-1}\,dt, \\
        A'(m,q) & = -(n-1)\binom{n-2}{m-1}(1-q)^{n-m-1}q^{m-1},
    \end{aligned}
    \]$I_A(n,q)=  \int_0^qA(m,t) \, dt$, $C(m,q) = A(m,q)/I_A(m,q)$, $P(q) = \frac{q}{1-q}$, $P'(q)= \frac{1-q+q}{(1-q)^2} = (1-q)^{-2}$ and,
    \[
    C'(m,q) = \frac{A'(m,q)I_A(m,q)-A^2(m,q)}{I^2_A(m,q)}.
    \]

    Then $D(q)=P(q)C(m,q)C^{-1}(m+1,q)$, and the derivative:
    \[
    \begin{aligned}
        D'(q) = & P'(q)C(m,q)C^{-1}(m+1,q)+P(q)C'(m,q)C^{-1}(m+1,q)\\
                & -P(q)C(m,q)C'(m+1,q)C^{-2}(m+1,q).
    \end{aligned}
    \]

    Next, canceling out $C^{-1}(m+1,q)$:
    \[
    P'(q)C(m,q) + P(q)C'(m,q)-P(q)C(m,q)C'(m+1,q)C^{-1}(m+1,q).
    \]

    Substituting the expression for $P$ and $C$, we have:
    \begin{multline*}
        \frac{A(m,q)}{(1-q)^2I_A(m,q)} + \frac{q(A'(m,q)I_A(m,q)-A^2(m,q))}{(1-q)I^2_A(m,q)} \\
        -\frac{qA(m,q)(A'(m+1,q)I_A(m+1,q)-A^2(m+1,q))I_A(m+1,q)}{(1-q)I_A(m,q)I^2_A(m+1,q)A(m+1,q)}
    \end{multline*}

    Canceling out $(1-q)I_A(m,q)$, it becomes:
    \begin{multline*}
        \frac{A(m,q)}{1-q}+\frac{q(A'(m,q)I_A(m,q)-A^2(m,q))}{I_A(m,q)} \\
        - \frac{qA(m,q)(A'(m+1,q)I_A(m+1,q)-A^2(m+1,q))}{I_A(m+1,q)A(m+1,q)}
    \end{multline*}

    Separate the numerators, and move $q$ to denominator, then we get:
    \[
    \frac{A(m,q)}{q(1-q)} + A'(m,q)-\frac{A^2(m,q)}{I_A(m,q)}-\frac{A(m,q)A'(m+1,q)}{A(m+1,q)}+\frac{A(m,q)A(m+1,q)}{I_A(m+1,q)}
    \]

    The third term and the last term can be combined:
    \[
    A(m,q)\left( \frac{A(m+1,q)}{I_A(m+1,q)}-\frac{A(m,q)}{I_A(m,q)}\right) \geq 0,
    \]where $I_A(m,q)/A(m,q)$ decrease with $m$ follows from the proof of Theorem~\ref{thm:ConpleteSimpleContest}.  

    Now, it is suffice to prove the remaining part also non-negative, i.e.:
    \[
    \frac{A(m,q)}{q(1-q)} + A'(m,q) - \frac{A(m,q)A'(m+1,q)}{A(m+1,q)} \geq 0.
    \]

    Re-organizing terms, it is equivalent to show that:
    \[
    \frac{A'(m,q)}{A(m,q)} - \frac{A'(m+1,q)}{A(m+1,q)} + \frac{1}{q(1-q)} \geq 0
    \]

    Since $[\ln q - \ln (1-q)]' = q^{-1}+(1-q)^{-1} = [q(1-q)]^{-1}$, it further becomes:
    \[
    \begin{aligned}
        & [\ln A(m,q)-\ln A(m+1,q)+\ln q-\ln(1-q)]' \\
        = & \left [ \ln \frac{qA(m,q)}{(1-q)A(m+1,q)}\right]'
    \end{aligned}
    \]

    Then the problem finally becomes to show that $\frac{qA(m,q)}{(1-q)A(m+1,q)}$ increase with $q$. 

    Expanding the expressions, we have:
    \[
    \begin{aligned}
        & \frac{\sum_{j=1}^m\binom{n-1}{j-1}(1-q)^{n-j}q^{j-1}}{\sum_{j=1}^{m+1}\binom{n-1}{j-1}(1-q)^{n-j}q^{j-1}} \frac{q}{1-q} = \frac{\sum_{j=1}^m\binom{n-1}{j-1}(1-q)^{n-j-1}q^{j}}{\sum_{j=1}^{m+1}\binom{n-1}{j-1}(1-q)^{n-j}q^{j-1}} \\
    =& \frac{\sum_{j=2}^{m+1}\binom{n-1}{j-2}(1-q)^{n-j}q^{j-1}}{\sum_{j=1}^{m+1}\binom{n-1}{j-1}(1-q)^{n-j}q^{j-1}} = \frac{\sum_{j=2}^{m+1}\frac{j-1}{m-j+1}\binom{n-1}{j-1}(1-q)^{n-j}q^{j-1}}{\sum_{j=1}^{m+1}\binom{n-1}{j-1}(1-q)^{n-j}q^{j-1}} \\
    = & \frac{\sum_{j=1}^{m+1}\frac{j-1}{n-j+1}\binom{n-1}{j-1}(1-q)^{n-j}q^{j-1}}{\sum_{j=1}^{m+1}\binom{n-1}{j-1}(1-q)^{n-j}q^{j-1}} = \frac{\sum_{j=1}^{m+1}\frac{j-1}{n-j+1}\binom{n-1}{j-1}(1-q)^{1-j}q^{j-1}}{\sum_{j=1}^{m+1}\binom{n-1}{j-1}(1-q)^{1-j}q^{j-1}}.
    \end{aligned}
    \]

    Let $t=\frac{q}{1-q}$, it becomes:
    \[
    \frac{\sum_{j=1}^{m+1}\frac{j-1}{n-j+1}\binom{n-1}{j-1}t^{j-1}}{\sum_{j=1}^{m+1}\binom{n-1}{j-1}t^{j-1}}.
    \]

    Since $\frac{j-1}{n-j+1}$ increases with $j$, Lemma~\ref{lem:IncFracSeq} shows this fraction increases with $t$ (also $q$), as desired.

    So far, We have proved that $H_{(m)}(q)/H_{(m+1)}(q)$ decreases with $q$. 
    
    For any two function $f(x), g(x) \geq 0$, if $f'(x), g'(x) \leq 0 $, then $[f(x)g(x)]' = f'(x)g(x)+f(x)g'(x) \leq0$. Since non-negative function $H_{(m)}(q)/H_{(m+1)}(q)$ is differentiable and decreasing, we have $[H_{(m)}(q)/H_{(m+2)}(q)]'=[(H_{(m)}(q)/H_{(m+1)}(q))\cdot (H_{(m+1)}(q)/H_{(m+2)}(q))]'\leq0$. Therefore, it can be prove by induction that $H_{(m)}(q)/H_{(m')}(q)$ for any $m'>m$, which completes the proof.
\end{proof}

\begin{proposition}\label{prop:SupM}
    For the total effort objective, given $n$, there exists an $O(n)$ algorithm that finds the largest shortlist capacity $m$ such that the complete simple contest $S(m,n,m-1)$ is the optimal contest under some ability distribution $F(x)$.    
\end{proposition}
\begin{proof}
    Since Theorem~\ref{thm:ConpleteSimpleContest} states that optimal contest is a complete simple contest, i.e., $l = m-1$, we therefore omit $l$ in the following discussion.

Recall the quantile representation of total effort:
\[
    S(m,n, l)= n\int_0^1|v'(q)|\int_0^qG_{(m,l)}(t)\,dt\,dq,
    \]

Then total effort becomes the integration of the multiplication of a function $|v'(q)|$ determined by ability distribution, and a function $H_{(m,l)}(q)=\int_0^qG_{(m,l)}(t)\,dt$ that is completely decided by the contest structure.

Let us focus on the distribution-free part. We can plot $H_{m}(q)$ as a function of $q\in[0,1]$ for $m = 2,\ldots,n$. The example of $n=10$ is shown in Figure~\ref{fig:universal-b}. In this case, we can see that for some $m$ (e.g., $m=3$), $H_{(m)}(q)>H_{(m+1)}(q)$ holds point-wise, thus, $S(m,n) > S(m+1,n)$ stands true for arbitrary distributions, indicating that $m+1$ is a strictly dominated choice. 

\begin{figure}[h]

\begin{subfigure}[ht]{0.48\textwidth}
    \centering
    \includegraphics[width=\textwidth]{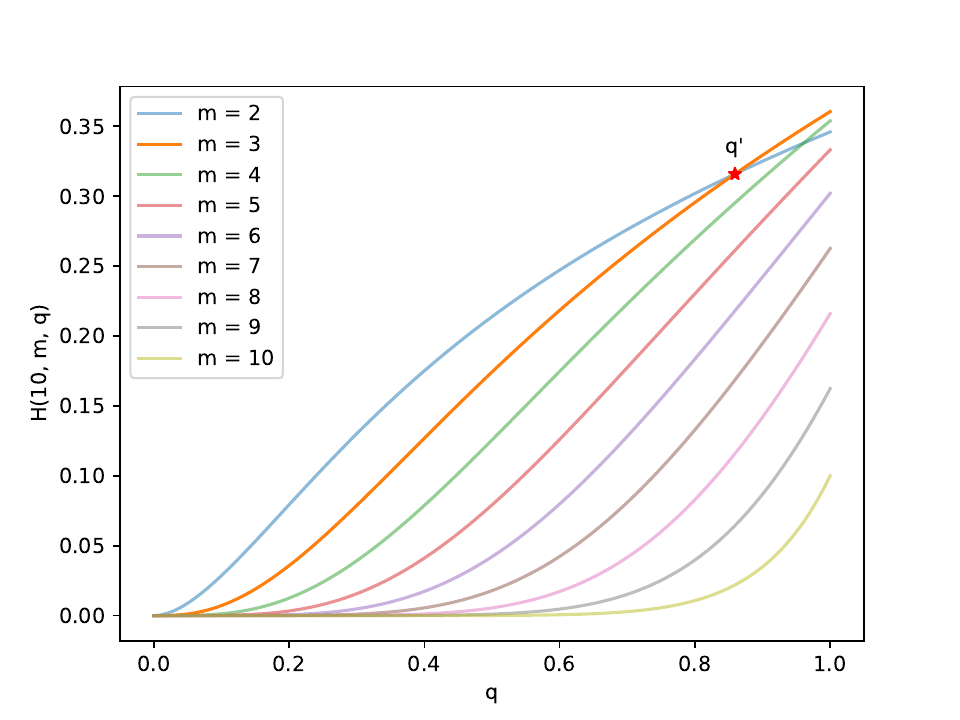}
    \subcaption{$H_{(m)}(q)$ for different $m$}
    \label{fig:universal-b}
    \end{subfigure}
\begin{subfigure}[ht]{0.48\textwidth}
    \centering
    \includegraphics[width=\textwidth]{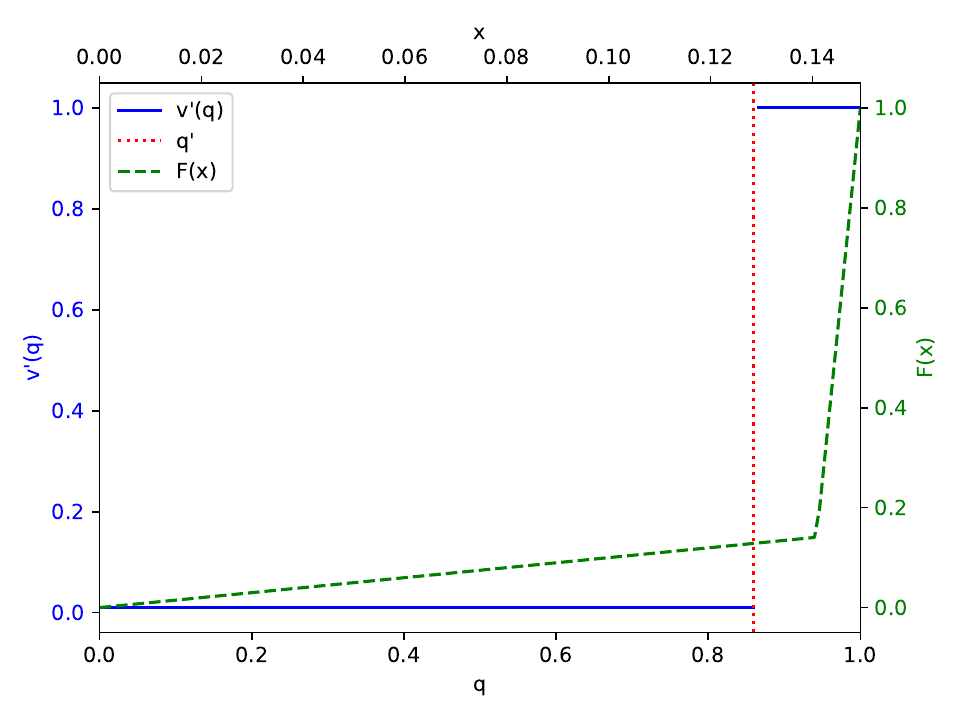}
    \subcaption{$|v'(q)|$ s.t. $S(3,10)>S(2,10)$}
    \label{fig:universal-c}
    \end{subfigure}
\label{fig:universal}
\end{figure}

Actually, it can be shown that $H_{(m)}(q)/H_{(m')}(q)$ is decreasing with $q$ for $m'>m$ (Lemma~\ref{lem:FracDesc}), then $H_{(m)}(1) > H_{(m')}(1)$ is a suffice and necessary condition for $H_{(m)}(q)>H_{(m')}(q)$ point-wise. On the other hand, if $H_{(m)}(1) < H_{(m')}(1)$, then there exists unique $q' \in(0,1)$ such that $H_{(m)}(q')=H_{(m')}(q')$ and $H_{(m)}(q)<H_{(m')}(q)$ afterwards (e.g., the $q'$ for $m=2$ and $m'=3$ is marked with asterisk in Figure~\ref{fig:universal-b}). Since $v'(q)$ can be any positive function, we can always construct a distribution that satisfies $v'(q)=1$ when $q\geq q'$ and $v'(q)=\epsilon$ elsewhere such that $S(m,n) < S(m',n)$ (See an example distribution that make $m'=3$ better than $m=2$ in Figure~\ref{fig:universal-c}, where we let $q' \approx 0.859$, $\epsilon=0.01$, and $S(2,10) \approx 0.481 < 0.489 \approx S(3,10)$.).

Therefore, we can find the $m$ that maximize $H_{m}(1)$. For $m'>m$, we have $H_{(m)}(q) > H_{(m')}(q)$ point-wise, so the optimal capacity can not be more than $m$. For $m'<m$, we have $H_{(m')}(1) < H_{(m)}(1)$, then we can still construct a distribution that satisfies $v'(q)=1$ when $q \geq \max\{\vec{q'}\}$ and $v'(q) = \epsilon$ elsewhere such that $S(m,n) > S(m',n)$ for all $m'<m$, hence we find an instance making $m$ the optimal capacity. We then conclude that $m$ is the tight upper bound for optimal capacity for given $n$, as desired. Since we enumerate over the value of $m$, the algorithm is $O(n)$. 
\end{proof}
\begin{remark}
    The insight from the construction of worst case distribution (e.g., Figure~\ref{fig:universal-c}) is, when almost all of the population are concentrated near the strongest end of ability, it tends to need larger shortlist capacity to reach optimality. On the other hand, if highest ability only takes up a little probability mass, or equivalently, $|v'(q)|$ is much higher when $q$ is small, it tends to obtain optimality with fewer contestants. An uniform distribution, i.e., $|v'(q)|=1$, whose probability mass is evenly distributed, is right in the middle, with $k^*\approx15\%$, as shown in Example~\ref{exam:OptimalUniform}. 
\end{remark}

\begin{corollary}\label{coro:shortlistAlways}
    All complete simple contests that admit $kn$ contestants ($k\in (0,1)$) will result in $\Theta(n)$ total effort. Specifically, when $n$ is large, letting admission ratio $k=31.62\%$ produces higher total effort than any $k'>k$, where $k$ is the solution of $(2-k)(k-1)-\ln k=0$.
\end{corollary}
\begin{proof}
    Recall from Lemma~\ref{lem:AsyRep} that, as $n\to \infty$:
\[
S(m,n,m-1)=F^{-1}(1-k)+n\int_0^{1-k}F^{-1}(q)\frac{q}{1-q}\frac{k}{1-k}(\frac{1}{q}-\frac{k}{q(1-q)})dq,\]
and the convergence rate is independent of the choice of $k$, therefore for any selected $k$, total effort is a linear function of $n$ asymptotically, then $S(kn,n,kn-1)=\Theta(n)$, as desired. 

Let $\phi(k):=\lim_{n\to \infty}\frac{S(m,n)}{n}$. It is shown in the proof of Theorem~\ref{thm:UniversalBound} that when $k'>k$, $\phi'(k) < 0$, then $S(n,m)$ is decreasing in $[k,k']$, therefore $S(n,km) > S(n,k'm)$ when $n$ is sufficiently large, which completes the proof.
\end{proof}

\end{document}